\documentclass[a4paper, 12pt, english]{article}
\usepackage[margin = 1.5cm]{geometry}
\usepackage{babel}
\usepackage{amsfonts}
\usepackage{amsmath}
\usepackage{amsthm}
\usepackage{amssymb}
\usepackage{mathtools}
\usepackage{accents}
\usepackage{xargs}
\usepackage{xifthen}
\usepackage{hyperref}
\usepackage{csquotes}
\usepackage[dvipsnames]{xcolor}
\usepackage{dsfont}
\usepackage[shortlabels]{enumitem}
\usepackage{comment}

\usepackage{apptools}
\AtAppendix{\counterwithin{lemma}{section}}
\AtAppendix{\counterwithin{proposition}{section}}
\AtAppendix{\counterwithin{corollary}{section}}
\AtAppendix{\counterwithin{theorem}{section}}
\AtAppendix{\counterwithin{equation}{section}}

\allowdisplaybreaks

\usepackage[backend = bibtex, citestyle = authoryear, style = authoryear, maxbibnames = 10, maxcitenames = 2, dashed = false, url = false, date = year, doi = false, isbn = false]{biblatex}
\AtEveryBibitem{\clearfield{note}}

\title{Understanding the worst-kept secret of\\ high-frequency trading\footnote{This research benefited from the financial support of the chairs ``Deep finance \& Statistics'' and ``Machine Learning \& systematic methods in finance'' of \'Ecole polytechnique. The authors thank Ali Baouan for valuable comments.}}
\author{Sergio Pulido\thanks{Universit\'e Paris-Saclay, CNRS, ENSIIE, Univ Evry, Laboratoire de Math\'ematiques et Mod\'elisation d'Evry (LaMME). \textbf{Email:} sergio.pulidonino@ensiie.fr}  \and Mathieu Rosenbaum\thanks{Centre de Mathématiques Appliquées (CMAP), CNRS, École polytechnique, Institut Polytechnique de
Paris. \textbf{Email:} mathieu.rosenbaum@polytechnique.edu} \and Emmanouil Sfendourakis\thanks{Centre de Mathématiques Appliquées (CMAP), CNRS, École polytechnique, Institut Polytechnique de
Paris. \textbf{Email:} emmanouil.sfendourakis@polytechnique.edu}}

\newcommand{\qmax}{\bar{q}}

\newcommand{\defeq}{\vcentcolon=}
\newcommand{\diff}{\mathrm{d}}
\newcommand{\integerInterval}[2][1]{\lbrace #1,\dots, #2 \rbrace}
\newcommandx{\PnL}[1][1={}]{\ifthenelse{\isempty{#1}}{\mathrm{PnL}^{t,Q,y,q}}{\mathrm{PnL}^{#1}}}
\newcommandx{\Proba}[2][1={}, 2={}]{\ifthenelse{\isempty{#1}}{\mathbb{P}}{\mathbb{P}^{#1}}\ifthenelse{\isempty{#2}}{}{\left(#2\right)}}
\newcommandx{\E}[2][1={}, 2={}]{\ifthenelse{\isempty{#1}}{\mathbb{E}}{\mathbb{E}^{#1}}\ifthenelse{\isempty{#2}}{}{\left[#2\right]}}
\newcommandx{\X}[1][1={}]{\ifthenelse{\isempty{#1}}{X^{t,Q,x,y,p,q}}{X^{#1}}}
\newcommandx{\Q}[1][1={}]{\ifthenelse{\isempty{#1}}{Q^{t,Q,q}}{Q^{#1}}}
\newcommandx{\Pm}[1][1={}]{\ifthenelse{\isempty{#1}}{P^{t,y,p}}{P^{#1}}}
\newcommandx{\Pe}[1][1={}]{\ifthenelse{\isempty{#1}}{S^{t,y,p}}{S^{#1}}}
\newcommandx{\Y}[1][1={}]{\ifthenelse{\isempty{#1}}{Y^{t,y}}{Y^{#1}}}

\newcommand{\Qa}{\mathcal{Q}^{+,a}_{n}}
\newcommand{\Qb}{\mathcal{Q}^{+,b}_{n}}
\newcommand{\Qi}{\mathcal{Q}^{+,i}_{n}}
\newcommand{\qa}{\bar{q}^a}
\newcommand{\qb}{\bar{q}^b}
\newcommand{\qi}{\bar{q}^i}
\newcommand{\Qmax}{\bar{Q}}

\newcommand{\ybar}{\bar{y}}
\newcommand{\yplus}{y_{+}}
\newcommand{\yminus}{y_{-}}

\newcommand{\D}{\mathcal{D}}

\DeclareMathOperator*{\argmax}{arg\,max}

\newenvironment*{dummyenv}{}{}

\newtheorem{lemma}{Lemma}
\newtheorem{proposition}{Proposition}
\newtheorem{corollary}{Corollary}
\newtheorem{theorem}{Theorem}

\theoremstyle{definition}
\newtheorem{definition}{Definition}
\theoremstyle{remark}
\newtheorem{rem}{Remark}

\begin{document}

\maketitle

\begin{abstract}
    Volume imbalance in a limit order book is often considered as a reliable indicator for predicting future price moves. In this work, we seek to analyse the nuances of the relationship between prices and volume imbalance. To this end, we study a market-making problem which allows us to view the imbalance as an optimal response to price moves. In our model, there is an underlying efficient price driving the mid-price, which follows the model with uncertainty zones. A single market maker knows the underlying efficient price and consequently the probability of a mid-price jump in the future. She controls the volumes she quotes at the best bid and ask prices. Solving her optimization problem allows us to understand endogenously the price-imbalance connection and to confirm in particular that it is optimal to quote a predictive imbalance. Our model can also be used by a platform to select a suitable tick size, which is known to be a crucial topic in financial regulation. The value function of the market maker's control problem can be viewed as a family of functions, indexed by the level of the market maker's inventory, solving a coupled system of PDEs. We show existence and uniqueness of classical solutions to this coupled system of equations. In the case of a continuous inventory, we also prove uniqueness of the market maker's optimal control policy.
    \\[2ex] 
    \noindent{\textbf{Keywords:} market microstructure, volume imbalance, high-frequency market-making, optimal tick size, stochastic optimal control, Hamilton–Jacobi–Bellman equation, classical solutions.}
\end{abstract}

\section{Introduction}

Volume imbalance's high predictive power of mid-price moves is the \textit{worst-kept secret of high-frequency trading}\footnote{Quoting Sasha Stoikov in a research seminar in 2014.}. In a limit order book, the volume imbalance is defined as
\begin{equation*}
    I = \frac{q^b - q^a}{q^b + q^a}
\end{equation*}
where $q^b$ is the quantity of an asset posted on the best bid price, and $q^a$ is the quantity posted on the best ask price. Recent empirical studies have confirmed the predictive nature of the imbalance: when it is close to 1, the mid-price is likely to jump upwards, and when it is close to -1, the mid-price is likely to jump downwards \parencite{huang_simulating_2015,lehalle_limit_2017,sfendourakis_lob_2023}.

Given the importance of the volume imbalance to explain price moves, practitioners commonly use volume imbalance to estimate an \enquote{efficient} price, representing the belief of all the market participants about the value of the asset. One simple way to estimate an efficient price is to consider the weighted mid-price $P^w \defeq \frac{I+1}{2}P^a + \frac{1-I}{2}P^b$, where $P^a$ and $P^b$ denote the ask and bid price, respectively. The weighted mid-price $P^w$ is closer to the ask price when the imbalance is higher, reproducing the empirical intuition. \textcite{stoikov_micro-price_2018} introduces the more sophisticated notion of \enquote{micro-price}, which is defined in terms of observable quantities among which the volume imbalance plays an important role. The micro-price provides a nice measure of the efficient price because it is a martingale and it is generally less noisy than the weighted mid-price. See also \parencite{delattre2013estimating} for the estimation of the efficient price from the order flow. 

In this work, we propose to take the opposite viewpoint, namely to understand the volume imbalance as a response to prices. We consider a market maker who is the only liquidity provider of the market and who knows an efficient price driving the mid-price. The limit orders of the market maker are executed by market takers. The aggregation of market takers is also well-informed about the efficient price: market orders arrive at a rate depending on the efficient price. The market maker quotes volumes at the best bid and ask prices. The posted volumes do not influence the price moves nor the arrival of market orders. We then investigate the character of the optimal posted volumes and, in particular, whether these volumes generate a predictive imbalance in the limit order book. Our model explains, in an endogenous fashion, the subtle price-imbalance connection.

The market-making literature in continuous time is rich. In the pioneering work of \textcite{avellaneda_high-frequency_2008}, which was studied in more mathematical detail by \textcite{gueant_dealing_2013}, the market maker quotes a fixed quantity of the asset on the bid and the ask side and, contrary to our approach, controls the price of these quotes. Their model, well-suited for over-the-counter markets, has been enriched in various ways: allowing to quote different volumes \parencite{bergault_size_2021,barzykin_algorithmic_2023}, or considering principal-agent problems related to the optimal make-take fees of trading platforms \parencite{el_euch_optimal_2021,baldacci2021optimal,baldacci_market_2023}. Models where the market maker controls the price of her quotes are less suited to large-tick assets where most of the trading activity occurs at a half-tick distance from the mid-price.

Models considering a fixed tick size are adapted to the limit order book framework. This is the point of view taken in \parencite{guilbaud_optimal_2013,guilbaud_optimal_2015,fodra_high_2015} who use a Markovian model for the limit order book. The best limits are observed, and the market maker can quote at them, or in the spread. More recently, \textcite{abergel_algorithmic_2020} propose a model where the whole (up to some depth) order book is modeled, and the market maker can have quotes pending anywhere. In these models, however, there is no efficient price driving the mid-price fluctuations.

In \parencite{baldacci_bid_2020}, two tick grids are considered -- one for the bid, and one for the ask. There is an underlying efficient price observed by the market maker. She can only quote at the considered \enquote{fair} bid and ask prices, lying on their respective tick grid, and evolving according to the model with uncertainty zones, introduced by \textcite{robert_new_2011}: when the efficient price hits a defined barrier, the fair bid or ask price jumps. In this work, however, the market maker has no control on the posted volumes.

We adopt a framework close to the one of \parencite{baldacci_bid_2020}: there is an underlying efficient price, following a Bachelier model, and a fair mid-price that jumps up if the efficient price hits the upper barrier and jumps down if it hits the lower barrier. A large-tick stock is considered, hence the market maker is only allowed to quote at half-tick distance from the mid-price. She controls the volume that she posts on each side. The market orders are represented by marked Poisson processes, where the mark is their volume, and they depend on the distance between the efficient price and the mid-price. Consequently, market orders are not controlled by the market maker, which rules out price manipulation in our model.

The market maker's goal is to maximize the exponential utility of her terminal gains. A Hamilton-Jacobi-Bellman system of PDEs with nonlocal boundary conditions is associated to this control problem. Contrary to the majority of the above mentioned studies on optimal market-making, where viscosity solutions are considered, we prove, using PDE techniques, the existence of a classical solution to this system. Using a verification argument, we show that this solution is equal to the value function, and we give a characterization of the optimal controls. In addition, and remarkably, we are able to derive the uniqueness of the optimal control policy. Numerical approximations of the HJB equation -- and consequently of the optimal posted volumes -- elucidate the connection between prices and volume imbalance. Our numerical experiments indicate that for reasonable choices of the model parameters, the optimal posted volume imbalance is predictive of price moves. More precisely, under reasonable scenarios, the optimal volume imbalance is a monotone function of the distance between the efficient and mid-prices, which in turn determine the price moves via the model with uncertainty zones. 

The paper is organized as follows. In Section \ref{sec:notation}, we introduce the mathematical notation used throughout the paper. In Section \ref{sec:control_problem}, we present the market dynamics and the market maker's control problem. In Section \ref{sec:results}, we state the main existence and uniqueness results of the control problem. Numerical illustrations are given in Section \ref{sec:numerics}. In particular, we present an interesting application of our framework to determine optimal tick sizes. We prove the verification theorem in Section \ref{sec:verification}. In Section \ref{sec:hjb}, we prove the existence of a classical solution to the HJB equation. Section \ref{sec:uniqueness} contains the proof of the uniqueness of the control policy. Some additional proofs are relegated to the appendices.

\section{Spaces and notations}\label{sec:notation}

We use many notations from \parencite{friedman_partial_1983} for the domains and norms. All the functions mentioned in this section are real-valued.

For a subset $D$ of $\mathbb{R}^n$, we denote by $\partial D$ its boundary. Let $T> 0$ and $a,b \in \mathbb{R}$, with $a<b$. For $P, Q \in [0,T] \times [a,b]$, $P = (t,y)$, $Q=(t',y')$, we denote $d(P,Q) \defeq \sqrt{|y-y'|^2 + |t-t'|}$ the parabolic distance between $P$ and $Q$. Let $D \subset [0,T] \times [a,b]$. For $P=(t,y) \in D$, we define
\begin{equation*}
    d^D_P \defeq \min\left\{d(P,Q): Q=(t',y') \in \partial D, t' \geqslant t \right\}
\end{equation*}
the distance from the \enquote{parabolic boundary} (here we consider parabolic PDEs with terminal condition and not initial one, as it suits our problem better). For $P, Q \in [0,T] \times [a,b]$, we set $d^D_{PQ} \defeq \min(d_P, d_Q)$.
When $D$ is the whole space $[0,T] \times [a,b]$, we omit the superscript $D$.

For a function $u$ defined on a space $X$, we write $|u|_{\infty} \defeq \sup_X |u|$. We do not specify $X$ when it is clear from the context. Let $C([0,T]\times [a,b])$ be the space of continuous functions on $[0,T] \times [a,b]$ and $C^{1,2}([0,T) \times (a,b))$ the functions $u$ on $[0,T) \times (a,b)$ that are one time differentiable with respect to the first variable (we write $\partial_t u$ for this derivative) and two times differentiable with respect to the second variable (we denote by $\partial_y u$ and $\partial^2_{yy}u$ these derivatives) and such that $\partial_t u$, $\partial_y u$ and $\partial^2_{yy}u$ are continuous on $[0,T) \times (a,b)$.

Let $k \geqslant 0$ and let $u$ be a function defined on $[0,T) \times (a,b)$. We define
\begin{equation*}
    |d^k u|_{\infty} \defeq \sup_{P \in [0,T) \times (a,b)} \left|d_P^k u(P) \right|.
\end{equation*}
For $\delta \in (0,1)$, we write
\begin{equation*}
    |u|_{\delta} \defeq |u|_{\infty} + \sup_{\substack{P,Q \in [0,T) \times (a,b) \\ P \neq Q}}d^{\delta}_{PQ}\frac{|u(P) - u(Q)|}{d(P,Q)^{\delta}}.
\end{equation*}
We write $u \in C^{\delta}$ if and only if $|u|_{\delta} < \infty$. For $k \geqslant 0$, we define
\begin{equation*}
    |d^ku|_{\delta} \defeq |d^ku|_{\infty} + \sup_{\substack{P,Q \in [0,T) \times (a,b) \\ P \neq Q}}d^{k+\delta}_{PQ}\frac{|u(P) - u(Q)|}{d(P,Q)^{\delta}}.
\end{equation*}
Finally, for $u \in C^{1,2}([0,T) \times (a,b))$, we define
\begin{equation*}
    |u|_{2+\delta} \defeq |u|_{\delta} + |d^2 \partial_t u|_{\delta} + |d\partial_y u|_{\delta} + |d^2 \partial^2_{yy} u|_{\delta},
\end{equation*}
and we write $u \in C^{2+\delta}$ if and only if  $|u|_{2+\delta} < \infty$.
We remark that for all $k \geqslant 0$, there exists a constant $C_k$ depending only on the domain such that $|d^k u|_{\delta} \leqslant C_k|u|_{\delta}$.

For a subset $D$ of $[0,T) \times (a,b)$, we define the \enquote{classical} Hölder seminorm of $u$ on $D$ by
\begin{equation*}
    H^{\delta}_D (u) \defeq \sup_{\substack{P,Q \in D \\ P \neq Q}} \frac{|u(P) - u(Q)|}{d(P,Q)^{\delta}}.
\end{equation*}
Let $K$ be a compact included in $[0,T) \times (a,b)$. It is easy to see that there exists a constant $C > 0$ (depending on $K$) such that, for $u$ defined on $[0,T) \times (a,b)$,
\begin{equation*}
    H^{\delta}_K (u) \leqslant C |u|_{\delta}.
\end{equation*}
We also have a constant $C' > 0$ (depending on $K$) such that, for $u \in C^{1,2}([0,T) \times (a,b))$,
\begin{equation*}
    H^{\delta}_K (u) + \sup_K|\partial_t u| + H^{\delta}_K (\partial_t u)+ \sup_K|\partial_y u| + H^{\delta}_K (\partial_y u) + \sup_K|\partial^2_{yy} u| + H^{\delta}_K (\partial^2_{yy} u) \leqslant C' |u|_{2+ \delta}.
\end{equation*}

For a finite measure $\mu$ on an interval $I$ of $\mathbb{R}$ with finite support (i.e. $\mu = \sum_{i \in J}a_i\delta_{i}$ where $J$ is a finite subset of $I$ and $\delta_i$ is the Dirac measure at $i$) and a function $f$ defined on $I$, we write$\int_{I} f(z) \mu(\diff z) = \sum_{i \in J}a_i f(i)$ without supposing any measurability condition on $f$. In general, \enquote{measurability} on an interval on $\mathbb{R}$ shall always be understood in the Borel sense.

We denote by $\mathbb{N}$ the set of natural integers containing 0, and $\mathbb{N}^* \defeq \mathbb{N} \setminus \{0\}$.
We use the notations $x \wedge y \defeq \min\{x,y\}$ and $x \vee y \defeq \max\{x,y\}$. For a subset $A\subset \mathbb{R}$,  $\bar{A}$ is the closure of $A$.
Whenever we take an expectation with respect to a probability measure $\mathbb{Q}$, we write $\mathbb{E}^{\mathbb{Q}}$. If $\mathbb{Q}=\mathbb{P}$, we omit the superscript in the expectation.

\section{Formulation of the control problem}
\label{sec:control_problem}

We consider an optimal market-making problem on a large-tick stock whose mid-price is not controlled by the market maker. The mid-price lies on the inter-tick grid and  follows a variant of the model with uncertainty zones as in \parencite{robert_new_2011}. The (unique) market maker can only quote volumes at the best bid and ask prices which are located half-tick distance from the mid-price. The admissible quotes keep her inventory bounded over a finite interval. She has information about the probability of future mid-price moves due to the fact that she observes the efficient price governing the mid-price dynamics. The knowledge of the efficient price can be interpreted as a signal for the market maker. Market orders arrive on the bid and ask sides with an intensity that depends on the distance between the efficient price and the mid-price, and consequently they cannot be influenced or manipulated by the actions of the market maker. The objective of the market maker is to maximize the expected utility of the profits over a finite time horizon. We consider and exponential utility function for the market maker. We now describe in detail the notation and the specific mathematical setting of this market-making control problem.
\begin{rem}
    \label{rem:efficient_price_interpretation}
    Following a large part of the economic literature on the topic, we adopt a price discovery approach (in contrast to price formation) where the fair price is determined and discovered by market participants through the trading process. In practice, the efficient price is not perfectly observable, and market participants use limit order book information, order flow, and transaction data to estimate it, see \parencite{delattre2013estimating, robert_new_2011, stoikov_micro-price_2018}. Another source of information can be a primary market when the market maker has to quote on a secondary market, as is the case for FX or fixed income. 
    %Usually, a proxy such as the weighted mid-price, computed with the volume imbalance, is used. Here, the market-maker could use signals derived from volumes posted on a primary exchange to derive the efficient price. She then computes her optimal quotes, which she posts in the \enquote{secondary} market.
\end{rem}

\subsection{The market}

Let $T \in (0,\infty)$ be the trading horizon. The market maker's inventory is bounded in absolute value by $\Qmax$. We consider two different frameworks depending on whether the market maker can quote volumes in a discrete or continuous fashion. In the first case, we fix $n\in\mathbb{N}^*$, and consider traded volumes which are multiples of $\frac{\Qmax}{n}$. For the continuous case, labeled with $n=\infty$, volumes can take values over a continuous interval. The bid and ask quoted volumes should be bounded by $\qa$ and $\qb$, respectively. Necessarily, $\qa,\qb\leqslant 2\Qmax$ and $(\qa, \qb) \in \left(\frac{\Qmax}{n}\mathbb{N}^*\right)^2$ (if $n < \infty$), $(\qa, \qb) \in (0,\infty)^2$ (if $n = \infty$).

%$n \in 2 \mathbb{N}^* \cup \{\infty\}$, $\Qmax \in (0,\infty)$ and $(\qa, \qb) \in \left(\frac{\Qmax}{n}\mathbb{N}^*\right)^2$ (if $n < \infty$), $(\qa, \qb) \in (0,\infty)^2$ (if $n = \infty$), with $\qa, \qb \leqslant 2 \Qmax$. We consider the limit order book of a large-tick stock on the time interval $[0,T]$.

If $n < \infty$, we introduce the sets of possible traded volumes $\Qa \defeq \left\{0,\frac{\Qmax}{n},\dots,\qa - \frac{\Qmax}{n},\qa\right\}$ and $\Qb \defeq \left\{0,\frac{\Qmax}{n},\dots,\qb - \frac{\Qmax}{n},\qb\right\}$. For $n=\infty$, let $\mathcal{Q}_{\infty}^{+,a} \defeq [0,\qa]$ and $\mathcal{Q}_{\infty}^{+,b} \defeq [0,\qb]$. Let $\mu^a$ (resp. $\mu^b$) be a probability measure on $\Qa$ (resp. $\Qb$) such that $\qa$ (resp. $\qb$) belongs to the support of $\mu^a$ (resp. $\mu^b$). The volumes of the arriving market orders will be independent, with distribution $\mu^a$ on the ask side and $\mu^b$ on the bid side.

\begin{rem}
    The assumption $\qa \in \mathrm{supp}(\mu^a)$ and the corresponding one on the bid side are made for identifiability purposes. Otherwise, one can always restrict $\Qa$ to its elements smaller than $\max \mathrm{supp}(\mu^a)$.
\end{rem}

To describe the dynamics of the mid-price and the market orders, we consider a filtered probability space $(\Omega, \mathcal{F}, (\mathcal{F}_t)_{t \in [0,T]}, \mathbb{P})$ satisfying the usual hypotheses which supports a Brownian motion $W$ and two marked Poisson processes $N^a$ (on $[0,T] \times \Qa$) and $N^b$ (on $[0,T] \times \Qb$), all independent. We assume that $N^a(\diff t ,\diff z)$ and $N^b(\diff t ,\diff z)$ have intensity kernels $\diff t\mu^a(\diff z)$ and  $\diff t\mu^b(\diff z)$, respectively. The processes $N^a,N^b$ model the dynamics of the market orders.

As in \parencite{avellaneda_high-frequency_2008}, we consider a stock whose efficient price $S$ follows a Bachelier model $\diff S_t = \sigma \diff W_t$ with volatility parameter $\sigma>0$. The tick size is $\delta>0$ and the mid-price $P$ is endogenous taking values on the inter-tick grid $\frac{\delta}{2} + \delta \mathbb{Z}$. The mid-price moves are modeled following a variant of the model with uncertainty zones defined by \textcite{robert_new_2011} with relative size of the uncertainty zone $\eta>0$: when $S$ hits the barrier $P_{t-} + \delta \left(\frac{1}{2} + \eta\right)$, then $P$ jumps upwards (i.e. $P_t = P_{t-} + \frac{\delta}{2}$), and when $S_t = P_{t-}-\delta\left(\frac{1}{2}+\eta\right)$, then $P$ jumps downwards (i.e. $P_t = P_{t-} - \frac{\delta}{2}$). The quantity $\delta \eta> 0$ is the size of the uncertainty zone. A similar model for price jumps was used by \textcite{baldacci_bid_2020}. A sample path is drawn in Figure \ref{fig:path_uncertainty_zones}.

The size of the uncertainty zone $\eta$ can be seen as a mean-reversion ratio: the lower it is, the higher the chances that the mid-price reverts to its previous value quickly after a change. As argued in \parencite{robert_new_2011}, it can be estimated empirically with $\hat{\eta} =\frac{N^c}{2 N^a}$, where $N^c$ is the number of \enquote{continuations} (occurring  when the mid-price moves in the same direction as its previous move) and $N^a$ is the number of \enquote{alternations} (occurring when the mid-price moves in the opposite direction).

\begin{rem}
    When no volume is quoted on one side, say the ask, $P_t$ is no longer the mid-price in the strict sense. It can be seen as a \enquote{reference mid-price}, as in \parencite{huang_simulating_2015}.
\end{rem}

\begin{figure}
    \centering
    \includegraphics[width=\textwidth]{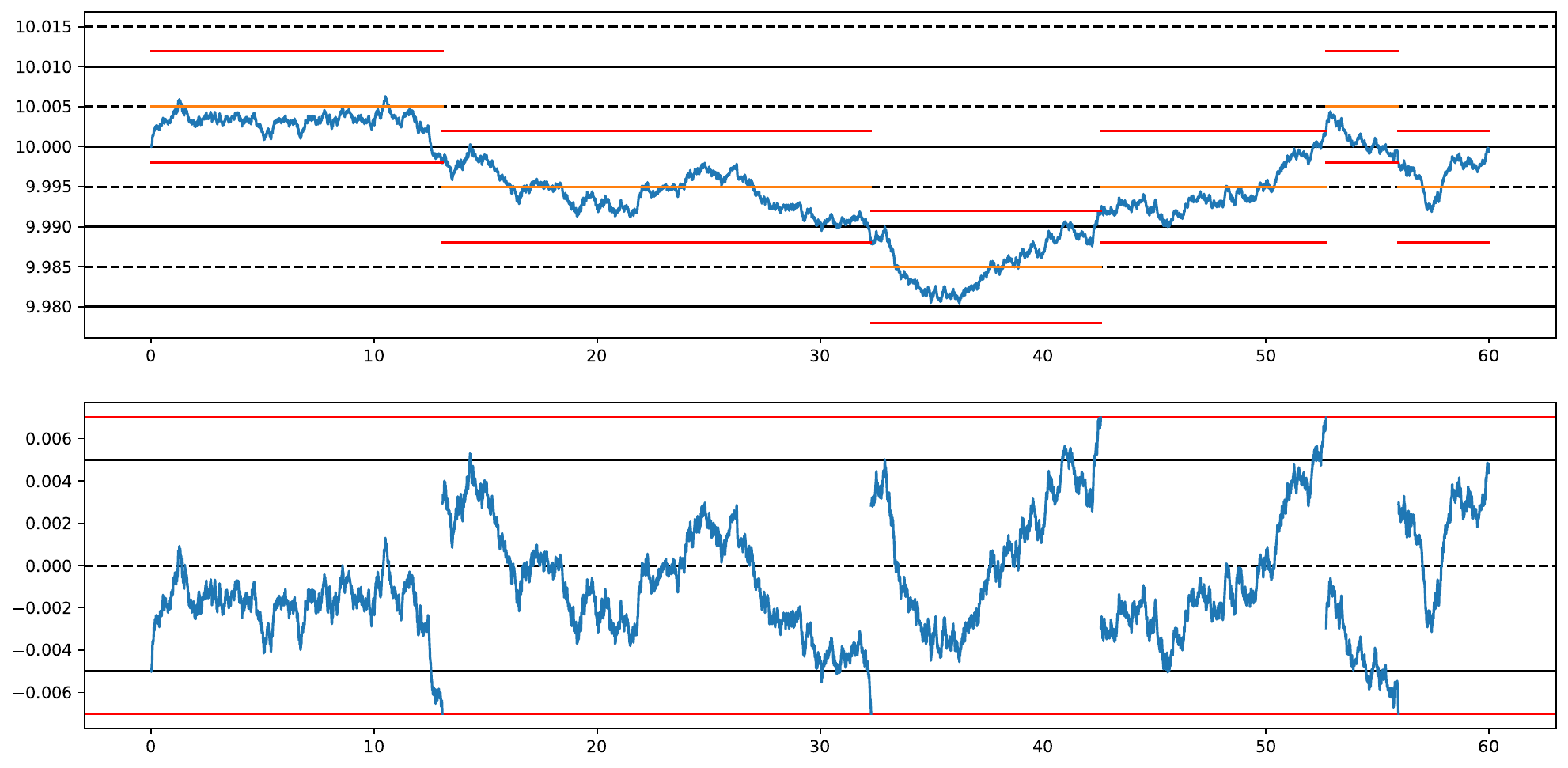}
    \caption{Above: One sample path of $S$ (in blue) and $P$ (in orange) on $[0,60]$. The parameters are $\sigma=0.002$, $\delta=0.01$, $\eta=0.2$. The uncertainty zones are drawn in red. Below: The corresponding $Y$ process.}
    \label{fig:path_uncertainty_zones}
\end{figure}

Instead of considering $S$ and $P$ separately, it is convenient to use the signed distance between the efficient price and the mid-price $Y = S-P$. We denote by $\Y$ this signed distance with $\Y_s = y$ for $s \leqslant t$. More rigorously, $\Y_s = Y(t,s,y,\sigma W)$ where $Y$ is the process taking values in $\mathcal{Y} \defeq \left(-\delta\left(\eta + \frac{1}{2}\right), \delta\left(\eta + \frac{1}{2}\right)\right)$ built in Appendix \ref{sec:midprice_construction} with $a = - \delta\left(\eta + \frac{1}{2}\right)$, $b = \delta\left(\eta + \frac{1}{2}\right)$, $a_0=\delta\left(-\eta + \frac{1}{2}\right)$, and $b_0=\delta\left(\eta - \frac{1}{2}\right)$. Let $p \in \frac{\delta}{2} + \mathbb{Z}$. From $\Y$, we can define $\Pm$ the mid-price process starting at $p$ at time $t$ as 
\begin{equation*}
     \Pm_s \defeq p + \delta\sum_{u \in [0,s]} \left(\mathds{1}_{\left\{\Delta \Y_u < 0\right\}} - \mathds{1}_{\left\{\Delta \Y_u > 0\right\}}\right),\quad s \in [0,T]
\end{equation*}
where $\Delta \Y_s = \Y_s - \Y_{s-}$. In order to lighten the notation, we will often write $\bar{y} \defeq \delta\left(\eta + \frac{1}{2}\right)$, $y_+ \defeq \delta\left(\eta - \frac{1}{2}\right)$ and $y_- \defeq \delta\left(-\eta + \frac{1}{2}\right)$.

Let $\Lambda^a, \Lambda^b:[0,T) \times \mathcal{Y} \mapsto (0,\infty)$ be two functions belonging to $C^{\alpha}$ for some $\alpha > 0$. In particular, $\Lambda^a, \Lambda^b$ are bounded by a constant $\Lambda^* > 0$. %{\color{red}We suppose further that they are bounded away from $0$.}
Let $\mathbb{P}^{t,y}$ a probability on $(\Omega, \mathcal{F}, (\mathcal{F}_t)_{t \in [0,T]})$, equivalent to $\Proba$, under which $N^a(\diff t, \diff z)$ and $N^b(\diff t, \diff z)$ have the stochastic intensity kernels $\Lambda^a(t,\Y_t) \mu^a(\diff z) \diff t$ and $\Lambda^b(t,\Y_t) \mu^b(\diff z) \diff t$, respectively. More precisely, 
\begin{equation*}
    \frac{\diff \mathbb{P}^{t,y}}{\diff \mathbb{P}} \defeq L^{t,y}_T
\end{equation*}
where $L^{t,y}$ is the Doléans-Dade exponential martingale given by
\begin{equation*}
    L^{t,y}_s = \prod_{i \in \{a,b\}} \exp\left(
        \int_{(0,s] \times \Qi} \ln\left(\Lambda^i(u,\Y_u)\right) N^{i}(\diff u, \diff z)
        - \int_0^s (\Lambda^i(u,\Y_u) - 1) \diff u
     \right), \quad s\in[0,T].
\end{equation*}
$L^{t,y}$ indeed defines a martingale, since $\Lambda^i$ is bounded \parencite[Section VI, Theorem 4]{bremaud_point_1981}.
In this way, under $\mathbb{P}^{t,y}$, the intensity of arrival of market orders depends on the distance between the efficient price and the mid-price. Whenever we take an expectation with respect to $\mathbb{P}^{t,y}$ we will write $\E[t,y][]$.

\subsection{The market maker's problem}

The market maker can quote, at prices $P \pm \frac{\delta}{2}$, volumes in $\Qa$ and $\Qb$ on the ask and bid side, respectively. We denote by $\mathcal{A}^{big}$ the space of predictable $\Qa \times \Qb$-valued processes.

Following the interpretation at the end of Remark \ref{rem:efficient_price_interpretation}, we suppose that in the primary market, the considered stock is a large-tick stock, meaning the bid-ask spread is almost always equal to one tick, and the best quotes are placed at $P \pm \frac{\delta}{2}$. If the market maker, on the platform where she trades, places her quote further, no one will trade with her since her price is worse than the one they would get on the other platform.

\begin{rem}
    When the market maker decides not to quote at the best bid or ask, the spread effectively increases. For example, if the efficient price is higher than $P + \frac{\delta}{2}$ and she decides not to place an ask order at $P + \frac{\delta}{2}$, she could place a bid limit at $P + \frac{\delta}{2}$ and still make a profit. Since she can also have her bid order executed at $P - \frac{\delta}{2}$, earning $\delta$ more (albeit with a lower probability), she would have to choose where to place the bid order. This would complicate the control problem, so we will not consider this case.
\end{rem}

Her inventory at time $s \in [0,T]$, starting at $t$ with value $Q$, if she chooses a control $q \in \mathcal{A}^{big}$, is
\begin{equation*}
    \Q_s \defeq Q + \int_{(t,s \vee t] \times \Qb} \left(q^b_u \wedge z\right) N^b(\diff u, \diff z) - \int_{(t,s \vee t] \times \Qa} \left(q^a_u \wedge z\right) N^a(\diff u, \diff z).
\end{equation*}

\begin{rem}
    In the formula above we have supposed that if a market order of size $z$ occurred, but the market maker only quoted $q < z$, then the whole $q$ she posted is consumed.
\end{rem}

As in \parencite{gueant_dealing_2013}, \parencite{el_euch_optimal_2021}, \parencite{barzykin_algorithmic_2023}, the market maker wishes to have an inventory bound in order to limit her exposure to big price moves. She therefore uses admissible controls from the space
\begin{equation*}
    \mathcal{A}^{t,Q} \defeq \left\{ q \in \mathcal{A}^{big} : q^a_s \leqslant \Q_{s-} + \Qmax \text{ and } q^b_s \leqslant -\Q_{s-} + \Qmax \text{, for all $s \in [0,T]$}\right\} 
\end{equation*}
if her inventory at time $t$ is $Q \in \mathcal{Q}_n$, where $\mathcal{Q}_n \defeq \left\{ -\Qmax, - \Qmax + \frac{\Qmax}{n},\dots, \Qmax - \frac{\Qmax}{n}, \Qmax\right\}$ if $n < \infty$ and $\mathcal{Q}_{\infty} \defeq \left[-\Qmax, \Qmax\right]$ if $n = \infty$.

If at time $t$ the inventory of the market maker is $Q \in \mathcal{Q}_n$, the mid-price is $p$, and the efficient price is $p+y$, her profit at the end of trading period using the control policy $q \in \mathcal{A}^{t,Q}$ is
\begin{equation*}
    \begin{split}
        \Q_T S_T -& \int_{(t,s \vee t] \times \Qb} \left(q^b_u \wedge z\right)\left(\Pm_u - \frac{\delta}{2}\right) N^b(\diff u, \diff z)\\
         &+ \int_{(t,s \vee t] \times \Qa} \left(q^a_u \wedge z\right) \left(\Pm_u + \frac{\delta}{2}\right)N^a(\diff u, \diff z),
    \end{split}
\end{equation*}
assuming that she is able to liquidate her position at the efficient price in the end. By a simple application of Itô's formula, this quantity is equal to $\PnL_T$ where, for $s\in[0,T]$,
\begin{equation*}
    %\label{eq:PnL}
    \begin{aligned}[t]
        \PnL_s = 
        &\int_{(t,s \vee t] \times \Qa} \left(q^a_u \wedge z\right) \left(\frac{\delta}{2} - \Y_u\right)N^a(\diff u, \diff z)\\
        &+
        \int_{(t,s \vee t] \times \Qb} \left(q^b_u \wedge z\right)\left(\frac{\delta}{2} + \Y_u\right) N^b(\diff u, \diff z)\\
        &+\sigma\int_t^{s \vee t} \Q_u \diff W_u.
    \end{aligned}
\end{equation*}

Let $\ell:\mathcal{Q}_n \mapsto \mathbb{R}$ be a \enquote{penalty} function for the inventory held after the trading period. If $n = \infty$, we assume that $\ell$ is continuous. The market maker optimizes an exponential utility function with risk aversion coefficient $\gamma > 0$. More precisely, at time $t$, if $Y = y$ and her inventory is $Q$, she aims to optimize
\begin{equation*}
    J(t,Q,y,q) \defeq \E[t,y][-e^{-\gamma\left(\PnL_T - \ell(\Q_T)\right)}]
\end{equation*}
over $q \in \mathcal{A}^{t,Q}$. We denote henceforth her optimal value function as
\begin{equation}
    \label{eq:value_function}
    U(t,Q,y) \defeq \sup_{q \in \mathcal{A}^{t,Q}} J(t,Q,y,q),\quad (t,Q,y) \in [0,T] \times \mathcal{Q}_n \times \mathcal{Y}.
\end{equation}

\section{Main results}
\label{sec:results}

Before stating the Hamilton-Jacobi-Bellman equation associated to the control problem \eqref{eq:value_function}, we define the corresponding Hamiltonians.
\begin{definition}\label{def:Hamiltonians}
    For $Q \in \mathcal{Q}_n$, $y \in \bar{\mathcal{Y}}$ and a bounded function $\phi:\mathcal{Q}_n \mapsto \mathbb{R}$, we define the Hamiltonians of our problem, provided that they are well-defined (i.e. either $\phi$ is measurable or the measures have finite support):
    \begin{align*}
        H_n^a(\phi,Q,y, \mu^a) & \defeq \sup_{q \in \Qa \cap [0, Q + \Qmax]}
        \int_{\Qa}e^{-\gamma (q \wedge z)\left(\frac{\delta}{2}-y\right)}\phi(Q - z \wedge q)\mu^a(\diff z)\\
        H_n^b(\phi,Q,y, \mu^b) & \defeq \sup_{q \in \Qb \cap [0, -Q + \Qmax]}
        \int_{\Qb}e^{-\gamma (q \wedge z)\left(\frac{\delta}{2}+y\right)}\phi(Q + z \wedge q)\mu^b(\diff z).
    \end{align*}
    We shall omit the arguments $n$, $\mu^a$ and $\mu^b$ whenever it is clear from the context.
\end{definition}

The Hamilton-Jacobi-Bellman (HJB) equation for our problem is, for $(t,y) \in [0,T) \times \mathcal{Y}$ and $Q \in \mathcal{Q}_n$,

\begin{equation}
    \label{eq:HJB_interior}
    \begin{aligned}
        0 = \partial_t u^Q(t,y) &+ \frac{\sigma^2}{2}\partial^2_{yy} u^Q(t,y)
        -\sigma^2 \gamma Q \partial_y u^Q(t,y) + \frac{\sigma^2 \gamma^2 Q^2}{2}u^Q(t,y) - \left(\Lambda^a + \Lambda^b\right)(t,y)u^Q(t,y)\\
        &+\Lambda^a(t,y) H^a\left(\left(u^R(t,y)\right)_{R \in \mathcal{Q}_n}, y, Q\right)
        +\Lambda^b(t,y) H^b\left(\left(u^R(t,y)\right)_{R \in \mathcal{Q}_n}, y, Q\right)
    \end{aligned}
\end{equation}
with boundary conditions, for $(t,y,Q) \in [0,T] \times \bar{\mathcal{Y}} \times \mathcal{Q}_n$,
\begin{equation}
    \label{eq:HJB_border}
    \left\{
        \begin{array}{ll}
            u^Q(T, y) &=-e^{\gamma\ell(Q)} \\
            u^Q\left(t, \ybar\right) &= u^Q\left(t, \yplus\right) \\
            u^Q\left(t, -\ybar\right) &= u^Q\left(t, \yminus\right).
        \end{array}
    \right.
\end{equation}

\begin{rem}
    In equation \eqref{eq:HJB_interior} we use the following convention: for each $(t,y)$, $\left(u^R(t,y)\right)_{R \in \mathcal{Q}_n}$ denotes the real-valued function $R \in \mathcal{Q}_n \mapsto u^R(t,y)$.
\end{rem}
Theorem \ref{thm:existence} below is our first main result. The first part, which is proven in Section \ref{sec:hjb}, states the existence of classical solutions of the HJB equation \eqref{eq:HJB_interior}-\eqref{eq:HJB_border}. The second part, proved in Appendix \ref{sec:convergence_discrete_continuous}, shows that the continuous inventory case can be uniformly approximated by the discrete inventory case. 

\begin{theorem}[Existence and convergence]
    \label{thm:existence}
    $ $
    \begin{itemize}
        \item[(i)] There exists a family of functions $(u^Q)_{Q \in \mathcal{Q}_n}$ and $\beta \in (0,1)$ such that: 
        \begin{itemize}
            \item $(t,y,Q)\mapsto u^Q(t,y)$ is continuous on $[0,T] \times \bar{\mathcal{Y}} \times \mathcal{Q}_{n}$,
            \item for all $Q \in \mathcal{Q}_n$, $u^Q \in  C^{1,2}([0,T)\times \mathcal{Y})$,
            \item $\sup\limits_{Q \in \mathcal{Q}_n}|u^Q|_{2+\beta}<\infty$,
            \item the HJB equation \eqref{eq:HJB_interior}-\eqref{eq:HJB_border} holds.
        \end{itemize}
        Furthermore, the constant $\beta$ can be chosen to depend only on $\alpha$, $T$, $\Qmax$, $\delta$, $\gamma$, $\sigma$, $\eta$, $\Lambda^*$ and not on $n$, $\mu^a$, $\mu^b$, $\qa$, $\qb$.
        \item[(ii)] Fix $\qa, \qb \in (0,2\Qmax]^2$. For $n \in \mathbb{N}^*$, let $\qa_{n} \defeq \frac{\Qmax}{n}\left\lfloor \frac{n}{\Qmax}\qa \right\rfloor$ and $\qb_{n} \defeq \frac{\Qmax}{n}\left\lfloor \frac{n}{\Qmax}\qb \right\rfloor$. Define, just for the scope of this theorem, $\mathcal{Q}_{n}^{+,a} \defeq \left\{0,\frac{\Qmax}{n},\dots,\qa_{n} - \frac{\Qmax}{n},\qa_{n}\right\}$ and $\mathcal{Q}_{n}^{+,b} \defeq \left\{0,\frac{\Qmax}{n},\dots,\qb_{n} - \frac{\Qmax}{n},\qb_{n}\right\}$.
    Let $(u^Q)_{Q \in \mathcal{Q}_{\infty}}$ be a solution given by point (i) with $n=\infty$, and $\mu^a$, $\mu^b$ be two probability measures on $\mathcal{Q}^{+,a}_{\infty}$ and $\mathcal{Q}^{+,b}_{\infty}$, respectively.\\
    Then, there exist two sequences of probability measures $(\mu^a_{n})_{n \in \mathbb{N}^*}$, $(\mu^b_{n})_{n \in \mathbb{N}^*}$ such that for $n \in \mathbb{N}^*$, $\mu^a_{n}$ (resp. $\mu^b_{n}$) is a measure on $\mathcal{Q}_{n}^{+,a}$ (resp. $\mathcal{Q}_{n}^{+,b}$) and $\qa_{n}$ (resp. $\qb_{n}$) is in its support, and if $(u_n^Q)_{Q \in \mathcal{Q}_{n}}$ verifies all the conditions of point (i), with associated measures $\mu^a_{n}$, $\mu^b_{n}$, we have
    \begin{equation*}
        \lim_{n \to \infty} \sup_{Q \in \mathcal{Q}_{n}}\left|u^{Q}_{n}-u^Q \right|_{\infty} = 0.
    \end{equation*}
    \end{itemize}
\end{theorem}

Our second main result is the following verification theorem, whose proof can be found in Section \ref{sec:verification}. It states that the value function is the solution of the HJB equation and the optimal controls maximize the corresponding Hamiltonians. Moreover, there exist Markovian optimal controls.

\begin{theorem}[Verification]
    \label{thm:verification}
    Let $(u^Q)_{Q \in \mathcal{Q}_n}$ be a family of functions verifying the conditions of (i) in Theorem \ref{thm:existence}.
    \begin{itemize}
        \item[(i)] There exist two measurable functions $\hat{q}^a$, $\hat{q}^b:[0,T] \times \mathcal{Q}_n \times \bar{\mathcal{Y}}\mapsto [0,\infty)$ such that, for all $(t,Q,y)\in [0,T] \times \mathcal{Q}_n \times \bar{\mathcal{Y}}$, $\hat{q}^a(t,Q,y) \leqslant Q + \Qmax $, $\hat{q}^b(t,Q,y) \leqslant -Q + \Qmax$ and
    \begin{align*}
        H^a\left(\left(u^R(t,y)\right)_{R \in \mathcal{Q}_n}, y, Q\right)
        & = \int_{\Qa} e^{-\gamma \left(\hat{q}^a(t,Q,y) \wedge z\right)\left(\frac{\delta}{2} - y\right)}u^{Q - \hat{q}^a(t,Q,y) \wedge z}(t,y)\mu^a(\diff z),\\
        H^b\left(\left(u^R(t,y)\right)_{R \in \mathcal{Q}_n}, y, Q\right) 
        & = \int_{\Qb}e^{-\gamma \left(\hat{q}^b(t,Q,y) \wedge z\right)\left(\frac{\delta}{2} + y\right)}u^{Q + \hat{q}^b(t,Q,y) \wedge z}(t,y) \mu^b(\diff z).
    \end{align*}
        \item[(ii)] For all $(t,Q,y) \in [0,T) \times \mathcal{Q}_n \times \mathcal{Y}$, there exists a unique (up to an evanescent set) control $q^*=(q^{*a}, q^{*b}) \in \mathcal{A}^{t,Q}$ such that, almost surely, for all $s\in [0,T)$
        \begin{equation}
        \label{eq:existence_control}
        q^{*a}_s = \hat{q}^a\left(s \vee t, Q^{t,Q,q^*}_{s-}, \Y_s\right),\quad q^{*b}_s = \hat{q}^b\left(s \vee t, Q^{t,Q,q^*}_{s-}, \Y_s\right).
        \end{equation} 
        \item[(iii)] For all $(t,Q,y) \in [0,T] \times \mathcal{Q}_n \times \mathcal{Y}$, $u^Q(t,y) = U(t,Q,y)$. Furthermore, $U(t,Q,y) = J(t,Q,y,q^*)$.
        \item[(iv)] Let $(t,Q,y) \in [0,T) \times \mathcal{Q}_n \times \mathcal{Y}$ and $q = (q^a, q^b)\in \mathcal{A}^{t,Q}$ such that $J(t,Q,y,q) = U(t,Q,y)$. Then, on $[t,T] \times \Omega$,
    \begin{align*}
        q^a_s &\in \argmax_{q \in \Qa \cap [0,\Q_s+\Qmax]}\int_{\Qa}e^{-\gamma (q \wedge z)\left(\frac{\delta}{2} - \Y_s\right)}u^{\Q_s - z \wedge q}\left(s,\Y_s\right)\mu^a(\diff z),\\
        q^b_s &\in \argmax_{q \in \Qb \cap [0,-\Q_s+\Qmax]}\int_{\Qb}e^{-\gamma (q \wedge z)\left(\frac{\delta}{2} + \Y_s\right)}u^{\Q_s + z \wedge q}\left(s,\Y_s\right)\mu^b(\diff z)
    \end{align*}
    $\diff s \otimes \Proba[][\diff \omega]$-almost everywhere.   
    \end{itemize}
\end{theorem}

\begin{rem}
    Theorem \ref{thm:verification}(iii) yields the uniqueness of families of functions $(u^Q)_{Q \in \mathcal{Q}_n}$ verifying the conditions of Theorem \ref{thm:existence}(i). It allows us to extend $U$ to $[0,T] \times \mathcal{Q}_n \times \bar{\mathcal{Y}}$ by continuity.
\end{rem}

Our last main result, Theorem \ref{thm:uniqueness}, proves the uniqueness of the optimal control policies in the case of continuous inventory. This follows from Proposition \ref{prop:log_concavity} which shows that the negative of the value function is log-convex. Section \ref{sec:uniqueness} contains the proofs of these results.

\begin{proposition}
    \label{prop:log_concavity}
    Suppose $n=\infty$.
    \begin{itemize}
        \item[(i)] For all $(t,Q,y) \in [0,T] \times \mathcal{Q}_{\infty} \times \bar{\mathcal{Y}}$, $U(t,Q,y) < 0$.
        \item[(ii)] Suppose $\ell$ is convex. Then, for all $(t,y) \in [0,T) \times \bar{\mathcal{Y}}$, the function $Q \in \mathcal{Q}_{\infty} \mapsto -\ln\left(-U(t,Q,y)\right)$ is strictly concave.
    \end{itemize}
\end{proposition}

\begin{theorem}
    \label{thm:uniqueness}
    Suppose $n = \infty$ and $\ell$ is convex. Let $(t,Q,y) \in [0,T) \times \mathcal{Q}_{\infty} \times \bar{\mathcal{Y}}$. Let $q^* \in \mathcal{A}^{t,Q}$ be a control policy given by Theorem \ref{thm:verification}(ii). Then, the following hold:
    \begin{itemize}
        \item[(i)]{\small
    \begin{align*}
        \argmax_{q \in \mathcal{Q}_{\infty}^{+,a} \cap [0,Q+\Qmax]}\int_{\mathcal{Q}_{\infty}^{+,a}}e^{-\gamma (q \wedge z)\left(\frac{\delta}{2} - y\right)}U\left(t,Q - q \wedge z,y\right)\mu^a(\diff z)
        &=
        \argmax_{q \in \mathcal{Q}_{\infty}^{+,a} \cap [0,Q+\Qmax]}e^{-\gamma q\left(\frac{\delta}{2} - y\right)}U\left(t,Q - q,y\right),\\
        \argmax_{q \in \mathcal{Q}_{\infty}^{+,b} \cap [0,-Q+\Qmax]}\int_{\mathcal{Q}_{\infty}^{+,b}}e^{-\gamma (q \wedge z)\left(\frac{\delta}{2} + y\right)}U\left(t,Q + q \wedge z,y\right)\mu^b(\diff z)
        &=
        \argmax_{q \in \mathcal{Q}_{\infty}^{+,b} \cap [0,-Q+\Qmax]}e^{-\gamma q\left(\frac{\delta}{2} + y\right)}U\left(t,Q + q,y\right),
    \end{align*}}
    and these sets contain only one element.
        \item[(ii)] For every control $q \in \mathcal{A}^{t,Q}$ such that $J(t,Q,y,q) = U(t,Q,y)$, $q_s(\omega) = q_s^*(\omega)$ $\diff s \otimes \Proba[][\diff \omega]$-almost everywhere on $[t,T] \times \Omega$.
        \item[(iii)] The functions $\hat{q}^a$ and $\hat{q}^b$ defined in Theorem \ref{thm:verification}(i) (which are uniquely determined on $[0,T) \times \mathcal{Q}_{\infty} \times \bar{\mathcal{Y}}$ thanks to (i) above) are continuous on $[0,T) \times \mathcal{Q}_{\infty} \times \bar{\mathcal{Y}}$.
    \end{itemize}
\end{theorem}

\section{Numerical results}
\label{sec:numerics}
We approximate the value function $U$ and the corresponding optimal controls $\hat{q}^a$, $\hat{q}^b$ using an implicit-explicit finite difference scheme for the Hamilton-Jacobi-Bellman equation. The scheme is implicit over the linear part of the equation and explicit over the non-linear part corresponding to the Hamiltonians. We discretize $[0,T]$ into $n_t = 7500$ subintervals and $\mathcal{Y}$ into $n_y = 350$ subintervals.

In our experiments, we use $n=100$, $T=3600$, $\delta = 0.01$, $\eta=0.2$, $\Qmax = 50$. We consider symmetric, time-independent, affine intensities at the bid and ask sides: $\Lambda^a:y \in \mathcal{Y}\mapsto Ay + B$, $\Lambda^b:y \in \mathcal{Y}\mapsto -Ay + B$ where $B > 0$ and $A \in \mathbb{R}$ such that $|A|\delta\left(\eta + \frac{1}{2}\right) < B$ (this guarantees that $\Lambda^a$ and $\Lambda^b$ remain positive). We choose $\qa = \qb = 100$ and the same power-law execution measure on the bid and ask sides: for $q \in \integerInterval[0]{100}$, $\mu^a\left(\left\{q\right\}\right) = \mu^b\left(\left\{q\right\}\right) = \frac{0.9^q}{\sum_{i=0}^{100}0.9^i}$. We consider a quadratic penalty function $\ell : Q \in \mathcal{Q}_n \mapsto 0.001 \cdot Q^2$.

\subsection{General observations}

%In this subsection, we take $\gamma = 1$, $\sigma = 0.005$, $A = 10$ and $B = 0.1$. The same observations can be made with other sets of parameters, even when $A < 0$, which is the (unrealistic) case where more market orders occur when $S_t - P_t$ is disadvantageous for the market takers.

In this subsection we consider positive values of $A$. The same observations can be made, even when $A$ is negative, which is the (unrealistic) case where more market orders occur when $S_t - P_t$ is disadvantageous for the market takers.

For the largest part of the trading period, and for fixed levels of $(Q,y)$, the market maker's optimal quotes $\hat{q}^a(t,Q,y)$ and $\hat{q}^b(t,Q,y)$ are constant with respect to $t$ (Figure \ref{fig:convergence_controls}). At the end of the trading interval, the market maker tends to choose a control policy that will bring her close to an inventory of $0$, since the profit made by a trade hardly compensates the penalty $\ell$.

\begin{figure}
    \centering
    \includegraphics[width=\textwidth,page=1]{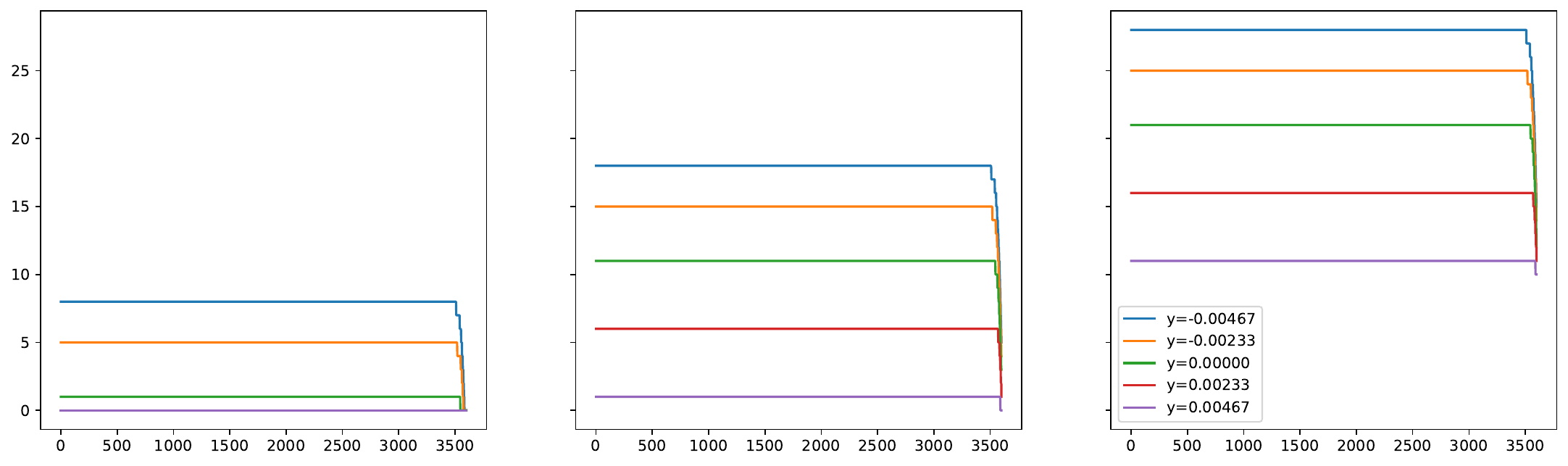}
    \caption{The optimal ask quote $\hat{q}^a(t,Q,y)$ over time, for different values of $y$. Left: $Q = -10$. Center: $Q=0$. Right: $Q=10$. $\gamma = 1$, $\sigma=0.005$, $A=10$ and $B = 0.1$.}
    \label{fig:convergence_controls}
\end{figure}

As expected, for fixed $(t,Q)$, $\hat{q}^a(t,Q,\cdot)$ is nonincreasing and $\hat{q}^b(t,Q,\cdot)$ is nondecreasing, see Figure \ref{fig:optimal_controls}. Indeed, the profit made by selling a unit of stock decreases with $y$. Similarly, the profit made by buying a unit of stock increases with $y$. We also observe that for some values of $y$, the market maker has no interest in buying or selling: if she sells when $y > \frac{\delta}{2}$, she \enquote{loses} money from the sell. Even when $y < \frac{\delta}{2}$, the gain might not be enough to compensate for the inventory risk. Of course, when her inventory is negative, she is willing to buy more and sell less in order to bring it back close to zero, making her immune to price changes. Similar behavior arises when her inventory is positive.

\begin{figure}
    \centering
    \includegraphics[width=\textwidth,page=1]{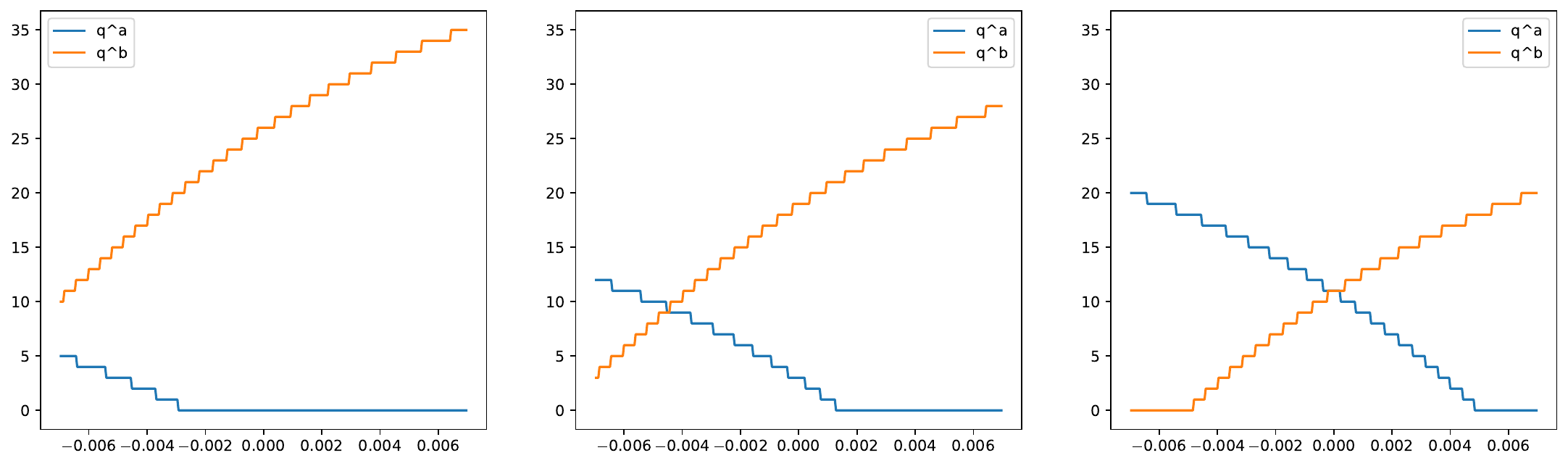}
    \caption{The optimal bid and ask quotes, as a function of $y$. Left: $Q = -15$. Center: $Q=-8$. Right: $Q=0$. $\gamma = 1$, $\sigma=0.005$, $A=10$ and $B = 0.1$.}
    \label{fig:optimal_controls}
\end{figure}

When quantities $q^a$ and $q^b$ are quoted on the ask and bid side,  respectively, we define the volume imbalance to be $\frac{q^b-q^a}{q^b + q^a}$ (undefined here when $q^a=0$ and $q^b=0$). It is a known empirical fact \parencite{huang_simulating_2015,lehalle_limit_2017,stoikov_micro-price_2018,sfendourakis_lob_2023} that volume imbalance has high predictive power for the direction of the next mid-price move: the closer it is to one, the more likely it is for the mid-price to go up. In our model, this translates to: the higher the imbalance at time $t$, the higher $Y_t$. We confirm numerically this fact in Figure \ref{fig:imbalance}. There is, however, a dependence on the inventory $Q_t$ of the market maker -- which is, in principle unknown to the external observer. This figure elucidates in a precise fashion the subtle connection between prices and volume imbalance. In particular, we can observe that the form of the price-imbalance curve depends on the chosen parameters. In some cases, given a fixed level of the inventory and ignoring the quoting bounds, the relationship is essentially affine.

\begin{figure}
    \centering
    \includegraphics[width=0.32\textwidth,page=2]{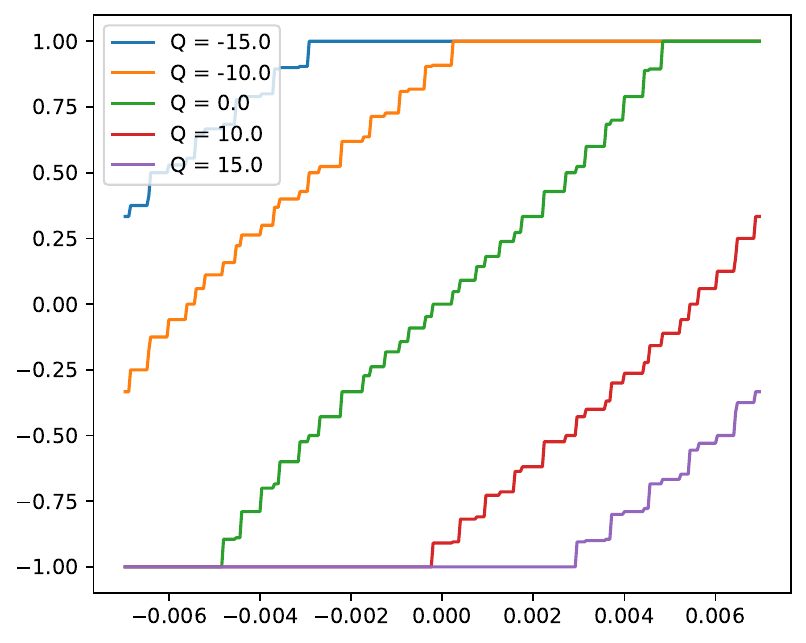}
    \includegraphics[width=0.32\textwidth,page=2]{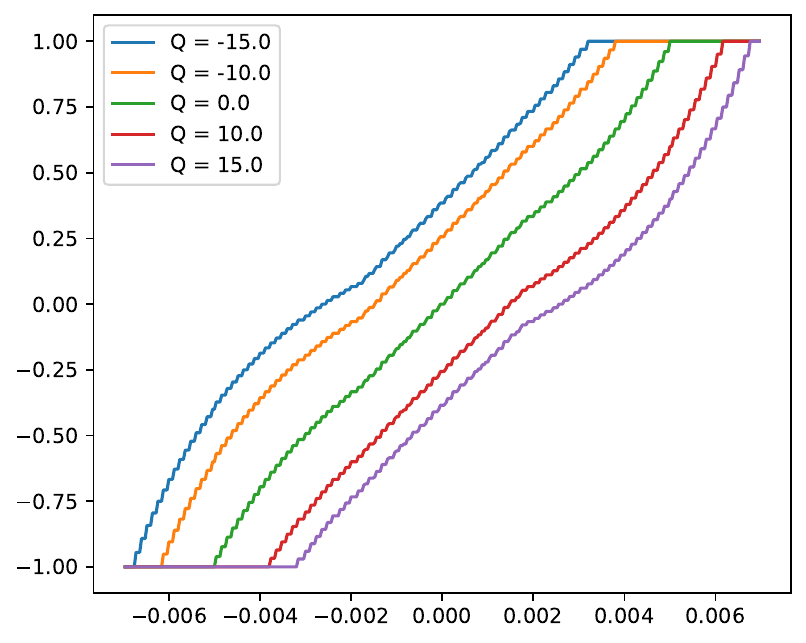}
    \includegraphics[width=0.32\textwidth,page=2]{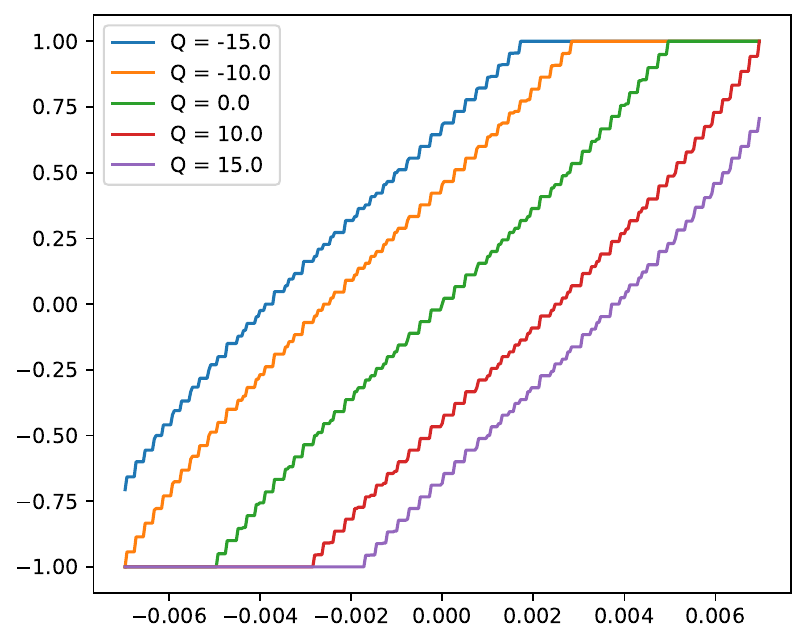}
    \caption{The quoted imbalance as a function of $y$, for different values of the market maker's inventory. Left: $\gamma = 1$, $\sigma=0.005$, $A=10$ and $B = 0.1$. Center: $\gamma = 0.5$, $\sigma=0.002$, $A=10$ and $B = 0.1$. Right: $\gamma = 0.5$, $\sigma=0.005$, $A=15$ and $B = 0.2$.}
    \label{fig:imbalance}
\end{figure}

\subsection{Dependence on the parameters}

In this subsection, we study the sensitivity of the optimal controls with respect to the parameters of the model.

When the risk sensitivity $\gamma$ increases, the market maker wants her inventory to stay closer to 0 (Figure \ref{fig:different_gamma}). Indeed, she becomes more sensitive to price moves. Formally, the efficient price being Gaussian, her \enquote{utility loss} while having nonzero inventory increases in a quadratic way with $\gamma$, while the gain from a trade only increases linearly.

\begin{figure}
    \centering
    \includegraphics[width=\textwidth]{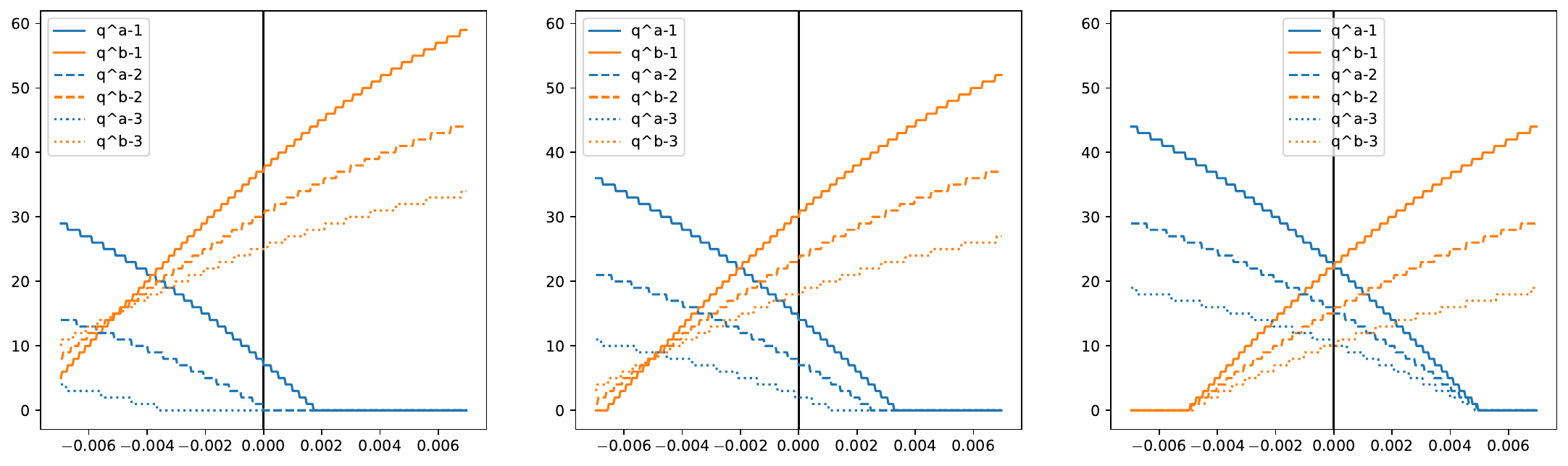}
    \caption{The optimal bid and ask quotes, as a function of $y$. Left: $Q = -15$. Center: $Q=-8$. Right: $Q=0$. $\sigma=0.005$, $A = 15$ and $B=0.2$. Configuration 1: $\gamma=0.5$. Configuration 2: $\gamma=1$. Configuration 3: $\gamma=2$.}
    \label{fig:different_gamma}
\end{figure}

The same behavior is observed when the volatility of the efficient price $\sigma$ increases (Figure \ref{fig:different_sigma}). In this situation, the price tends to move further, and it puts the market maker at risk.

\begin{figure}
    \centering
    \includegraphics[width=\textwidth]{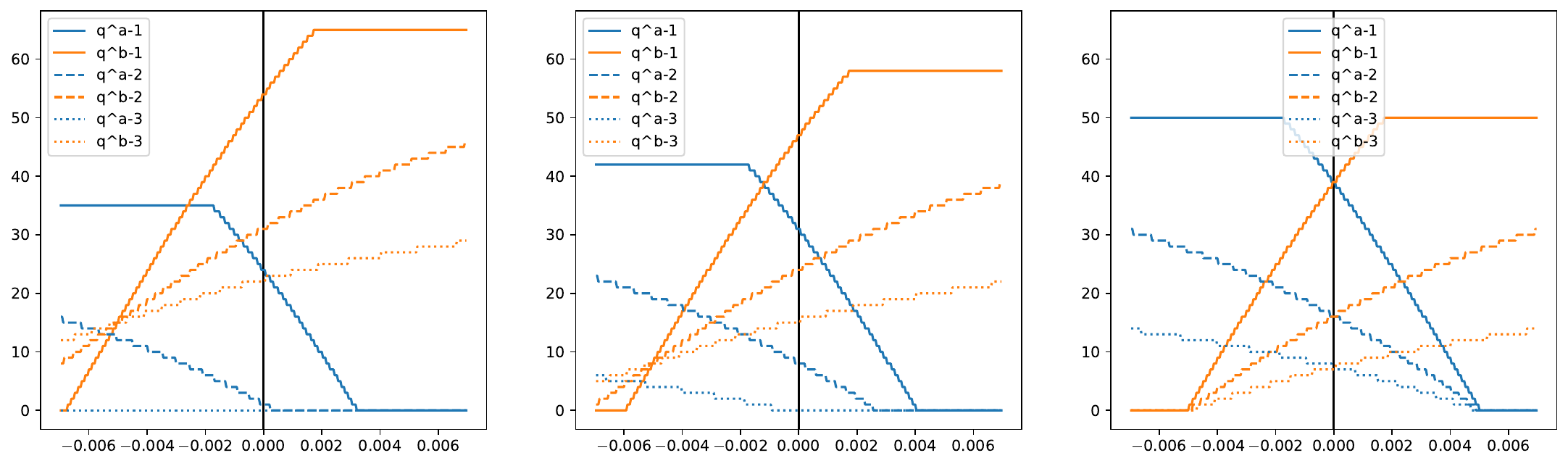}
    \caption{The optimal bid and ask quotes, as a function of $y$. Left: $Q = -15$. Center: $Q=-8$. Right: $Q=0$. $\gamma=0.5$, $A = 10$ and $B=0.1$. Configuration 1: $\sigma=0.002$. Configuration 2: $\sigma=0.005$. Configuration 3: $\sigma=0.01$.}
    \label{fig:different_sigma}
\end{figure}

The market maker tends to post bigger quotes when the market order intensities are bigger (Figure \ref{fig:different_intensities}). When the market orders are more frequent, the market maker can tolerate having a large inventory: she knows she can liquidate it rather quickly. In addition, for high intensities of the market orders (configuration 4 in Figure \ref{fig:different_intensities}), the market maker can quote even when she \enquote{loses} money from this very trade. For example, if we concentrate on the ask side, if her inventory is 0, $Y_t=y > \frac{\delta}{2}$ and she quotes some quantity $q > 0$ to sell, she expects to lose $q\left(y-\frac{\delta}{2}\right)$ from this trade (additionally, she exposes herself to price moves). However, she expects this quantity to be bought again quickly, earning $q\left(y+\frac{\delta}{2} -\varepsilon\right)$ ($\varepsilon$ being the small efficient price move in between), making a net profit of $q\delta - \varepsilon$. 

\begin{figure}
    \centering
    \includegraphics[width=\textwidth]{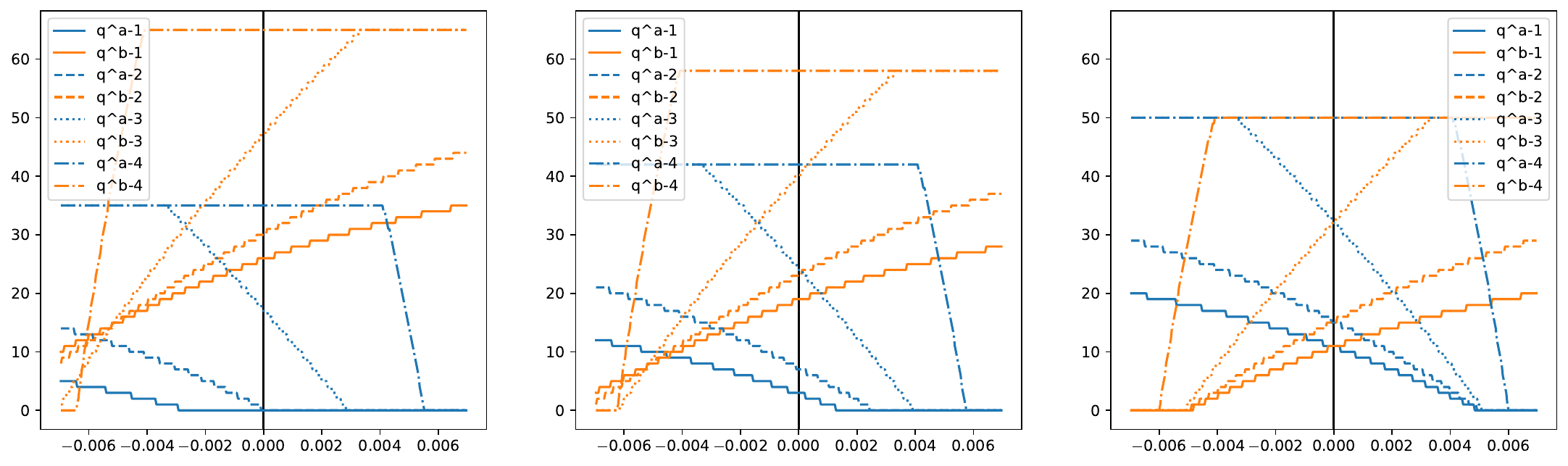}
    \caption{The optimal bid and ask quotes, as a function of $y$. Left: $Q = -15$. Center: $Q=-8$. Right: $Q=0$. $\gamma=1$ and $\sigma=0.005$. Configuration 1: $B=0.1$ and $A = 10$. Configuration 2: $B=0.2$ and $A = 10$. Configuration 3: $B=1$ and $A = 10$. Configuration 4: $B=40$ and $A = 0$.}
    \label{fig:different_intensities}
\end{figure}

\subsection{Platform's optimal tick size}

In this section, we present an interesting application of our framework to determine a tick size that is optimal in a sense to be described below.

We consider a platform aiming to attract liquidity by maximizing the total traded volume. The motivation behind this objective could be, for example, to maximize profits derived from fixed transaction costs per unit traded. A principal-agent approach to achieve the platform's objective would involve offering a contract to the market maker, encouraging her to quote in a specific manner. This approach was studied in \parencite{el_euch_optimal_2021}.

Changing the tick size also impacts market liquidity, as examined, for example, by \textcite{laruelle2019assessing}. Consequently, as discussed in \parencite{baldacci_bid_2020}, selecting an appropriate tick size can be an alternative strategy to maximize traded volumes. In our setting, we adopt this strategy by computing the total expected traded volume across a grid of tick sizes and selecting the one that maximizes it.

For a fixed tick size, the value function (total expected traded volume) of the platform is, for $(t,Q,y) \in [0,T] \times \mathcal{Q}_n \times \mathcal{Y}$:
\begin{equation}
    \begin{aligned}[t]
    W(t,Q,y) \defeq \mathbb{E}^{t,y}\Bigg[\int_{(t,T]\times\Qa}&\left(\hat{q}^a(u, Q^{t,Q,\hat{q}}_u, Y^{t,y}_u)\wedge z\right) N^a(\diff u, \diff z) \\
    &+ \int_{(t,T]\times\Qb}\left(\hat{q}^b(u, Q^{t,Q,\hat{q}}_u, Y^{t,y}_u)\wedge z\right) N^b(\diff u, \diff z)\Bigg]
    \end{aligned}
\end{equation}
where $\hat{q}_u \defeq (\hat{q}^a(u, Q^{t,Q,\hat{q}}_u, Y^{t,y}_u), \hat{q}^b(u, Q^{t,Q,\hat{q}}_u, Y^{t,y}_u))$ is the optimal strategy of the market maker as considered in Theorem \ref{thm:verification}. The value function verifies, at least formally, the system (over $Q \in \mathcal{Q}_n$) of PDEs:
\begin{equation}
    \label{eq:hjb_platform_interior}
    \begin{aligned}
        0 =&\partial_t W(t,Q,y) + \frac{\sigma^2}{2}\partial^2_{yy} w(t,Q,y)\\
        &+\Lambda^a(t,y) \int_{\Qa}\left(W(t, Q - \hat{q}^a(t,Q,y)\wedge z, y) - W(t,Q,y) + \hat{q}^a(t,Q,y)\wedge z\right)\mu^a(\diff z)\\
        &+ \Lambda^b(t,y)
        \int_{\Qb}\left(W(t, Q + \hat{q}^b(t,Q,y)\wedge z, y) - W(t,Q,y) + \hat{q}^b(t,Q,y)\wedge z\right)\mu^b(\diff z).
    \end{aligned}
\end{equation}
for $(t,Q,y) \in [0,T] \times \mathcal{Q}_n \times \mathcal{Y}$, with boundary conditions
\begin{equation}
    \label{eq:HJB_platform_border}
    \left\{
        \begin{array}{ll}
            W(T, Q, y) &=0 \\
            W\left(t, Q, \ybar\right) &= W\left(t, Q, \yplus\right) \\
            W\left(t, Q, -\ybar\right) &= W\left(t, Q, \yminus\right).
        \end{array}
    \right.
\end{equation}
We compute $W(0,0, \cdot)$ from \eqref{eq:hjb_platform_interior}-\eqref{eq:HJB_platform_border} using as before an implicit-explicit finite-difference scheme, for each considered tick size.

We conduct numerical experiments with intensity functions $\Lambda^a:y\mapsto Ae^{(By) \wedge 0}$, $\Lambda^b:y\mapsto Ae^{(-By) \wedge 0}$ with $A=1.5$ and $B=200$. The other parameters are $T=240$ (the time horizon plays a minimal role as observed in Figure \ref{fig:convergence_controls}), $\qa = \qb = 100$, power-law execution measures: for $q \in \integerInterval[0]{100}$, $\mu^a\left(\left\{q\right\}\right) = \mu^b\left(\left\{q\right\}\right) = \frac{0.9^q}{\sum_{i=0}^{100}0.9^i}$, and a penalty function $\ell : Q \in \mathcal{Q}_n \mapsto 0.005 \cdot Q^2$.

As observed empirically by \textcite[equation (5)]{dayri2015large}, the uncertainty zone parameter $\eta$ depends on the tick size $\delta$ according to $\eta = \eta_0\sqrt{\frac{\delta_0}{\delta}}$ (one could take a more general power-law dependence, depending on the typical limit order book shape). We choose here $\eta_0=0.2$ and $\delta_0=0.1$. A similar parametrization is used by \textcite{baldacci_bid_2020}. We only consider tick sizes $\delta$ for which $\eta<\frac12$ since we study large-tick assets.

The value function of the platform, averaged over the argument $y$, as a function of the tick size is plotted in Figure \ref{fig:platform_tick} for different levels of volatility. With a volatility $\sigma=0.005$, we observe an optimal tick size of $0.0032$ for the platform. With a smaller tick size, the market maker makes a low profit per trade and thus is not incentivized to trade much. With a bigger tick size, the proposed price by the market maker is often further away from the efficient price, reducing the number of incoming market orders. A similar effect is observed with different parameters $\sigma$.

Also, as $\sigma$ increases, the value function of the platform decreases: the market maker takes a higher risk holding inventory, hence trades smaller quantities. In that case, it is better for the platform to offer her a larger tick size so she can earn more per unit traded: the optimal tick sizes for volatilities $\sigma$ equal to $0.005$, $0.0075$, $0.01$ and $0.015$ are $0.0032$, $0.0044$, $0.0064$ and $0.015$, respectively.

\begin{figure}
    \centering
    \includegraphics[width=0.45\textwidth]{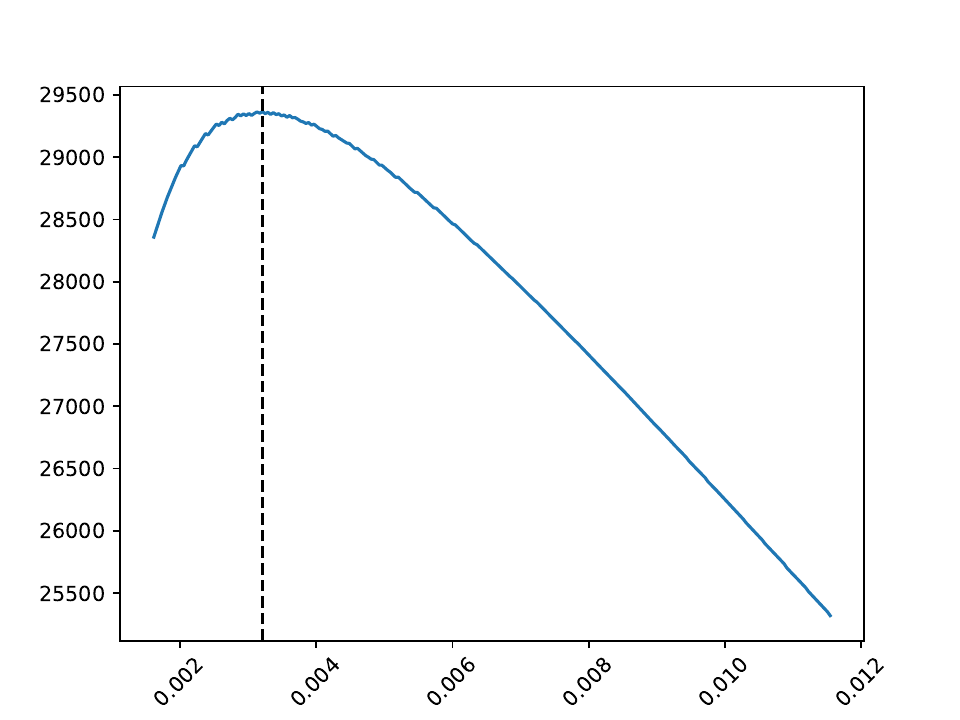}
    \includegraphics[width=0.45\textwidth]{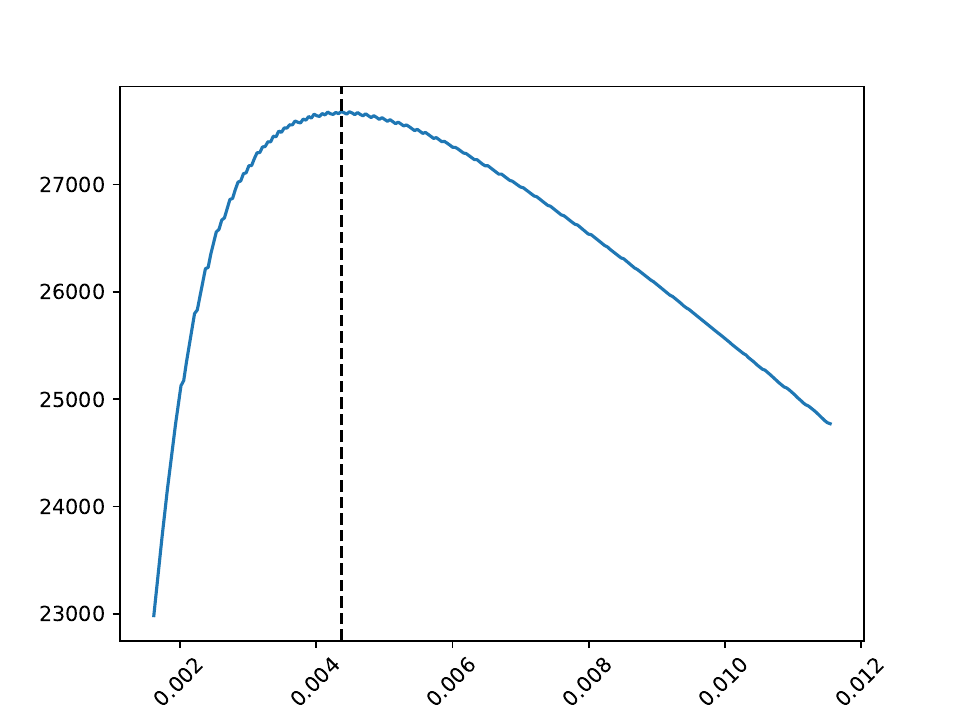}
    \includegraphics[width=0.45\textwidth]{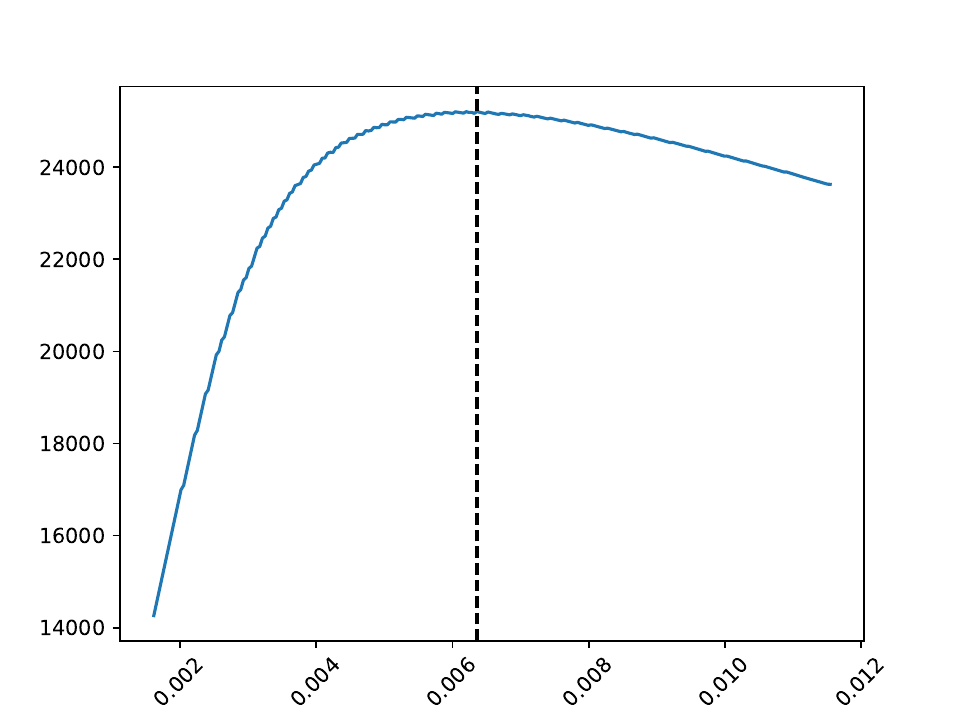}
    \includegraphics[width=0.45\textwidth]{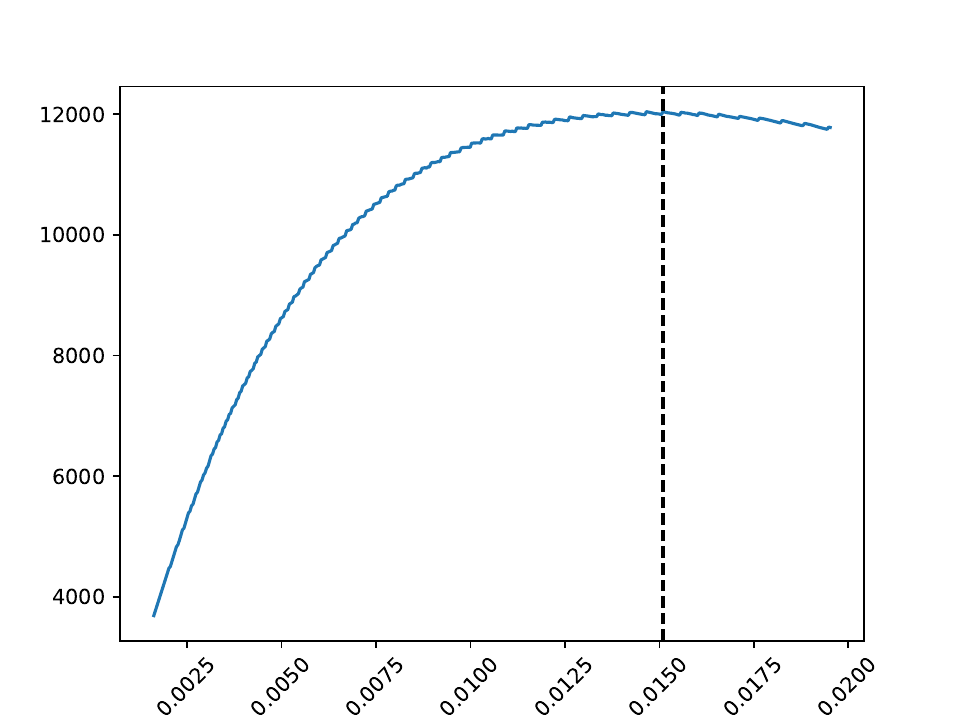}
    \caption{The platform's value function $W$, averaged uniformly over $y$, as a function of the tick size $\delta$, at $t=0$, $Q=0$. Upper left: $\sigma=0.005$. Upper right: $\sigma=0.0075$. Lower left: $\sigma=0.01$. Lower right: $\sigma=0.015$.}
    \label{fig:platform_tick}
\end{figure}

\section{Proof of the verification theorem}
\label{sec:verification}
In this section we prove Theorem \ref{thm:verification}.
\subsection{Proof of point (i)}

Point (i) follows from a standard measurable selection argument. Define
\begin{align*}
D^a &\defeq \left\{(t,Q,y,q^a) \in [0,T] \times \mathcal{Q}_n \times \bar{\mathcal{Y}}\times \Qa : q^a \leqslant Q + \Qmax\right\},\\
D^b &\defeq \left\{(t,Q,y,q^b) \in [0,T] \times \mathcal{Q}_n \times \bar{\mathcal{Y}}\times \Qb : q^b \leqslant -Q + \Qmax\right\}.
\end{align*}
The sets $D^a$ and  $D^b$ are closed. In addition, the functions 
\begin{align*}
    (t,Q,y,q^a) \in D^a &\mapsto \int_{\Qa}e^{-\gamma (q^a \wedge z)\left(\frac{\delta}{2}-y\right)}u^{Q-q^a \wedge z}\left(t,y\right) \mu^a(\diff z),\\
    (t,Q,y,q^b) \in D^b &\mapsto \int_{\Qb}e^{-\gamma (q^b \wedge z)\left(\frac{\delta}{2}+y\right)}u^{Q+q^b \wedge z}\left(t,y\right) \mu^b(\diff z).
\end{align*}
are continuous.
Hence, by \parencite[Proposition 7.33]{bertsekas_stochastic_1996}, there exist two measurable nonnegative functions $\hat{q}^a$ and $\hat{q}^b$ on $[0,T] \times \mathcal{Q}_n \times \bar{\mathcal{Y}}$ such that for all $(t,Q,y) \in [0,T] \times \mathcal{Q}_n \times \bar{\mathcal{Y}}$, $\hat{q}^a(t,Q,y) \leqslant Q + \Qmax$, $\hat{q}^b(t,Q,y) \leqslant -Q + \Qmax$, and
\begin{align*}
    H^a\left(\left(u^R(t,y)\right)_{R \in \mathcal{Q}_n}, y, Q\right) & = \int_{\Qa}e^{-\gamma (\hat{q}^a(t,Q,y) \wedge z)\left(\frac{\delta}{2}-y\right)}u^{Q-\hat{q}^a(t,Q,y) \wedge z}\left(t,y\right) \mu^a(\diff z)\\
    H^b\left(\left(u^R(t,y)\right)_{R \in \mathcal{Q}_n}, y, Q\right) & = \int_{\Qb}e^{-\gamma (\hat{q}^b(t,Q,y) \wedge z)\left(\frac{\delta}{2}+y\right)}u^{Q+\hat{q}^b(t,Q,y) \wedge z}\left(t,y\right) \mu^b(\diff z).
\end{align*}

\subsection{Proof of point (ii)}

Let $(t_i)_{1 \leqslant i \leqslant N}$ be the ordered jump times of $N^a$ and $N^b$ on $[t,T]$ (which are almost surely distinct and different from $t$ and $T$), $(z_i)_{1 \leqslant i \leqslant N}$ the associated marks and $((e_i, s_i))_{1 \leqslant i \leqslant N}$ the random variable that takes the value $(-1,a)$ if the index is associated to a jump of $N^a$ and $(1,b)$ for a jump of $N^b$. For convenience, we set $t_0 = -1$, $t_{N+1} = T$, $z_0=z_{N+1} = 0$, $e_0=e_{N+1} = 0$. Define $Q_0 \defeq Q$ and, recursively, $Q_{i+1} \defeq Q_i + e_{i+1}z_{i+1}\wedge\hat{q}^{s_i}(t_{i+1}, Q_i, \Y_{t_{i+1}})$, which is $\mathcal{F}_{t_{i+1}}$-measurable (by definition of $\hat{q}^{s_i}$, it is valued in $\mathcal{Q}_n$). For $s \in [0,T]$, define
\begin{equation*}
    q^{*a}_s \defeq \sum_{i=0}^N \mathds{1}_{\{t_i < s \leqslant t_{i+1}\}} \hat{q}^a(s \vee t, Q_i, \Y_s) \text{ and } q^{*b}_s \defeq \sum_{i=0}^N \mathds{1}_{\{t_i < s \leqslant t_{i+1}\}} \hat{q}^b(s \vee t, Q_i, \Y_s).
\end{equation*}
The process $q^* = (q^{a*}, q^{*b})$ is in $\mathcal{A}^{t,Q}$. It is easy to check that for $s \in [t_i,t_{i+1})$, $Q^{t,Q,q^*}_s = Q_i$. This proves the existence part.

We prove uniqueness by induction. We fix $\omega \in \Omega$ (omitting the dependence on $\omega$ for concision) such that the $(t_i)_{1 \leqslant i \leqslant N}$ are distinct and different from $t$ and $T$.
Let $q \in \mathcal{A}^{t,Q}$ be a control verifying \eqref{eq:existence_control} on $\omega$. Since $\Q$ is constant on $[0, t_1)$ and equal to $Q$, $q=q^{*}$ on $[0,t_1]$ by \eqref{eq:existence_control}. Suppose that $q = q^{*}$ on $[0,t_i]$ and $\Q = Q^{t,Q,q^*}$ on $[0, t_i)$, $i \leqslant N$. Since $q_{t_i} = q^*_{t_i}$, we have $\Q_{t_i} = Q^{t,Q,q^*}_{t_i}$. As before, this implies that $Q^{t,Q,q}= Q^{t,Q,q^*}$ on $[t_i, t_{i+1})$, which yields $q = q^{*}$ on $(t_i,t_{i+1}]$ thanks to \eqref{eq:existence_control}.

\subsection{A useful lemma}

To prove the verification theorem we need the following lemma. As a byproduct, we get the finiteness of the value function $U$.

\begin{lemma}
    \label{lem:uniform_maj_pnl}
    There exists a nonnegative random variable $M$ such that
    \begin{equation*} 
        e^{-\gamma \PnL_s} \leqslant M\quad\text{and}\quad \E[t,y][M^k] < \infty 
    \end{equation*}
    for all $(t,s,Q,y) \in [0,T] \times [0,T] \times \mathcal{Q}_n \times \mathcal{Y}$ and for all $k \geqslant 0$.
\end{lemma}

\begin{proof}
    Let $(t,s,Q,y) \in [0,T] \times [0,T] \times \mathcal{Q}_n \times \mathcal{Y}$. We have
    \begin{equation*}
        e^{-\gamma \PnL_s} \leqslant \exp\left(\gamma\delta\eta N_T + \gamma\sigma \left|\int_t^{t \vee s} \Q_r \diff W_r \right|\right)
    \end{equation*}
    where $N \defeq N^a(\cdot \times \Qa) + N^b(\cdot \times \Qb)$.

    Let $t_0< t_1 < \dots$ be the jump times of $N$ with $t_0 = 0$ and $t_i = T$ after the last jump. The process $\Q$ remains constant in each interval $[t_i, t_{i+1})$, hence we can write
    \begin{equation*}
        \int_{t}^{s \vee t} \Q_u \diff W_u = \sum_{i = 0}^{\infty}\Q_{t_{i}\vee t} \left( W_{\left(t_{i+1} \wedge s\right)\vee t} - W_{\left(t_i \wedge s \right) \vee t }\right).
    \end{equation*}
   This sum is well-defined (almost surely) because $N_T < \infty$ and therefore only a finite number of terms are non-zero. 
    
    Suppose that $s \geqslant t$ and fix $\omega \in \Omega$. There exist two integers $i_0$ and $i_1$ such that $t \in [t_{i_0}(\omega), t_{i_0+1}(\omega))$ and $s \in [t_{i_1}(\omega), t_{i_1 +1}(\omega))$. Hence,
    \begin{equation*}
        \left(\int_{t}^{s\vee t} \Q_u \diff W_u\right)(\omega) =
        Q \left( W_{t_{i_0+1}} - W_{t}\right)(\omega)
        + \Q_s \left( W_{s} - W_{t_{i_1}}\right)(\omega)
        +\sum_{i = i_0+1}^{i_1-1}\Q_{t_{i}} \left( W_{t_{i+1}} - W_{t_i}\right)(\omega).
    \end{equation*}

    Thus, 
    \begin{equation*}
        \left|\int_{t}^{s \vee t} \Q_u \diff W_u \right|
        \leqslant
        4 \Qmax \sup_{u \in [0,T]} \left|W_u \right|
        + \Qmax
        \sum_{i = 0}^{\infty} \left| W_{t_{i+1}} - W_{t_i}\right|
    \end{equation*}
    and the inequality is trivially verified if $s < t$. Define
    \begin{equation*}
        M \defeq \exp\left(\gamma\delta\eta N_T +  4 \gamma \sigma\Qmax \sup_{u \in [0,T]} \left|W_u \right|
        + \gamma\sigma\Qmax
        \sum_{i = 0}^{\infty} \left| W_{t_{i+1}} - W_{t_i}\right|\right).
    \end{equation*}
    The random variable $M$ does not depend on $(t,s,y,Q)$ and
    \begin{equation*} 
        e^{-\gamma \PnL_s} \leqslant M,\quad (t,s,Q,y) \in [0,T] \times [0,T] \times \mathcal{Q}_n \times \mathcal{Y}.
    \end{equation*}
    Let $k \geqslant 0$ and $(t,y) \in [0,T] \times \mathcal{Y}$. By the three factor Cauchy-Schwarz inequality,
    \begin{equation*}
        \begin{split}
            \E[t,y][M^k] \leqslant \mathbb{E}^{t,y}&\left[\exp\left(3k \gamma \delta\eta N_T\right)\right]^{\frac{1}{3}} \times
            \mathbb{E}^{t,y}\left[\exp\left(12k \gamma\sigma \Qmax \sup_{u \in [0,T]} \left|W_u \right|\right)\right]^{\frac{1}{3}} \\
            &\times \mathbb{E}^{t,y}\left[\exp\left(3k \gamma \sigma \Qmax \sum_{i = 0}^{\infty} \left| W_{t_{i+1}} - W_{t_i}\right|\right)\right]^{\frac{1}{3}}.
        \end{split}
    \end{equation*}
    It is easy to see that the first two factors are finite. For the first factor, it suffices to go back to the probability $\Proba$ and use the fact that $N_T$ is a Poisson random variable under $\Proba$. For the second factor, one can use the reflection principle for Brownian motion.

    We now show that the third factor is also finite. Define $A \defeq 6 k \gamma \sigma \Qmax$. A change of probability measure and the Cauchy-Schwarz inequality yield
    \begin{align*}
        \E[t,y][\exp\left(\frac{A}{2} \sum_{i=0}^{\infty}\left|W_{t_{i+1}} - W_{t_i}\right|\right)]
        & \leqslant e^{2T}\E[][e^{\ln(\Lambda^*)N_T}\exp\left(\frac{A}{2} \sum_{i=0}^{\infty}\left|W_{t_{i+1}} - W_{t_i}\right|\right)] \\
        & \leqslant e^{2T}\E[][e^{2\ln(\Lambda^*)N_T}]^{\frac{1}{2}}\E[][\exp\left(A \sum_{i=0}^{\infty}\left|W_{t_{i+1}} - W_{t_i}\right|\right)]^{\frac{1}{2}}. 
    \end{align*}
    The first two factors of the last expression are finite, hence we only have to prove that it is also the case for the third one. We have
    \begin{align*}
        \E[][\exp\left(A \sum_{i=0}^{\infty}\left|W_{t_{i+1}} - W_{t_i}\right|\right)]
        & \leqslant \sum_{n = 0}^{\infty} \E[][\mathds{1}_{\{N_T = n\}}
        \exp\left(A \sum_{i=0}^{n}\left|W_{t_{i+1}} - W_{t_i}\right|\right)]
        \\
        & \leqslant \sum_{n = 0}^{\infty} \E[][\mathds{1}_{\{N_T = n\}}
        \prod_{i=0}^n\exp\left(A \left|W_{t_{i+1}} - W_{t_i}\right|\right)].
    \end{align*}
Since the time $t_i$ is $\sigma(N^a, N^b)$-measurable, the factors are independent conditionally on $(N^a, N^b)$. Hence,
\begin{equation}\label{lem5:1}
    \E[][\exp\left(A \sum_{i=0}^{\infty}\left|W_{t_{i+1}} - W_{t_i}\right|\right)]
         \leqslant \sum_{n = 0}^{\infty} \mathbb{E}\left[\mathds{1}_{\{N_T = n\}}
        \prod_{i=0}^n \mathbb{E} \left[ \exp\left(A \left|W_{t_{i+1}} - W_{t_i}\right|\right)|N^a,N^b\right]\right].
\end{equation}
For fixed $0 \leqslant s \leqslant t \leqslant T$,
\begin{equation}\label{lem5:2}
    \E[][\exp(A |W_t-W_s|)] = 2 \exp\left(\frac{A^2}{2}(t-s)\right)
    \mathcal{N}\left(A\sqrt{t-s}\right)
    \leqslant 2 \exp\left(\frac{A^2}{2}(t-s)\right)
\end{equation}
where $\mathcal{N}$ is the cumulative distribution function of the standard normal law. Recalling that $W$ is independent of $(N^a,N^b)$, combining the inequalities \eqref{lem5:1} and \eqref{lem5:2} yields
\begin{align*}
    \E[][\exp\left(A \sum_{i=0}^{\infty}\left|W_{t_{i+1}} - W_{t_i}\right|\right)]
        & \leqslant \sum_{n = 0}^{\infty} \mathbb{E}\left[\mathds{1}_{\{N_T = n\}}
        2^n \exp\left(\frac{A^2}{2}T \right)\right] \\
        & \leqslant \sum_{n = 0}^{\infty} \exp\left(\frac{A^2}{2}T \right)
        e^{- 2T} \frac{\left(4 T\right)^n}{n!} \\
        & \leqslant \exp\left(\frac{A^2}{2}T + 2T\right) < \infty.
\end{align*}
\end{proof}

\subsection{Proof of points (iii) and (iv)}

Let $(t,y,Q) \in [0,T) \times \mathcal{Y} \times \mathcal{Q}_n$ and $q \in \mathcal{A}^{t,Q}$. We denote by $t_1< \dots < t_{N^{ab}_T}$ the jump times of $N^{ab} \defeq N^a(\cdot \times \Qa) + N^b(\cdot \times \Qb)$. We denote by $\tau_1 <\dots <\tau_{N-1}$ the jump times of $\Y$ ($N$ is a random variable), and define $\tau_0= t$, $\tau_N = T$. Since the $t_i$'s are totally inaccessible and the $\tau_i$'s are predictable (they are the limit of $(\tau_i^{(k)})_k$ defined in \eqref{eq:tauui_predictable} below), we have almost surely
\begin{equation}\label{eq:disjoint_jumps}
    \{t_i: i \in {1\dots,N^{ab}_T}\} \cap \{\tau_i: i \in \{1\dots,N\}\} = \emptyset.
\end{equation}
%(this remark will be used in this proof).

Let $k_0 \in \mathbb{N}^*$ be such that $\min\{2\delta \eta, y + \bar{y}, \bar{y} - y\} < \frac{1}{k_0}$. Fix $k \geqslant k_0$ and define for $i \leqslant N-1$,
\begin{equation}\label{eq:tauui_predictable}
    \tau_i^{(k)} \defeq \inf\left\{s \in (\tau_i, \tau_{i+1}) : \Y_s \in\left\{-\bar{y} + \frac{1}{k}, \bar{y} - \frac{1}{k}\right\}\right\}\wedge\left(T - \frac{1}{k}\right).
\end{equation}
By the intermediate value theorem, we have that for all $k \geqslant k_0$ and $i\leqslant N-1$, $\tau_i \leqslant \tau_i^{(k)} \leqslant \tau_{i+1}$. Also, for all $s \in [\tau_i, \tau_{i}^{(k)}]$, $(s,\Y_s) \in \left[0, T - \frac{1}{k}\right] \times \left[-\bar{y} + \frac{1}{k}, \bar{y} - \frac{1}{k}\right]$. Furthermore, for each $i$, $(\tau_i^{(k)})_{k \geqslant k_0}$ is nondecreasing. One can show that for all $i \leqslant N-1$,
\begin{equation*}
    \tau_i^{(k)} \xrightarrow[k \to \infty]{} \tau_{i+1}.
\end{equation*}

Fix $k \geqslant k_0$ and define $A  \defeq - e^{-\gamma \left(\PnL_T - \ell\left(\Q_T\right)\right)} - u^Q(t,y)$. Then,
\begin{equation*}
    \begin{split}
     A =
    \sum_{i=0}^{N-1}& \left(e^{-\gamma \PnL_{\tau^{(k)}_i}}u^{\Q_{\tau^{(k)}_i}}\left(\tau^{(k)}_{i}, \Y_{\tau^{(k)}_{i}}\right) -  e^{-\gamma\PnL_{\tau_i}}u^{\Q_{\tau_i}}(\tau_i, \Y_{\tau_i})\right) \\
    &+
    \sum_{i=0}^{N-1} \left( e^{-\gamma\PnL_{\tau_{i+1}}}u^{\Q_{\tau_{i+1}}}\left(\tau_{i+1}, \Y_{\tau_{i+1}}\right) - e^{-\gamma \PnL_{\tau^{(k)}_i}}u^{\Q_{\tau^{(k)}_i}}\left(\tau^{(k)}_{i}, \Y_{\tau^{(k)}_{i}}\right)\right).
    \end{split}
\end{equation*}
Using Itô's formula with jumps on each interval $[\tau_i, \tau_i^{(k)}]$ (on which $\Y_s- \Y_{\tau_i} = \sigma\left(W_s - W_{\tau_i}\right)$), we get
\begin{equation}\label{eq:proofverification1}
    \begin{split}
        A = 
        &\sum_{i=0}^{N-1} \left( e^{-\gamma\PnL_{\tau_{i+1}}}u^{\Q_{\tau_{i+1}}}\left(\tau_{i+1}, \Y_{\tau_{i+1}}\right) - e^{-\gamma \PnL_{\tau^{(k)}_i}}u^{\Q_{\tau^{(k)}_i}}\left(\tau^{(k)}_{i}, \Y_{\tau^{(k)}_{i}}\right)\right)\\
        &+\sum_{i=0}^{N-1}
        \int_{\tau_i}^{\tau_i^{(k)}} e^{-\gamma \PnL_r} \left(\partial_t u^{\Q_r}(r, \Y_r) + \frac{\sigma^2}{2} \partial^2_{yy} u^{\Q_r}(r, \Y_r) \right)\diff r\\
        &+\sum_{i=0}^{N-1}
        \sigma\int_{\tau_i}^{\tau_i^{(k)}} e^{-\gamma \PnL_r}\partial_y u^{\Q_r}(r, \Y_r) \diff W_r \\
        &+\sum_{i=0}^{N-1}
        \frac{1}{2}\int_{\tau_i}^{\tau_i^{(k)}} \gamma^2 \sigma^2 \left(\Q_r\right)^2 e^{-\gamma \PnL_r} u^{\Q_r}(r, \Y_r) \diff r \\
        &-\sum_{i=0}^{N-1}
        \int_{\tau_i}^{\tau_i^{(k)}} \gamma \sigma^2 \Q_r e^{-\gamma \PnL_r} \partial_y u^{\Q_r}(r, \Y_r) \diff r \\
        &-\sum_{i=0}^{N-1}\int_{\tau_i}^{\tau_i^{(k)}} \gamma \sigma \Q_r e^{-\gamma \PnL_r} u^{\Q_r}(r, \Y_r) \diff W_r\\
        &+ \sum_{i=0}^{N-1}
        \int_{(\tau_i,\tau_i^{(k)}] \times \Qa} e^{-\gamma \PnL_{r-}}\left[e^{\gamma (q^{a}_r \wedge z) \left(\frac{\delta}{2} - Y^{t,y}_r\right)}u^{\Q_{r-} - (q^{a}_r \wedge z)}(r, \Y_r) - u^{\Q_{r-}}(r, \Y_r)\right] N^a(\diff u, \diff z) \\
        &+ \sum_{i=0}^{N-1}
        \int_{(\tau_i,\tau_i^{(k)}] \times \Qb} e^{-\gamma \PnL_{r-}}\left[e^{\gamma (q^{b}_r \wedge z)\left(\frac{\delta}{2} + Y^{t,y}_r\right)}u^{\Q_{r-} + (q^{b}_r \wedge z)}(r, \Y_r) - u^{\Q_{r-}}(r, \Y_r)\right] N^b(\diff u, \diff z).
        \end{split}
\end{equation}

The integrals in $\diff W$ have zero expectation (under $\mathbb{P}^{t,y}$). Indeed, since $\sup_{R \in \mathcal{Q}_n}|u^R|_{2+\beta} < \infty$, $u^Q$ and $\partial_y u^Q$ are bounded in the compact set $\left[0,T-\frac{1}{k}\right] \times \left[-\bar{y} + \frac{1}{k}, \bar{y}-\frac{1}{k}\right]$, uniformly in $Q \in \mathcal{Q}_n$, by some constant $m > 0$. Then, letting $M$ be the random variable given by Lemma \ref{lem:uniform_maj_pnl},
\begin{align*}
    & \mathbb{E}^{t,y}\left[\int_0^T \left(\sum_{i=0}^{N-1}\mathds{1}_{\left[\tau_i,\tau_i^{(k)}\right]}(r) e^{-\gamma \PnL_r}\partial_y u^{\Q_r}(r, \Y_r)\right)^2 \diff r\right]
    \leqslant Tm^2\E[t,y][M^2] < \infty,\\
    & \mathbb{E}^{t,y}\left[\int_0^T \left(\sum_{i=0}^{N-1}\mathds{1}_{\left[\tau_i,\tau_i^{(k)}\right]}(r)\gamma\sigma\Q_r e^{-\gamma \PnL_r}u^{\Q_r}(r, \Y_r)\right)^2 \diff r\right]
    \leqslant \gamma^2\sigma^2\Qmax^2 Tm^2\E[t,y][M^2] < \infty.
\end{align*}

Taking expectation in \eqref{eq:proofverification1}, and using the fact that $(u^Q)_{Q \in \mathcal{Q}_n}$ solves \eqref{eq:HJB_interior}, yields
\begin{equation}\label{eq:proofverification2}
    \begin{split}
        J(t,Q,&y,q) - u^Q(t,y)\\
        & = \E[t,y][\sum_{i=0}^{N-1} \left( e^{-\gamma\PnL_{\tau_{i+1}}}u^{\Q_{\tau_{i+1}}}\left(\tau_{i+1}, \Y_{\tau_{i+1}}\right) - e^{-\gamma \PnL_{\tau^{(k)}_i}}u^{\Q_{\tau^{(k)}_i}}\left(\tau^{(k)}_{i}, \Y_{\tau^{(k)}_{i}}\right)\right)] \\
        & + \mathbb{E}^{t,y}\left[\sum_{i=0}^{N-1}
        \int_{\tau_i}^{\tau_i^{(k)}} \int_{\Qa} e^{-\gamma \PnL_{r-}}\left[e^{\gamma (q^{a}_r \wedge z) \left(\frac{\delta}{2} - Y^{t,y}_r\right)}u^{\Q_{r-} - (q^{a}_r \wedge z)}(r, \Y_r) \right] \Lambda^a(\Y_r) \mu^a(\diff z) \diff r\right] \\
        & + \mathbb{E}^{t,y}\left[\sum_{i=0}^{N-1}
        \int_{\tau_i}^{\tau_i^{(k)}} \int_{\Qb} e^{-\gamma \PnL_{r-}}\left[e^{\gamma (q^{b}_r \wedge z) \left(\frac{\delta}{2} + Y^{t,y}_r\right)}u^{\Q_{r-} + (q^{b}_r \wedge z)}(r, \Y_r) \right] \Lambda^b(\Y_r)\mu^b(\diff z)\diff r\right] \\
        & -\mathbb{E}^{t,y}\left[
            \sum_{i=0}^{N-1}
            \int_{\tau_i}^{\tau_i^{(k)}}
            e^{-\gamma \PnL_{r}} H^a\left(\left(u^R\right)_{R \in \mathcal{Q}_n}, \Q_r, \Y_r\right)\Lambda^a(\Y_r)\diff r
        \right] \\
        & - \mathbb{E}^{t,y}\left[
            \sum_{i=0}^{N-1}
            \int_{\tau_i}^{\tau_i^{(k)}}
            e^{-\gamma \PnL_{r}}H^b\left(\left(u^R\right)_{R \in \mathcal{Q}_n}, \Q_r, \Y_r\right)\Lambda^b(\Y_r)\diff r
        \right].
    \end{split}
\end{equation}

The term inside the expectation in the first term is bounded by $2\sup\limits_{R \in \mathcal{Q}_n}|u^{R}|_{\infty}N M$. The Cauchy-Schwarz inequality and Lemma \ref{lem:finite_activity_expectation} imply that
\begin{equation*}
    \E[t,y][2\sup\limits_{R \in \mathcal{Q}_n}|u^{R}|_{\infty}N M] \leqslant
    2\sup\limits_{R \in \mathcal{Q}_n}|u^{R}|_{\infty} \sqrt{\E[t,y][N^2]} \sqrt{\E[t,y][M^2]} < \infty.
\end{equation*} 
Thus, we can use the dominated converge theorem in the first expectation. By \eqref{eq:disjoint_jumps}, almost surely, for all $i$, $\PnL_{\tau_i^{(k)}} \xrightarrow[k \to \infty]{} \PnL_{\tau_{i+1}}$ and $\Q_{\tau_i^{(k)}} \xrightarrow[k \to \infty]{} \Q_{\tau_{i+1}}$. Combined with the continuity conditions \eqref{eq:HJB_border}, we conclude that the first term tends to 0 when $k \to \infty$.

The four other expectations in \eqref{eq:proofverification2} have integrands bounded by $\Lambda^* M e^{\gamma (\qa \vee \qb) \delta (1+\eta)} \sup_{R \in \mathcal{Q}_n}|u^{R}|_{\infty}$ which has finite expectation, therefore we can again use the dominated convergence theorem to obtain
\begin{equation}\label{eq:proofverification3}
    \begin{split}
        J(t,Q,y,q) &- u^Q(t,y)=\\
        &\mathbb{E}^{t,y}\left[
        \int_{t}^{T} \int_{\Qa} e^{-\gamma \PnL_{r}}\left[e^{\gamma (q^{a}_r \wedge z) \left(\frac{\delta}{2} - Y^{t,y}_r\right)}u^{\Q_{r} - (q^{a}_r \wedge z)}(r, \Y_r) \right] \Lambda^a(\Y_r) \mu^a(\diff z) \diff r\right] \\
        & + \mathbb{E}^{t,y}\left[
        \int_{t}^{T} \int_{\Qb} e^{-\gamma \PnL_{r}}\left[e^{\gamma (q^{b}_r \wedge z) \left(\frac{\delta}{2} + Y^{t,y}_r\right)}u^{\Q_{r} + (q^{b}_r \wedge z)}(r, \Y_r) \right] \Lambda^b(\Y_r)\mu^b(\diff z)\diff r\right] \\
        & -\mathbb{E}^{t,y}\left[
            \int_{t}^{T}
            e^{-\gamma \PnL_{r}} H^a\left(\left(u^R\right)_{R \in \mathcal{Q}_n}, \Q_r, \Y_r\right)\Lambda^a(\Y_r)\diff r
        \right] \\
        & - \mathbb{E}^{t,y}\left[
            \int_{t}^{T}
            e^{-\gamma \PnL_{r}}H^b\left(\left(u^R\right)_{R \in \mathcal{Q}_n}, \Q_r, \Y_r\right)\Lambda^b(\Y_r)\diff r
        \right].
    \end{split}
\end{equation}
Hence, by the definition of the Hamiltonians, $J(t,Q,y,q) \leqslant u^Q(t,y)$ and $U(t,Q,y) \leqslant u^{Q}(t,y)$. Furthermore, with $q = q^*$, $u^Q(t,y) = J(t,Q,y,q^*)$ and point (iii) follows. Point (iv) is also a direct consequence of equation \eqref{eq:proofverification3} since $\Lambda^a, \Lambda^b >0$.

\section{Existence of a solution to the Hamilton-Jacobi-Bellman equation}
\label{sec:hjb}
\begin{dummyenv}
\renewcommand{\Q}{\mathcal{Q}_n}
\renewcommand{\Y}{\mathcal{Y}}

This section is devoted to the proof of Theorem \ref{thm:existence}(i), namely the existence of a classical solution of the HJB equation \eqref{eq:HJB_interior} with boundary conditions \eqref{eq:HJB_border}. The proof of Theorem \ref{thm:existence}(ii) is postponed to Appendix \ref{sec:convergence_discrete_continuous}. For this section, we define $\D \defeq C([0,T] \times \bar{\Y}) \cap C^{1,2}((0,T] \times \mathcal{Y})$ and let $n \in \mathbb{N}^* \cup \{\infty\}$.

We now fix some constants for the whole section. Let $\qmax \defeq \max\{\qa, \qb\}$. For $\beta \leqslant \alpha$, define $\Lambda_{\beta}^* \defeq \max\left\{|\Lambda^a|_{\beta},  |\Lambda^b|_{\beta}\right\}$. Let $K_0 \defeq \max \left\{\frac{\sigma^2}{2}, \frac{2}{\sigma^2}, \sigma^2 \gamma \Qmax,\frac{\sigma^2 \gamma^2 \Qmax^2}{2} + 2 \Lambda^*\right\}$, and $C_0 > 0$ and $\beta_0 \in (0,1)$ be constants given by 
Corollary \ref{corol:krylovsafonov} (Krylov-Safonov) for coefficients bounded by $K_0$ and coefficient in front of $\partial^2_{yy}$ greater than $K_0^{-1}$. We fix $\beta \defeq \alpha \wedge \beta_0$, and let $K_1 \defeq \max\{1, T, \bar{y}, \bar{y}^2\} \cdot \max\left\{\frac{\sigma^2}{2}, \sigma^2 \gamma \Qmax,\frac{\sigma^2 \gamma^2 \Qmax}{2} + 2 \Lambda^* + 2 \Lambda^*_{\beta}\right\}$ and $C_1 > 0$ be a constant given by Corollary \ref{corol:linear_existence} with $a = -\bar{y}$, $b = \bar{y}$, $a_0 = y_-$, $b_0 = y_+$, $\delta = \beta$, $\rho_1=\rho_2 = 1$, $ \alpha \equiv \frac{\sigma^2}{2}$.

We denote by $\ell^*$ the maximum of $\ell$. If $n = \infty$, let $\varpi$ be a finite modulus of continuity of $\ell$. We suppose that for $\Delta \in [0, \infty)$, $\varpi(\Delta) \geqslant \Delta$, otherwise we can replace $\varpi(\Delta)$ by $\max\{\Delta, \varpi(\Delta)\}$.

In Subsection \ref{subsec:discrete_case}, we show using a contraction argument the existence of a uniformly bounded family of functions $(u_Q)_{Q \in \Q}$ in $\D$ verifying \eqref{eq:HJB_interior}-\eqref{eq:HJB_border} if $n < \infty$ or $\mu^a$ and $\mu^b$ have finite support. This way, there is no measurability issue with the definition of the Hamiltonians. We also show that there exists a constant $C$, depending on $T, \Qmax, \bar{y}, \sigma, \gamma, \Lambda^*, \beta, \Lambda_{\beta}, \ell^*$ but not on $n$, $\mu^a$ or $\mu^b$ such that $\sup_{Q \in \Q} |u^Q|_{2+\beta} \leqslant C$.

In Subsection \ref{subsec:continuity}, we show that, for $n=\infty$ and $\mu^a$, $\mu^b$ having finite support, any family of functions $(u_Q)_{Q \in \mathcal{Q}_{\infty}}$ in $\D$ verifying \eqref{eq:HJB_interior}-\eqref{eq:HJB_border} with $\sup_{Q \in \mathcal{Q}_{\infty}}|u_n|_{\beta} < \infty$ is necessarily continuous on $[0,T] \times \mathcal{Q}_{\infty} \times \bar{\Y}$.

In Subsection \ref{subsec:existence_general_case}, we show, for $n=\infty$ and any $\mu^a$, $\mu^b$, the existence of a solution of \eqref{eq:HJB_interior}-\eqref{eq:HJB_border}, as a limit of solutions associated to finitely supported measures.

\subsection{Discrete case}
\label{subsec:discrete_case}

We suppose that $n < \infty$ or $\mu^a$ and $\mu^b$ have finite support. We build a sequence of solutions to an iterated linear problem. We then show that it converges uniformly towards some family of continuous functions $(u^Q)_{Q \in \Q}$. Finally, we derive the regularity of $(u^Q)_{Q \in \Q}$ and show that it is a solution of \eqref{eq:HJB_interior}-\eqref{eq:HJB_border}.

\subsubsection{Construction of a sequence of solutions to an iterated linear problem}

For all $(t,y,Q) \in [0,T] \times \bar{\Y} \times \Q$, we define $u^Q_0(t,y) \defeq -e^{\gamma \ell(Q)}$. Let $k \in \mathbb{N}$ and suppose we have built a family $(u_k^Q)_{Q \in \Q}$ in $\D$ such that $\sup_{Q \in \Q}|u^Q_k|_{2+\beta} < \infty$.

Let $Q \in \Q$. For $(t,y) \in [0,T) \times \Y$, define
\begin{equation}
    \label{eq:definition_fk}
    f^Q_k(t,y) \defeq \Lambda^a(t,y)H^a\left((u^Q_k(t,y))_{Q\in\Q}, Q, y, \mu^a\right) + \Lambda^b(t,y)H^b\left((u^Q_k(t,y))_{Q\in\Q}, Q, y, \mu^b\right).
\end{equation}
For $(P, P') \in ([0,T) \times \Y)^2$
\begin{equation*}
    \sup_{R \in \Q} \left|u^R_k(P) - u^R_k(P')\right| \leqslant \frac{d(P,P')}{d_{PP'}^{\beta}}\sup_{R \in \Q} |u^R_k|_{\beta} < \infty.
\end{equation*}
Hence, by Lemmas \ref{lem:H_ineq_ty} and \ref{lem:H_ineq_w}, $|f_k^Q|_{\beta} \leqslant C \sup_{R \in \Q}|u^R_k|_{\beta} < \infty$ for some constant $C$ depending only on $T$, $\gamma$, $\delta$, $\eta$, $\Qmax$, $\Lambda^*$, $\Lambda^*_{\beta}$. Hence, by Corollary \ref{corol:linear_existence}, there exists a unique function $u_{k+1}^Q \in \D$ solving, for $(t,y) \in [0,T) \times \Y$,
\begin{equation}
    \label{eq:hjb_iter}
    0 = \partial_t u^Q_{k+1}(t,y) + \frac{\sigma^2}{2}\partial^2_{yy} u^Q_{k+1}(t,y)
        -\sigma^2 \gamma Q \partial_y u^Q_{k+1}(t,y) + \left(\frac{\sigma^2 \gamma^2 Q^2}{2}- \left(\Lambda^a + \Lambda^b\right)(y)\right)u^Q_{k+1}(t,y)+
            f^Q_k(t,y)
\end{equation}
and for $(t,y) \in [0,T] \times \bar{\Y}$,
\begin{equation*}
    \left\{
        \begin{array}{ll}
            u_{k+1}^Q(T, y) &=-e^{\gamma\ell(Q)} \\
            u_{k+1}^Q\left(t, \bar{y}\right) &= u_{k+1}^Q\left(t, y_+\right) \\
            u_{k+1}^Q\left(t, -\bar{y}\right) &= u_{k+1}^Q\left(t, y_-\right).
        \end{array}
    \right.
\end{equation*}
Furthermore,
\begin{equation}
    \label{eq:maj_holder_norm_uk}
    |u^Q_{k+1}|_{2+\beta} \leqslant C_1(e^{\gamma \ell^*} + |f_k^Q|_{\beta}) \leqslant C_1 \left(e^{\gamma \ell^*} + C\sup_{R \in \Q}|u^R_{k}|_{\beta}\right) < \infty.
\end{equation}

The construction follows iteratively.

\subsubsection{Convergence of the sequence}

We have that, for $k\geqslant 1$ and $Q\in\Q$, $v \defeq u^Q_{k+1} - u^{Q}_k$ solves 
\begin{equation*}
    \begin{aligned}
        0 =& \left(\partial_t v + \frac{\sigma^2}{2}\partial^2_{yy} v
        -\sigma^2 \gamma Q \partial_y v + \left(\frac{\sigma^2 \gamma^2 Q^2}{2}- \left(\Lambda^a + \Lambda^b\right)(y)\right)v+
            f^Q_k - f^Q_{k-1}\right)(t,y)
    \end{aligned}
\end{equation*}
with boundary conditions
\begin{equation*}
    \left\{
        \begin{array}{ll}
            v(T, y) &=0 \\
            v\left(t, \bar{y}\right) &= v\left(t, y_+\right) \\
            v\left(t, -\bar{y}\right) &= v\left(t, y_-\right).
        \end{array}
    \right.
\end{equation*}

By Lemma \ref{lem:H_ineq_w}, there exists a constant $C$ depending only on $\gamma$, $\delta$, $\eta$, $\Qmax$, $\Lambda^*$ such that for all $(t,y) \in [0,T) \times \Y$, $\left|f^Q_{k}(t,y) - f^Q_{k-1}(t,y)\right| \leqslant C \sup_{R \in \Q} \left|u^R_{k}(t,y) - u^R_{k-1}(t,y)\right|$. Define
\begin{equation}
    \label{eq:g_k}
    g_k(t) \defeq e^{ct} \sup_{\substack{R \in \Q \\ y \in \bar{\Y}}}\left|u^R_{k+1}(t,y) - u^R_{k}(t,y)\right|,\quad k\geqslant 0,\, t \in [0,T]
\end{equation}
where $c \defeq \frac{\sigma^2\gamma^2\Qmax^2}{2} + 2\Lambda^*$. The function $g_k$ is measurable (see Lemma \ref{lem:meas_gk}), hence, by Corollary \ref{corol:bound_on_u},
\begin{equation*}
    g_k(t) \leqslant C\int_t^T g_{k-1}(s) \diff s, \quad t \in [0,T].
\end{equation*}
By induction, we get, for all $k \in \mathbb{N}^*$ and $t \in [0,T]$,
\begin{equation*}
    g_k(t) \leqslant C^{k-1}\frac{(T-t)^{k-1}}{(k-1)!}\sup_{t' \in [0,T]}g_1(t')
    \leqslant 3C_1e^{\gamma \ell^*} C^{k-1}\frac{(T-t)^{k-1}}{(k-1)!}.
\end{equation*}.
Thus, there exists a finite constant $C' > 0$, depending only on $T$, $\Qmax$, $\gamma$, $\sigma$, $\delta$, $\eta$, $\Lambda^*$, $\ell^*$ such that
\begin{equation}\label{eq:cauchy_sum}
    \sum_{k=0}^{\infty}\sup_{R\in \Q}|u^R_{k+1} - u^R_k|_{\infty} \leqslant \sum_{k=0}^{\infty}\sup_{t \in [0,T]} g_k(t) \leqslant C' < \infty.
\end{equation}

We conclude that for $Q \in \Q$, $(u^Q_k)_{k \in \mathbb{N}}$ is a Cauchy sequence in $C([0,T] \times \bar{\Y})$, and consequently it converges (with respect to the sup norm) towards a continuous function $u^Q$. The convergence is also uniform in $Q$, i.e.
\begin{equation*}
    \sup_{Q \in \Q}|u^Q - u_k^{Q}|_{\infty} \xrightarrow[k \to \infty]{} 0.
\end{equation*}

\subsubsection{Uniform estimates of the Hölder norms of the sequence}
\label{subsubsec:uniform_majoration}

Let $k \in \mathbb{N}^*$ and $Q \in \Q$. By the Krylov-Safonov estimate in Corollary \ref{corol:krylovsafonov}, we have
\begin{equation*}
    |u_k^Q|_{\beta} \leqslant C_0 (|u_k^Q|_{\infty} + |f_{k-1}^Q|_{\infty})
    \leqslant C_0 \left(|u_k^Q|_{\infty} + 2\Lambda^*e^{\gamma \delta(\eta+1)\qmax}\sup_{R\in \Q}|u^R_{k-1}|_{\infty}\right) \leqslant C'' \sup_{l \in \mathbb{N}} \sup_{R\in \Q} |u^R_l|_{\infty}
\end{equation*}
where $C'' > 0$ is a constant depending only on $T$, $\Qmax$, $\gamma$, $\sigma$, $\delta$, $\eta$, $\Lambda^*$. From \eqref{eq:cauchy_sum}, it is immediate that
\begin{equation*}
    |u_k^Q|_{\beta} \leqslant C''(e^{\gamma \ell^*} + C').
\end{equation*}
The inequality \eqref{eq:maj_holder_norm_uk} then yields
\begin{equation*}
    \sup_{k \in \mathbb{N}} \sup_{Q \in \Q}|u^Q_k|_{2+\beta} \leqslant C''' < \infty
\end{equation*}
for some finite constant $C''' > 0$ depending on $T$, $\Qmax$, $\gamma$, $\sigma$, $\delta$, $\eta$, $\Lambda^*$, $\beta$, $\Lambda^*_{\beta}$, $\ell^*$ but not on $\mu^a$ and $\mu^b$.

\subsubsection{Regularity of the limit and existence of a solution}

Let $Q \in \mathcal{Q}_n$. Passing to the limit, we have that $u^Q$ verifies the boundary conditions \eqref{eq:HJB_border}.

Let $K$ be a compact included in $[0,T) \times \mathcal{Y}$. The sequences $\left(\partial_t u_k^Q\right)_{k \in \mathbb{N}^*}$, $\left(\partial_y u_k^Q\right)_{k \in \mathbb{N}^*}$ and $\left(\partial^2_{yy} u_k^Q\right)_{k \in \mathbb{N}^*}$ are uniformly bounded and equicontinuous functions (because they are $\beta$-Hölder continuous with uniformly bounded Hölder constant) on $K$. Hence, by the Arzelà-Ascoli theorem, there exists an increasing sequence $(n_i)_i$ such that $\left(\partial_t u_{n_i}^Q\right)_{i}$, $\left(\partial_y u_{n_i}^Q\right)_{i}$ and $\left(\partial^2_{yy} u_{n_i}^Q\right)_{i}$ converge uniformly on $K$. Since $\left(u_{n_i}^Q\right)_{i}$ converges uniformly towards $u^Q$, then $u^Q$ is $C^{1,2}$ on $K$ and $\left(\partial_t u_{n_i}^Q\right)_{i}$, $\left(\partial_y u_{n_i}^Q\right)_{i}$ and $\left(\partial^2_{yy} u_{n_i}^Q\right)_{i}$ converge towards $\partial_t u^Q$, $\partial_y u^Q$ and $\partial^2_{yy} u^Q$, respectively. In addition, $\sup\limits_{Q \in \mathcal{Q}_n}|u^Q|_{2+\beta} \leqslant C'''<\infty$.

Let $(t,y) \in K$ and define
\begin{equation}
    \label{eq:definition_f}
    f^Q(t,y) \defeq \Lambda^a(t,y)H^a\left((u^Q(t,y))_{Q\in\Q}, Q, y, \mu^a\right) + \Lambda^b(t,y)H^b\left((u^Q(t,y))_{Q\in\Q}, Q, y, \mu^b\right)
\end{equation}
By Lemma \ref{lem:H_ineq_w}, there exists a constant $C$, independent of $k$, such that
\begin{equation*}
    \left|f^Q(t,y) - f^Q_k(t,y)\right| \leqslant C \sup_{R \in \Q} |u - u_k|_{\infty}.
\end{equation*}
Since the right-hand side tends to zero as $k$ tends to infinity, then $\lim\limits_{k\to \infty} f_k^Q(t,y) = f^Q(t,y)$. Passing to the limit over $(n_i)_i$ in equation \eqref{eq:hjb_iter} shows that $u^Q(t,y)$ verifies \eqref{eq:HJB_interior}. This concludes the proof of Theorem \ref{thm:existence}(i) with the exception of the continuity with respect to $Q$ and the case with general $\mu^a$ and $\mu^b$.

\subsection{Continuity with respect to the inventory}
\label{subsec:continuity}

\renewcommand{\Q}{\mathcal{Q}_{\infty}}

In this section, we suppose $n = \infty$, and $\mu^a$, $\mu^b$ have finite support. Let $(u^Q)_{Q \in \mathcal{Q}_{\infty}}$ be a solution of \eqref{eq:HJB_interior}-\eqref{eq:HJB_border} such that $\sup_{Q \in \mathcal{Q}_{\infty}}|u^Q|_{\beta} < \infty$. Define $M \defeq \sup_{Q \in \mathcal{Q}_{\infty}}|u^Q|_{\infty}$.

The goal of this section is to prove that $(t,Q,y)\mapsto u^Q(t,y)$ is continuous. For $\Delta \geqslant 0$, $t \in [0,T]$, we define
\begin{equation}
    \label{eq:modulus_cty_Q}
    m^{\Delta}(t) \defeq \sup \left\{u^Q(t,y) - u^{Q'}(t,y): y \in \bar{\mathcal{Y}}, Q,Q' \in \mathcal{Q}_{\infty}, |Q'-Q| \leqslant \Delta\right\}.
\end{equation}
We want to show that $\sup_t m^{\Delta}(t) \xrightarrow[\Delta \to 0]{} 0$. This implies the desired result, since for a fixed $Q$, $u^Q$ is continuous. To this end, we will use Grönwall's inequality on $m^D$, which is measurable by Lemma \ref{lem:meas_md}.

For $Q \in \Q$, we define $f^Q$ as in \eqref{eq:definition_f}. Let $(\Omega, \mathcal{F}, (\mathcal{F}_t)_t, \mathbb{P})$ be a probability space supporting a Brownian motion $W$. Let $Q \in \Q$ and $(t,y)\in [0,T] \times \mathcal{Y}$. We define
\begin{equation*}
    Z \defeq \exp \left(-\sigma \gamma Q W_T - \frac{1}{2}\sigma^2 \gamma^2 Q^2 T\right) \text{ and } \tilde{W}_r \defeq W_r + \sigma \gamma Q r,\quad r \in [0,T].
\end{equation*}
By Novikov's criterion, $\frac{\diff \mathbb{Q}}{\diff \mathbb{P}} =Z$ defines a probability measure on $(\Omega, \mathcal{F}, (\mathcal{F}_t)_t)$ and, by Girsanov's theorem, $\tilde{W}$ is a Brownian motion under $\mathbb{Q}$. Hence, by Proposition \ref{prop:feynmankac},
\begin{align*}
    u^Q(t,y) &= \mathbb{E}^{\mathbb{Q}}\left[
        - e^{\gamma \ell(Q)} \beta_t^Q + \int_t^T f^Q\left(s,Y\left(t,y,s,(\sigma\tilde{W}_r - \sigma^2 \gamma Q r)_r\right)\right) \beta_s^Q \diff s \right] \\
    &= \mathbb{E}^{\mathbb{Q}}\left[
        - e^{\gamma \ell(Q)} \beta_t^Q + \int_t^T f^Q\left(s,Y\left(t,y,s,\sigma W\right)\right) \beta_s^Q \diff s \right]
\end{align*}
where
\begin{align*}
    \beta_t^Q &\defeq
    \exp\left(\int_t^T \left(\frac{\sigma^2 \gamma^2 Q^2}{2} - \left(\Lambda^a + \Lambda^b\right)\left(s,Y(t,y,s,(\sigma\tilde{W}_r - \sigma^2 \gamma Q r)_r)\right)\right)\diff s\right)\\
     & =
    \exp\left(\int_t^T \left(\frac{\sigma^2 \gamma^2 Q^2}{2} - \left(\Lambda^a + \Lambda^b\right)\left(s, Y(t,y,s,\sigma W)\right)\right)\diff s\right).
\end{align*}

%From now on, we do not write to which probability the expectation refers to as we only work with probability $\mathbb{P}$. 
We have
\begin{equation*}
    u^Q(t,y) = \E[][-e^{\gamma \ell(Q)}\beta_t^Q Z + Z \int_t^T f^Q\left(s,Y\left(t,y,s,\sigma W\right)\right) \beta_s^Q \diff s].
\end{equation*}

Let $Q' \in \mathcal{Q}_{\infty}$ and define $Z' \defeq \exp \left(-\sigma\gamma Q' W_T - \frac{1}{2}\sigma^2 \gamma^2 (Q')^2 T\right)$. By a similar argument, we have
\begin{equation*}
    u^{Q'}(t,y) = \E[][-e^{\gamma \ell(Q')}\beta_t^{Q'} Z' + Z' \int_t^T f^{Q'}\left(s,Y\left(t,y,s,\sigma W\right)\right) \beta_s^{Q'} \diff s].
\end{equation*}
Using triangular inequalities, we get
\begin{equation}\label{eq:contQ1}
    \left|u^Q(t,y) - u^{Q'}(t,y) \right|
    \leqslant A + B + C + D + E + F
\end{equation}
where
\begin{align*}
    A & = \left|e^{\gamma \ell(Q)} - e^{\gamma \ell(Q')}\right|\mathbb{E}\left[\beta^Q_t Z\right] \\
    B & = e^{\gamma \ell(Q')}\mathbb{E}\left[\left|\beta_t^Q - \beta_t^{Q'}\right|Z\right] \\
    C & = e^{\gamma \ell(Q')}\mathbb{E}\left[\beta_t^{Q'}|Z - Z'|\right] \\
    D & = \mathbb{E}\left[Z \int_t^T f^Q\left(s,Y\left(t,y,s,\sigma W\right)\right) \left|\beta^Q_s - \beta^{Q'}_s \right| \diff s \right] \\
    E & = \mathbb{E}\left[\left|Z - Z'\right| \int_t^T f^Q\left(s,Y\left(t,y,s,\sigma W\right)\right) \beta^{Q'}_s \diff s \right] \\
    F & = \mathbb{E}\left[Z' \int_t^T \left| f^Q\left(s,Y\left(t,y,s,\sigma W\right)\right) -  f^{Q'}\left(s,Y\left(t,y,s,\sigma W\right)\right) \right| \beta^{Q'}_s \diff s \right].
\end{align*}

We want to bound $A,B,C,D,E,F$ by some function of $|Q'-Q|$ and then use Grönwall's inequality. We first state some inequalities involving the $\beta$'s and $Z$'s. Let $s \in [0,T]$. Then,
\begin{align*}
    \beta^Q_s  &\leqslant e^{\frac{1}{2}\sigma^2 \gamma^2 \Qmax^2 T + 2\Lambda^* T} \leqslant K, \\
    \left|\beta^Q_s - \beta^{Q'}_s \right| & \leqslant
    \sigma^2 \gamma^2 \Qmax^3 T e^{\frac{1}{2}\sigma^2 \gamma^2 \Qmax^2 T + 2\Lambda^* T}
     \left|Q' - Q \right| \leqslant K |Q'-Q|,\\
     %\mathbb{E}[Z] &= 1 \\
     \mathbb{E}\left[\left|Z-Z'\right|\right]
     & \leqslant
     \mathbb{E}\left[e^{\sigma \gamma \Qmax |W_T|}|Q-Q'|\left(\sigma \gamma  |W_T| + \sigma^2 \gamma^2 T \Qmax\right)\right]
     \leqslant K |Q'-Q|
\end{align*}
where $K > 0$ is a constant only depending on the constants of the problem ($\gamma$, $\delta$, $\sigma$, $T$, $\Qmax$, $\Lambda^*$). This yields
\begin{equation}\label{eq:contQ2}
\begin{split}
    A & \leqslant K' \varpi(|Q'-Q|)\\
    B, C, D, E & \leqslant K' |Q'-Q| \leqslant K' \varpi(|Q'-Q|)
\end{split}
\end{equation}
where $K'> 0$ is a constant depending only on $\gamma$, $\delta$, $\sigma$, $\qmax$, $T$, $\ell^*$, $\Lambda^*$ and $M$.

By Lemma \ref{lem:H_ineq_Q}, there exists a constant $K'' > 0$, depending only on $\gamma$, $\delta$, $\sigma$, $T$, $\Qmax$, $\ell^*$, $\Lambda^*$ and $M$, such that
\begin{equation}\label{eq:contQ3}
    |f^{Q}(s,y') - f^{Q'}(s,y')| \leqslant K''(|Q'-Q| + m^{|Q'-Q|}(s)),\quad (s,y') \in [0,T) \times \Y.
\end{equation}
Setting $L \defeq K'' + 5 K'$, \eqref{eq:contQ1}, \eqref{eq:contQ2} and \eqref{eq:contQ3} imply
\begin{equation*}
    \left|u^Q(t,y) - u^{Q'}(t,y) \right| \leqslant L \varpi(|Q'-Q|) + L\int_t^T m^{|Q'-Q|}(s) \diff s,
\end{equation*}
and consequently
\begin{equation*}
    \ m^{\Delta}(t) \leqslant  L \varpi(\Delta) + L \int_t^T m^{\Delta}(s) \diff s,\quad t \in [0,T],\,\Delta\geqslant 0.
\end{equation*}
By Grönwall's inequality, we conclude that
\begin{equation*}
    \sup_{t\in [0,T]}m^{\Delta}(t) \leqslant L \varpi(\Delta) e^{LT} \xrightarrow[\Delta \to 0]{} 0.
\end{equation*}

\subsection{Existence in the general case}
\label{subsec:existence_general_case}

In this section, we consider the case $n = \infty$, but we no longer make any assumption on the measures $\mu^a$ and $\mu^b$. Let $(a_k)_{k \in \mathbb{N}}$ be an increasing sequence of natural integers such that for all $k \in \mathbb{N}^*$, $\varpi\left(\frac{\qmax}{2^{a_k}}\right) \leqslant \frac{1}{2^k}$. For $k \in \mathbb{N}^*$, define
\begin{align}
    \label{eq:discrete_measures}
    \begin{split}
        \mu^a_k & = \sum_{j = 0}^{2^{a_k} - 1}\delta_{\qa\frac{j}{2^{a_k}}}\mu^a\left(\left[\qa\frac{j}{2^{a_k}}, \qa\frac{j+1}{2^{a_k}}\right)\right) + \delta_{\qa}\mu^a\left(\left\{\qa\right\}\right) \\
    \mu^b_k & = \sum_{j = 0}^{2^{a_k} - 1}\delta_{\qb\frac{j}{2^{a_k}}}\mu^b\left(\left[\qb\frac{j}{2^{a_k}}, \qb\frac{j+1}{2^{a_k}}\right)\right) + \delta_{\qa}\mu^b\left(\left\{\qb\right\}\right),
    \end{split}
\end{align}
which are probability measures converging in distribution to $\mu^a$ and $\mu^b$, respectively (it is straightforward to show that their respective cumulative distribution functions converge pointwise to the one of the limit distribution).

Let $(u^Q_k)_{Q \in \Q}$ be the solution of \eqref{eq:HJB_interior}-\eqref{eq:HJB_border} constructed in Section \ref{subsec:discrete_case} with measures $\mu^a_k$ and $\mu^b_k$. In particular,
\begin{equation*}
    \sup_{k \in \mathbb{N}^*} \sup_{Q \in \Q} |u^Q_k|_{2+\beta} \leqslant C < \infty
\end{equation*}
for some $C$ depending only on $T$, $\Qmax$, $\gamma$, $\sigma$, $\delta$, $\eta$, $\Lambda^*$, $\beta$, $\Lambda^*_{\beta}$, $\ell^*$ and each function $(t,Q,y)\mapsto u_k^Q(t,y)$ is continuous, as shown in Section \ref{subsec:continuity}. For $k \in \mathbb{N}^*$ and $Q \in \Q$, define $f_k$ as in \eqref{eq:definition_fk}, using the $u_k$'s of this section and replacing $\mu^a$ by $\mu^a_k$ and $\mu^b$ by $\mu^b_k$. Then, $u_k^Q$ solves, for $(t,y) \in [0,T) \times \Y$,
\begin{equation*}
    %\label{eq:hjb_discrete}
    0 = \partial_t u^Q_{k}(t,y) + \frac{\sigma^2}{2}\partial^2_{yy} u^Q_{k}(t,y)
        -\sigma^2 \gamma Q \partial_y u^Q_{k}(t,y) + \left(\frac{\sigma^2 \gamma^2 Q^2}{2}- \left(\Lambda^a + \Lambda^b\right)(y)\right)u^Q_{k}(t,y)+
            f^Q_k(t,y)
\end{equation*}
and for $(t,y) \in [0,T] \times \bar{\Y}$,
\begin{equation*}
    \left\{
        \begin{array}{ll}
            u_{k}^Q(T, y) &=-e^{\gamma\ell(Q)} \\
            u_{k}^Q\left(t, \bar{y}\right) &= u_{k}^Q\left(t, y_+\right) \\
            u_{k}^Q\left(t, -\bar{y}\right) &= u_{k}^Q\left(t, y_-\right).
        \end{array}
    \right.
\end{equation*}

From the results of Section \ref{subsec:continuity}, there exists a constant $L' > 0$ depending only on $\gamma$, $\delta$, $\sigma$, $T$, $\Qmax$, $\ell^*$, $\Lambda^*$ and $C$ such that for all $k \in \mathbb{N}^*$, $(t,y) \in [0,T] \times \Y$, $L'\varpi$ is a modulus of continuity of the function $Q \in \Q\mapsto u_{k}^Q(t,y)$. Thus, by Lemma \ref{lem:ineq_H_discrete_meas}, there exists a constant $L > 0$, depending only on $\gamma$, $\delta$, $\sigma$, $T$, $\Qmax$, $\ell^*$, $\Lambda^*$ and $C$, such that
\begin{equation*}
    \left|f_{k+1}^Q(t,y) - f_k^Q(t,y) \right|
    \leqslant L\left|u_{k+1}^Q(t,y) - u_k^Q(t,y) \right| + L\varpi\left(\frac{\qmax}{2^{a_k}}\right),\quad (t,Q,y) \in [0,T) \times \Q \times \Y,\,k \in \mathbb{N}^*.
\end{equation*}

For $k \in \mathbb{N}^*$ and $t \in [0,T]$, define $g_k(t) \defeq e^{ct}\sup_{(Q,y) \in \Q \times \Y}\left|u^Q_{k+1}(t,y) - u^Q_{k}(t,y) \right|_{\infty}$ where $c \defeq \frac{\sigma^2 \gamma^2 \Qmax^2}{2} + 2 \Lambda^*$ -- there is no measurability issue since the $u_k$'s are continuous with respect to $(t,Q,y)$ and we can take the supremum over a countable set.

Let $k \in \mathbb{N}^*$ and $Q \in \Q$. The function $v \defeq u^Q_{k+1} - u^{Q}_k$ solves 
\begin{equation*}
    \begin{aligned}
        0 =& \left(\partial_t v + \frac{\sigma^2}{2}\partial^2_{yy} v
        -\sigma^2 \gamma Q \partial_y v + \left(\frac{\sigma^2 \gamma^2 Q^2}{2}- \left(\Lambda^a + \Lambda^b\right)(y)\right)v+
            f^Q_{k+1} - f^Q_{k}\right)(t,y)
    \end{aligned}
\end{equation*}
with boundary conditions
\begin{equation*}
    \left\{
        \begin{array}{ll}
            v(T, y) &=0 \\
            v\left(t, \bar{y}\right) &= v\left(t, y_+\right) \\
            v\left(t, -\bar{y}\right) &= v\left(t, y_-\right).
        \end{array}
    \right.
\end{equation*}
Combined with Corollary \ref{corol:bound_on_u}, this yields
\begin{equation*}
     g_k(t) \leqslant LT e^{cT} \varpi\left(\frac{\qmax}{2^{a_k}}\right) + \int_t^T L g_k(s)\diff s,\quad t \in [0,T],\, k \in \mathbb{N}^* .
\end{equation*}
By Grönwall's inequality,
\begin{equation*}
    \sup_{(t,Q,y) \in [0,T] \times \Q \times \Y}\left|u^Q_{k+1}(t,y) - u^Q_{k}(t,y) \right|
    \leqslant LTe^{(L+c)T} \varpi\left(\frac{\qmax}{2^{a_k}}\right),\quad k \in \mathbb{N}^*.
\end{equation*}
But $\sum_{k = 1}^{\infty} \varpi\left(\frac{\qmax}{2^{a_k}}\right) \leqslant \sum_{k = 1}^{\infty} \frac{1}{2^k} < \infty$. Thus, $\sum_{k=1}^{\infty} \sup_{Q \in \Q}|u_{k+1}^Q - u_k^Q|_{\infty} < \infty$. Hence, for all $Q \in \Q$, $(u^Q_k)_k$ is a Cauchy sequence for the sup norm and it converges uniformly to a continuous function $u^Q$. Moreover, since $\sup_{Q \in \Q}|u^Q - u_k^Q|_{\infty} \to 0$, $(t,Q,y)\mapsto u^Q(t,y)$ is continuous on $[0,T] \times \Q \times \bar{\Y}$.

It is clear that $(u^Q)_{Q \in \Q}$ verifies the boundary conditions \eqref{eq:HJB_border}.

Fix $Q \in \Q$. On a fixed compact $K$ included in $[0,T) \times \Y$, $\left(\partial_t u_k^Q\right)_{k \in \mathbb{N}^*}$, $\left(\partial_y u_k^Q\right)_{k \in \mathbb{N}^*}$ and $\left(\partial^2_{yy} u_k^Q\right)_{k \in \mathbb{N}^*}$ are sequences of uniformly bounded and equicontinuous functions. Hence, by the Arzelà-Ascoli theorem, they converge uniformly on $K$ up to a subsequence. Thus, $u^Q \in \D$ and, up to a subsequence, $\partial_t u_k^Q \to \partial_t u^Q$, $\partial_y u_k^Q \to \partial_y u^Q$ and $\partial^2_{yy}u_k^Q \to \partial^2_{yy} u^Q$ uniformly on $K$.

Define $f^Q$ by \eqref{eq:definition_f}. In order to show that $(u^Q)_{Q \in \Q}$ verifies \eqref{eq:HJB_interior}, it suffices to prove that for all $Q \in \Q$ and $(t,y) \in [0,T) \times \Y$, $\lim\limits_{k\to \infty}f^Q_k(t,y) = f^Q(t,y)$. 

Defining, for $(t,Q,y) \in [0,T) \times \Q \times \Y$
\begin{equation*}
    \tilde{f}_k^Q(t,y) \defeq \Lambda^a(t,y)H^a\left((u^Q(t,y))_{Q\in\Q}, Q, y, \mu^a_k\right) + \Lambda^b(t,y)H^b\left((u^Q(t,y))_{Q\in\Q}, Q, y, \mu^b_k\right),
\end{equation*}
Lemma \ref{lem:H_ineq_w} implies that $\lim\limits_{k\to \infty}(f^Q_k(t,y) - \tilde{f}_k^Q(t,y))= 0$. Lemma \ref{lem:H_convergence_measures} shows that $\lim\limits_{k\to \infty}\tilde{f}^Q_k(t,y) = f^Q(t,y)$, which concludes the proof of Theorem \ref{thm:existence}.

\end{dummyenv}

\section{Uniqueness of the control policy}
\label{sec:uniqueness}
\begin{dummyenv}
\renewcommand{\Q}{\mathcal{Q}_{\infty}}
\renewcommand{\Y}{\mathcal{Y}}
\renewcommand{\Qa}{\mathcal{Q}^{+,a}_{\infty}}
\renewcommand{\Qb}{\mathcal{Q}^{+,b}_{\infty}}
\renewcommand{\Qi}{\mathcal{Q}^{+,i}_{\infty}}

In this section, we prove Proposition \ref{prop:log_concavity} and Theorem \ref{thm:uniqueness}.

Recall that $\mu^a$ and $\mu^b$ are two probability measures on $\Qa$ and $\Qb$, respectively, such that $\qa$ belongs to the support of $\mu^a$ and $\qb$ belongs to the support of $\mu^b$ (the purpose of this assumption is to ensure the uniqueness of the optimal control).

For a real-valued continuous function $g$ defined on a compact interval $I\subset \mathbb{R}$, we denote by $\hat{g}$ its continuous concave envelope. More precisely,
\begin{equation*}
    \hat{g}(x) = \inf\left\{h(x): h \geqslant g \text{, and } h \text{ is continuous and concave}\right\},\quad x\in I.
\end{equation*}

\begin{rem}
    The continuous concave envelope $\hat{g}$ is concave, hence lower semi-continuous. In addtion, $\hat{g}$ is upper semi-continuous being the infimum of a family of continuous functions. Thus, $\hat{g}$ is continuous.
\end{rem}

For a real-valued continuous function $g$ defined on an interval $I$, we denote, for $x < \max I$, by $g'_+(x)$ its right derivative at $x$ (which may be equal to $+\infty$ at $x = \min I$) and, for $x>\min I$, by $g'_-(x)$ its left derivative at $x$ (which may be equal to $-\infty$ at $x = \max I$).

Let $g:\Q \mapsto \mathbb{R}$ be a continuous function. We define
\begin{equation}\label{eq:defCg}
    \begin{array}{rccl}
        C_g : & \Q \times \Q \times [0,1] & \to & \mathbb{R} \\
        & (Q,Q',\lambda) & \mapsto & g\left((1-\lambda)Q + \lambda Q'\right) - (1-\lambda) g(Q) - \lambda g(Q').
    \end{array}
\end{equation}

\begin{rem}\label{rem:concave}
    The function $g$ is concave if and only if $\min C_g  = 0$. Moreover, $g$ is strictly concave if and only if for all $(Q,Q') \in \Q \times \Q$ such that $Q < Q'$, $C_g\left(Q,Q',\frac{1}{2}\right)>0$.
\end{rem}

For $k \in \mathbb{R}$, $Q \in \Q$ and a continuous function $g:\Q \mapsto \mathbb{R}$, we define the \enquote{log-Hamiltonians}
\begin{equation}\label{eq:deflogHamiltonian}
\begin{split}
    h^a(g, k, Q) &\defeq \sup_{q \in \Qa \cap [0, Q + \Qmax]}
    \int_{\Qa} -e^{-k(q \wedge z) - g(Q - q \wedge z) + g(Q)}\mu^a(\diff u), \\
    h^b(g, k, Q) &\defeq \sup_{q \in \Qb \cap [0, -Q + \Qmax]}
    \int_{\Qb} -e^{-k(q \wedge z) - g(Q + q \wedge z) + g(Q)}\mu^b(\diff u). 
\end{split}
\end{equation}

For $(t,Q,y) \in [0,T] \times \Q \times \Y$, by Theorem \ref{thm:verification}, there exists some $q^* \in \mathcal{A}^{t,Q}$ such that $U(t,Q,y) = J(t,Q,y,q^*)$, which is negative being the expectation of a negative random variable. This also holds for $y \in \bar{\Y}$ since $U(t,Q,\bar{y}) = U(t,Q,y_+)$ and $U(t,Q,-\bar{y}) = U(t,Q,y_-)$. Proposition \ref{prop:log_concavity}(i) is proven. Thus, we can define, for $(t,Q,y) \in [0,T] \times \Q \times \bar{\Y}$,
\begin{equation*}
    v^Q(t,y) \defeq - \ln \left(-U(t,Q,y)\right).
\end{equation*}

We suppose that the penalty function $\ell$ is convex.

In Subsection \ref{subsec:log_concavity}, we prove, by contradiction, the strict concavity of $v$ with respect to $Q$ on $[0,T) \times \Q \times \bar{\Y}$, i.e. Proposition \ref{prop:log_concavity}(ii).

In Subsection \ref{subsec:proof_thm_uniqueness}, we show that the strict concavity of $Q \mapsto v^Q(t,y)$ implies the uniqueness of the control policy, i.e. Theorem \ref{thm:uniqueness}.

\subsection{Log-convexity of the negative of the value function}
\label{subsec:log_concavity}

For $Q \in \Q$, $v^Q \in C([0,T] \times \bar{\Y}) \cap C^{1,2}([0,T) \times \Y)$ and $(t,Q,y) \mapsto v^Q(t,y)$ is continuous on $[0,T] \times \Q \times \Y$. It is straightforward to check that for all $(t,y,Q) \in [0,T) \times \mathcal{Y} \times \Q$,
\begin{equation}
    \label{eq:log_HJB_interior}
    \begin{aligned}
        0 =&\partial_t v^Q(t,y) + \frac{\sigma^2}{2}\partial^2_{yy} v^Q(t,y) - \frac{\sigma^2}{2}\left(\partial_y v^Q(t,y) + \gamma Q\right)^2 + \left(\Lambda^a + \Lambda^b\right)(t,y)\\
        &+\Lambda^a(t,y)
            h^a\left((v^R(t,y))_{R \in \Q}, \gamma\left(\frac{\delta}{2} - y\right), Q\right)
        + \Lambda^b(t,y)
        h^b\left((v^R(t,y))_{R \in \Q}, \gamma\left(\frac{\delta}{2} + y\right), Q\right),
    \end{aligned}
\end{equation}
and for all $(t,y,Q) \in [0,T] \times \bar{\mathcal{Y}} \times \Q$, the following boundary conditions hold:
\begin{equation}
    \label{eq:log_HJB_border}
    \left\{
        \begin{array}{ll}
            v^Q(T, y) &= - \ell(Q) \\
            v^Q\left(t, \ybar\right) &= v^Q\left(t,\yplus\right) \\
            v^Q\left(t, -\ybar\right) &= v^Q\left(t, \yminus\right).
        \end{array}
    \right.
\end{equation}

We introduce the continuous function
\begin{equation*}
    \begin{array}{rccl}
        \mathcal{C} : & [0,T] \times \bar{\mathcal{Y}} \times \Q \times \Q \times [0,1] & \to & \mathbb{R} \\
        & (t, y, Q,Q',\lambda) & \mapsto & v^{(1-\lambda)Q + \lambda Q'}(t,y) - (1-\lambda) v^Q(t,y) - \lambda v^{Q'}(t,y).
    \end{array}
\end{equation*}
For fixed $(Q,Q', \lambda)$, $\mathcal{C}(\cdot, \cdot, Q, Q', \lambda) \in C^{1,2}([0,T)\times \mathcal{Y})$, and for fixed $(t,y)$, $\mathcal{C}(t,y, \cdot, \cdot, \cdot) = C_g$ with $g: R \in \Q \mapsto v^R(t,y)$.

We show the concavity of $Q \mapsto v^Q(t,y)$ with respect to the $Q$ variable. The proof is divided in two parts: first we prove the concavity of $Q \mapsto v^Q(t,y)$, then its strict concavity on $[0,T) \times \mathcal{Y}$. We use Remark \ref{rem:concave} and a maximum principle argument, following similar arguments as in \parencite{korevaar_convex_1983}.

\subsubsection{Concavity with respect to the inventory}

Suppose there exists $(t,y) \in [0,T] \times \mathcal{Y}$ such that $Q \mapsto v^Q(t,y)$ is not concave. Then, by Remark \ref{rem:concave}, $\min \mathcal{C} < 0$. Since for all $(y,Q)$, $v^Q(T,y) = -\ell(Q)$ and $-\ell$ is concave, necessarily $\mathcal{C}(T,\cdot, \cdot, \cdot, \cdot) \geqslant 0.$ Suppose then that $(t,y,Q,Q',\lambda) \in [0,T) \times \mathcal{Y} \times \Q \times \Q \times (0,1)$ minimizes $e^{t}\mathcal{C}(t,y,Q,Q',\lambda)$, with $Q < Q'$. We can choose $y \in \mathcal{Y}$ because if $\mathcal{C}$ is minimized on the boundary of $\mathcal{Y}$, then it is also minimized on its interior, thanks to \eqref{eq:log_HJB_border}; see \parencite{baldacci_bid_2020} for a similar argument.

Define $\tilde{Q} \defeq (1-\lambda)Q + \lambda Q'$. The optimality of $(t,y)$ implies
\begin{equation*}
    \partial_t \mathcal{C}(t,y,Q,Q',\lambda) + \mathcal{C}(t,y,Q,Q',\lambda) \geqslant 0
    ,\ \partial_y \mathcal{C}(t,y,Q,Q',\lambda) = 0
    \text{ and } \partial^2_{yy} \mathcal{C}(t,y,Q,Q',\lambda) \geqslant 0,
\end{equation*}
and therefore
\begin{align*}
        \partial_t v^{\tilde{Q}}(t,y) + v^{\tilde{Q}}(t,y) & \geqslant (1-\lambda)\partial_t v^Q(t,y) + \lambda\partial_t v^{Q'}(t,y) + (1-\lambda)v^Q(t,y) + \lambda v^{Q'}(t,y)\\
        \partial_y v^{\tilde{Q}}(t,y) & = (1-\lambda)\partial_y v^Q(t,y) + \lambda\partial_y v^{Q'}(t,y)\\
        \partial^2_{yy} v^{\tilde{Q}}(t,y) & \geqslant (1-\lambda)\partial^2_{yy} v^Q(t,y) + \lambda\partial^2_{yy} v^{Q'}(t,y).
    \end{align*}
    Hence,
    \begin{equation*}
        \begin{split}
            \partial_t &v^{\tilde{Q}}(t,y) + \frac{\sigma^2}{2}\partial^2_{yy} v^{\tilde{Q}}(t,y) + v^{\tilde{Q}}(t,y)
        \\ &\geqslant
        (1-\lambda)\left(\partial_t v^{Q}(t,y) + \frac{\sigma^2}{2}\partial^2_{yy} v^{Q}(t,y) + v^{Q}(t,y)\right)
        +\lambda \left(\partial_t v^{Q'}(t,y) + \frac{\sigma^2}{2}\partial^2_{yy} v^{Q'}(t,y) + v^{Q'}(t,y)\right).
        \end{split}
    \end{equation*}
    Using equation \eqref{eq:log_HJB_interior} and denoting $g \defeq (v^R(t,y))_{R \in \Q}$, we obtain 
    \begin{equation*}
        \begin{split}
            \Lambda^a&(t,y)\left((1-\lambda)h^a\left(g, \gamma\left(\frac{\delta}{2} - y\right),Q\right)
            + \lambda h^a\left(g, \gamma\left(\frac{\delta}{2} - y\right), Q'\right) - h^a\left(g, \gamma\left(\frac{\delta}{2} - y\right),\tilde{Q} \right)\right)\\
            &+\Lambda^b(t,y)\left((1-\lambda)h^b\left(g, \gamma\left(\frac{\delta}{2} + y\right),Q\right)
            + \lambda h^b\left(g, \gamma\left(\frac{\delta}{2} + y\right), Q'\right) - h^b\left(g, \gamma\left(\frac{\delta}{2} + y\right),\tilde{Q} \right)\right)\\
            & \geqslant
            (1-\lambda) \left(\partial_y v^{Q}(t,y) + \gamma Q\right)^2 + \lambda \left(\partial_y v^{Q'}(t,y) + \gamma Q'\right)^2 - \left(\partial_y v^{\tilde{Q}}(t,y) + \gamma \tilde{Q}\right)^2\\
            &+(1-\lambda)v^Q(t,y) + \lambda v^{Q'}(t,y) - v^{\tilde{Q}}(t,y).
        \end{split}
    \end{equation*}
    Since the square function is convex and
    \begin{equation*}
        \partial_y v^{\tilde{Q}}(t,y) + \gamma \tilde{Q} = (1-\lambda)\left(\partial_y v^Q(t,y) + \gamma Q\right) + \lambda\left(\partial_y v^{Q'}(t,y) + \gamma Q'\right),
    \end{equation*}
    $(1-\lambda) \left(\partial_y v^{Q}(t,y) + \gamma Q\right)^2 + \lambda \left(\partial_y v^{Q'}(t,y) + \gamma Q'\right)^2 - \left(\partial_y v^{\tilde{Q}}(t,y) + \gamma \tilde{Q}\right)^2 \geqslant 0$. Moreover, by assumption, $(1-\lambda)v^Q(t,y) + \lambda v^{Q'}(t,y) - v^{\tilde{Q}}(t,y) > 0$. Using Lemma \ref{lem:envelope_inequality}, we get
    \begin{equation}
        \label{eq:concavity_final_ineq}
        \begin{split}
                &\Lambda^a(t,y)\left((1-\lambda)h^a\left(\hat{g}, \gamma\left(\frac{\delta}{2} - y\right),Q\right)
                + \lambda h^a\left(\hat{g}, \gamma\left(\frac{\delta}{2} - y\right), Q'\right) - h^a\left(\hat{g}, \gamma\left(\frac{\delta}{2} - y\right),\tilde{Q} \right)\right)\\
                &+\Lambda^b(t,y)\left((1-\lambda)h^b\left(\hat{g}, \gamma\left(\frac{\delta}{2} + y\right),Q\right)
                + \lambda h^b\left(\hat{g}, \gamma\left(\frac{\delta}{2} + y\right), Q'\right) - h^b\left(\hat{g}, \gamma\left(\frac{\delta}{2} + y\right),\tilde{Q} \right)\right) > 0.
        \end{split}
    \end{equation}
    By Corollary \ref{corol:envelope_shape}, $\hat{g}(\tilde{Q}) = (1-\lambda)\hat{g}(Q) + \lambda \hat{g}(Q')$. Thus, by Lemma \ref{lem:bid_ask_max}(i), the left-hand side of \eqref{eq:concavity_final_ineq} is nonpositive, leading to a contradiction.

\subsubsection{Strict concavity with respect to the inventory}

We have shown that for all $(t,y) \in [0,T] \times \bar{\Y}$, $Q \mapsto v^Q(t,y)$ is concave. Suppose there exists $(t,y) \in [0,T) \times \bar{\Y}$ such that $Q \mapsto v^Q(t,y)$ is not strictly concave. We can take $y \in \Y$, thanks to \eqref{eq:log_HJB_border}. Then, by Remark \ref{rem:concave}, there exists $(Q_0, Q_0') \in \Q \times \Q$ such that $Q_0 < Q_0'$ and $\mathcal{C}\left(t,y,Q_0,Q_0',\frac{1}{2}\right) = 0$. Define $g:Q \in \Q \mapsto v^Q(t,y)$ and $p \defeq \frac{g(Q_0') - g(Q_0)}{Q_0' - Q_0}$. Set
\begin{align*}
    Q &\defeq \inf\left\{R \in \Q :\frac{g(Q_0') - g(R)}{Q_0' - R} = p \right\} \leqslant Q_0 \\
    Q' &\defeq \sup\left\{R \in \Q :\frac{g(R) - g(Q_0)}{R - Q_0} = p \right\} \geqslant Q_0'.
\end{align*}
We have $\frac{g(Q')-g(Q)}{Q'-Q} = p$, for all $ R \in \Q \cap (-\infty,Q)$, $g'_+(R) > p$, and for all $ R \in \Q \cap (Q', \infty)$, $g'_-(R) < p$.

Since $\mathcal{C}\left(t,y,Q,Q',\frac{1}{2}\right) = 0$, $(t,y)$ minimizes $\mathcal{C}\left(\cdot,\cdot,Q,Q',\frac{1}{2}\right)$. Defining $\tilde{Q} \defeq \frac{1}{2}Q + \frac{1}{2}Q'$, and proceeding in the exact same way as in the previous section, we obtain
\begin{equation*}
    \begin{split}
            &\Lambda^a(t,y)\left((1-\lambda)h^a\left(g, \gamma\left(\frac{\delta}{2} - y\right),Q\right)
            + \lambda h^a\left(g, \gamma\left(\frac{\delta}{2} - y\right), Q'\right) - h^a\left(g, \gamma\left(\frac{\delta}{2} - y\right),\tilde{Q} \right)\right)\\
            &+\Lambda^b(t,y)\left((1-\lambda)h^b\left(g, \gamma\left(\frac{\delta}{2} + y\right),Q\right)
            + \lambda h^b\left(g, \gamma\left(\frac{\delta}{2} + y\right), Q'\right) - h^b\left(g, \gamma\left(\frac{\delta}{2} + y\right),\tilde{Q} \right)\right) \geqslant 0.
    \end{split}
\end{equation*}
By Lemma \ref{lem:bid_ask_max}(i), both terms in the left-hand side are nonpositive and therefore equal to $0$. By Lemma \ref{lem:bid_ask_max}(ii), we must have $p \geqslant \gamma\left(\frac{\delta}{2} - y\right)$ and $p \leqslant -\gamma \left(\frac{\delta}{2}+y\right)$ (otherwise, one of the terms would be strictly negative, since $\Lambda^a, \Lambda^b > 0$). This implies $\gamma \delta \leqslant 0$, which is a contradiction.

\subsection{Proof of Theorem \ref{thm:uniqueness}}
\label{subsec:proof_thm_uniqueness}

Fix $(t,Q,y) \in [0,T) \times \Q \times \Y$. For $q \in \Qa \cap [0,Q + \Qmax]$, $e^{-\gamma q \left(\frac{\delta}{2}- y\right)}U(t,Q-q,y) = -e^{-\gamma q \left(\frac{\delta}{2}- y\right) - v^{Q-q}(t,y)}$. The function $q \mapsto \gamma q \left(\frac{\delta}{2}- y\right) + v^{Q-q}(t,y)$ is concave. Then, Corollary \ref{corol:max_integral} implies the first equality of Theorem \ref{thm:uniqueness}(i). Moreover, because $q \mapsto \gamma q \left(\frac{\delta}{2}- y\right) + v^{Q-q}(t,y)$ is strictly concave, it has a unique maximizer on $\Qa \cap [0,Q + \Qmax]$. A similar argument can be used on the bid side, to conclude the proof of Theorem \ref{thm:uniqueness}(i).

Consequently, the functions $\hat{q}^a$ and $\hat{q}^b$ given by Theorem \ref{thm:verification} are uniquely determined on $[0,T) \times \Q \times \bar{\Y}$. We show that they are continuous on this set. As usual, we only do it only for the ask side because the argument is analogous on the bid side. Let $(t_k)_k$, $(Q_k)_k$ and $(y_k)_k$ be three sequences of elements of $[0,T)$, $\Q$ and $\bar{\Y}$, converging to $t$, $Q$ and $y$,  respectively. Let $(k_i)_{i}$ be a (strictly) increasing sequence on nonnegative integers such that $\left(\hat{q}^a(t_{k_i}, Q_{k_i}, y_{k_i})\right)_i$ converges to some $q \in \Qa$. It is sufficient to show that $q = \hat{q}^a(t,Q,y)$. Since for all $i$, $\hat{q}^a(t_{k_i}, Q_{k_i}, y_{k_i}) \leqslant Q_{k_i} + \Qmax$, we have $q \leqslant Q + \Qmax$. By definition, for all $i$,
\begin{equation*}
    \begin{split}
        e^{-\gamma \hat{q}^a(t_{k_i}, Q_{k_i}, y_{k_i})\left(\frac{\delta}{2}-y_{k_i}\right)}U(t_{k_i}, &Q_{k_i} - \hat{q}^a(t_{k_i}, Q_{k_i}, y_{k_i}), y_{k_i})\\
    &\geqslant e^{-\gamma \left(\hat{q}^a(t, Q, y)\wedge\left(Q_{k_i} + \Qmax\right)\right)\left(\frac{\delta}{2}-y_{k_i}\right)}U(t_{k_i}, Q_{k_i} - \hat{q}^a(t, Q, y) \wedge \left(Q_{k_i} + \Qmax\right), y_{k_i}).
    \end{split}
\end{equation*}
Passing to the limit, we get
\begin{equation*}
    e^{-\gamma q\left(\frac{\delta}{2}-y\right)}U(t, Q - q, y)\\
    \geqslant e^{-\gamma \hat{q}^a(t, Q, y)\left(\frac{\delta}{2}-y\right)}U(t, Q - \hat{q}^a(t, Q, y), y).
\end{equation*}
In addition, by the definition on $\hat{q}^a$, the right-hand side is greater or equal than the left-hand side, and we have equality. Since the maximizer is unique, $q = \hat{q}^a(t,Q,y)$, which yields Theorem \ref{thm:uniqueness}(iii).
\end{dummyenv}

Let $q =(q^a, q^b) \in \mathcal{A}^{t,Q}$ be a control policy such that $J(t,Q,y,q) = U(t,Q,y)$. By Theorem \ref{thm:uniqueness}(i) and Theorem \ref{thm:verification}(iv), we have, $\diff s \otimes \Proba[][\diff \omega]$-almost everywhere on $[t,T] \times \Omega$,
\begin{align*}
    q^a_s(\omega) = \hat{q}^a(s, \Q_{s-}, \Y_s) \text{ and }
    q^b_s(\omega) = \hat{q}^b(s, \Q_{s-}, \Y_s).
\end{align*}
In order to show that $q_s(\omega) = q_s^*(\omega)$  $\diff s \otimes \Proba[][\diff \omega]$-almost everywhere on $[t,T] \times \Omega$, it is sufficient to show that, with probability one for all $s \in [t,T)$, $\Q_s = Q^{t,Q,q^*}_s$. We denote by $t_1 < t_2 < \dots$ the jump times of $N \defeq N^a(\cdot \times \mathcal{Q}_{\infty}^{+,a}) + N^b(\cdot \times \mathcal{Q}_{\infty}^{+,b})$ on $(t,T)$, with $t_i = T$ if $i \geqslant N_T$. We show by induction that $\Q \mathds{1}_{[t,t_i)} = Q^{t,Q,q^*} \mathds{1}_{[t,t_i)}$ with probability one for all $i$. It holds for $i = 1$ since $\Q$ and $Q^{t,Q,q^*}$ are constant and equal to $Q$ on the (random) interval $[t,t_1)$. Suppose that the property holds for $i \in \mathbb{N}^*$. We will show that once we make this assumption, $q$ and $q^*$ coincide at the next jump of $N$ with probability one, which in turn implies the induction step. We have
\begin{align*}
    \Proba[][\left\{t_{i} < T\right\} \cap \left\{q_{t_{i}} \neq  q^*_{t_{i}}\right\}] & = \E[][\int_{(t,T)}\mathds{1}_{\{N_{s-} - N_t = i-1\}}\mathds{1}_{\{q_s \neq q_s^*\}} N(\diff s)]\\
    & = 2\E[][\int_{(t,T)}\mathds{1}_{(t_{i-1},t_{i})}(s)\mathds{1}_{\{q_s \neq q_s^*\}}\diff s].
\end{align*}
This probability is equal to zero. Indeed, almost surely for all $s \in (t,T)$,
\[\mathds{1}_{[t,t_i)}(s)\mathds{1}_{\{q_s \neq q_s^*\}} = \mathds{1}_{[t,t_i)}(s)\mathds{1}_{\{q^a_s \neq \hat{q}^a(s,\Q_{s-},\Y_s)\}}\mathds{1}_{\{q^b_s \neq \hat{q}^b(s,\Q_{s-},\Y_s)\}},
\]
and therefore $\mathds{1}_{[t,t_i(\omega))}(s)\mathds{1}_{\{q_s(\omega) \neq q_s^*(\omega)\}} = 0$, $\diff s \otimes \Proba[][\diff \omega]$-almost everywhere on $[t,T] \times \Omega$. Theorem \ref{thm:uniqueness}(ii) follows by induction.

\appendix

\section{Construction of the mid-price process}
\label{sec:midprice_construction}

\subsection{Setup}

Let $a < b$ be two real numbers and $a_0, b_0 \in (a,b)$. Let $T > 0$ and define $\Delta \defeq \min\{a_0-a, b-a_0, b_0-a, b - b_0\} > 0$.

In this section, we will construct a measurable mapping $Y:[0,T] \times [0, T] \times (a,b) \times \mathcal{C}_0^0([0, T]) \mapsto \mathbb{R}$, valued in $(a,b)$, where $\mathcal{C}_0^0([0,T])$ is the space of real-valued continuous functions $w$ on $[0,T]$ verifying $w(0)=0$, equipped with the sup norm.

Intuitively, $Y(t,\cdot,y,w)$ is a function starting with the value $y$ at $t$ which behaves like $w$ until it attains the boundary $a$ or $b$. If it reaches the boundary at $a$, it jumps to $a_0$, and if it reaches the boundary at $b$, it jumps to $b_0$.

\subsection{Construction}
\label{subsec:midrpice_construction}

For $(t, y, w) \in [0,T] \times (a,b) \times \mathcal{C}_0^0([0, T])$, define for notational convenience $\tau_0(t,y,w) \defeq t$.

Assume that we have constructed measurable mappings $\tau_0 \leqslant \dots \leqslant \tau_n \leqslant T$ and $\epsilon_1, \dots, \epsilon_n$  from $[0,T] \times (a,b) \times \mathcal{C}_0^0([0, T])$ to $[0,T]$ and $\{-1,0,1\}$, respectively, for some $n \in \mathbb{N}$. The variable $\tau_i$ represents the barrier hitting times, $\epsilon_i = 1$ if the boundary is attained at $b$, $\epsilon_i = -1$ if the boundary is attained at $a$. We define $\epsilon'$ to be the jump size, namely $\epsilon_i'(t,y,w) \defeq b_0-b$ if $\epsilon_i(t,y,w)=1$, $\epsilon_i'(t,y,w) \defeq a_0-a$ if $\epsilon_i(t,y,w)=-1$, and $\epsilon_i'(t,y,w) \defeq 0$ otherwise. 

Assume moreover that for every $i \in \integerInterval{n}$ and $(t,y,w)$,
\begin{equation}\label{eq:appA1}
\begin{split}
    &\tau_i(t,y,w) = T \implies \forall j \in \integerInterval[i+1]{n}, \left\{\begin{array}{l}
        \tau_j(t,y,w) = T \\
        \epsilon_j(t,y,w) = 0
    \end{array} \right. \\
    &\tau_i(t,y,w) < T \implies \epsilon_i(t,y,w) \neq 0 \\
    &\tau_i(t,y,w) < T \implies y + (w_{\tau_{i}(t,y,w)} - w_t) + \sum_{j=1}^{i-1} \epsilon_{j}'(t,y,w) \in \left\lbrace a, b \right\rbrace \\
    & i < n,\ \epsilon_{i+1}(t,y,w) \neq 0 \implies \left|w_{\tau_{i+1}(t,y,w)} - w_{\tau_{i}(t,y,w)}\right| \geqslant \Delta \\
    &y + (w_{u} - w_t) +  \sum_{j=1}^{n} \epsilon_{j}'(t,y,w)\mathds{1}_{[\tau_i(t,y,w), T]}(u) \in (a,b),\quad u \in (t,\tau_n(t,y,w)].
\end{split}
\end{equation}

Let $(t, y, w) \in [0,T] \times (a,b) \times \mathcal{C}_0^0([0, T])$. If $\tau_n(t,y,w) = T$, set $\tau_{n+1}(t,y,w) = T$ and $\epsilon_{n+1}(t,y,w) = 0$. In this case the five properties \eqref{eq:appA1} are trivially verified for $\tau_{n+1}(t,y,w)$ and $\epsilon_{n+1}(t,y,w)$.

Suppose now that $\tau_n(t,y,w) < T$. We omit temporarily the dependence on $(t,y,w)$ for concision. We define
\begin{equation}\label{eq:appAtaunepsn}
    \begin{split}
    &\tau_{n+1} \defeq \inf \left\lbrace r \in (\tau_n, T]: y + (w_r - w_t) + \sum_{i=1}^n \epsilon_{i}' \in \left\lbrace a,b \right\rbrace \right\rbrace \wedge T\\
    &\epsilon_{n+1} \defeq \mathds{1}_{\left\{y + (w_{\tau_{n+1}} - w_t) + \sum_{i=1}^n \epsilon_{i}' = b\right\}} - \mathds{1}_{\left\{y + (w_{\tau_{n+1}} - w_t) +\sum_{i=1}^n \epsilon_{i}' = a\right\}}
\end{split}
\end{equation}

The first three properties in \eqref{eq:appA1} are trivially verified for $\tau_{n+1}$ and $\epsilon_{n+1}$.
By assumption, $y + (w_{\tau_n} - w_t) + \sum_{i=1}^n\epsilon_i' \in (a,b)$. Therefore, by the continuity of $w$, the intermediate value theorem and the definition of $\tau_{n+1}$, $y + (w_{u} - w_t) + \sum_{i=1}^n\epsilon_i' \in (a,b)$ for every $u \in (\tau_{n}, \tau_{n+1})$. If $\epsilon_{n+1} = 0$, this is also verified at $u = T$, and the fourth property in \eqref{eq:appA1} holds. Suppose now that $\epsilon_{n+1} \neq 0$. Then, $y + (w_{\tau_n} - w_t) + \sum_{i=1}^{n-1}\epsilon_i' \in \left\lbrace a,b\right\rbrace$ and $y + (w_{\tau_n} - w_t) + \sum_{i=1}^{n}\epsilon_i'  \in \{a_0,b_0\}$. Since $y + (w_{\tau_{n+1}} - w_t) + \sum_{i=1}^{n}\epsilon_i' \in \left\lbrace a,b \right\rbrace$, by taking the difference of these last two expressions, $(w_{\tau_{n+1}} - w_{\tau_n}) \in \left\{a-a_0, b-a_0, a-b_0, b - b_0 \right\}$, hence the fourth property in \eqref{eq:appA1} holds. Similarly, $y + (w_{\tau_{n+1}} - w_t) + \sum_{i=1}^{n+1}\epsilon_i' \in \left\lbrace a_0,b_0\right\rbrace \subset (a,b)$, and the fifth property in \eqref{eq:appA1} is verified.

The following lemma shows that $Y$ can only jump finitely many times.
\begin{lemma}
    \label{lem:finite_activity}
    For all $(t, y, w) \in [0,T] \times \mathcal{Y} \times \mathcal{C}_0^0([0, T])$, there exists $N \in \mathbb{N}$ such that for all $n \geqslant N$, $\tau_n(t,y,w) = T$.
\end{lemma}

\begin{proof}
    Suppose that the conclusion does not hold for some $(t,y,w)$. Then, $(\tau_n(t,y,w))_{n \in \mathbb{N}}$ increases to some limit $t'$. By the left-continuity of $w$, $w_{\tau_n(t,y,w)} \to w_{t'}$, therefore $w_{\tau_{n+1}(t,y,w)}- w_{\tau_n(t,y,w)} \to 0$, contradicting the fourth property in \eqref{eq:appA1}.
\end{proof}

We can now define the function $Y$ for $(t,s,y,w) \in [0,T] \times \mathcal{Y} \times \mathcal{C}_0^0([0, T])$ by
\begin{equation}
    \label{eq:y}
    Y \left(t,s,y,w\right) = y + (w_s - w_t) + \sum_{n = 1}^{\infty}\epsilon'_n(t,y,w)\mathds{1}_{[\tau_{n}(t,y,w),T]}(s).
\end{equation}
It is well-defined thanks to Lemma \ref{lem:finite_activity}.

\subsection{Some useful properties of the barrier reaching times}
\label{subsec:tau_properties}

In this section, we will state some properties of the $\tau_n$'s and the $\epsilon_n$'s that will help us to prove useful properties of $Y$, as stated in Lemmas \ref{lem:finite_activity_expectation} and \ref{lem:y_independence} below.

The following two lemmas are a direct consequence of the recursive construction of the $\tau_n$'s and $\epsilon_n$'s in \eqref{eq:appAtaunepsn}.
\begin{lemma}
    \label{lem:tau_independence_past}
    For all $(t,y,w) \in [0,T] \times (a,b) \times \mathcal{C}_0^0([0,T])$, $n \in \mathbb{N}$ and $t_1 \in [0,t]$,
    $\tau_n(t,y,w) = \tau_n \left(t,y,\left(w_{u \vee t_1} - w_{t_1}\right)_{u \in [0, T]}\right)$ and $\epsilon_n(t,y,w) = \epsilon_n \left(t,y,\left(w_{u \vee t_1} - w_{t_1}\right)_{u \in [0, T]}\right)$.
\end{lemma}

%\begin{proof}
%    This is a direct consequence of the recursive construction of the $\tau_n$'s and $\epsilon_n$'s, see in particular \eqref{eq:appAtaunepsn}, and the fact that, for $r>t$, $w_{r} -w_{t} = w_{r \vee t_1} - w_{t \vee t_1} = (w_{r \vee t_1} - w_{t_1} )- (w_{t \vee t_1} - w_{t_1})$.
    %Let $(t,y,w) \in [0,T] \times (a,b) \times \mathcal{C}_0^0([0,T])$, $n \in \mathbb{N}$ and $t_1 \in [0,t]$. For all $r > \tau_{n}(t,y,w)$, we have . The construction of $\tau_{n+1}$ and $\epsilon_{n+1}$ is thus exactly the same.
%\end{proof}

\begin{lemma}
    \label{lem:tau_independence_future}
    For all $(t,s,y,w) \in [0,T] \times [0,T] \times (a,b) \times \mathcal{C}_0^0([0,T])$, $n \in \mathbb{N}$ and $t_2 \in [s, T]$,
    \begin{equation*}
        (\tau_n(t,y,w), \epsilon_n(t,y,w)) \mathds{1}_{\{\tau_n(t,y,w) \leqslant s\}} =
        (\tau_n(t,y,w_{\cdot \wedge t_2}), \epsilon_n(t,y,w_{\cdot \wedge t_2})) \mathds{1}_{\{\tau_n(t,y,w_{\cdot \wedge t_2}) \leqslant s\}}.
    \end{equation*}
\end{lemma}

The next lemma states that if the driving function is a Brownian motion with bounden drift, the number of jumps of $Y$ has finite moments. We already know that is is finite thanks to Lemma \ref{lem:finite_activity}.

\begin{lemma}
    \label{lem:finite_activity_expectation}
    Let $(t,y) \in [0,T] \times (a,b)$ and $(\Omega, \mathcal{F}, (\mathcal{F}_t)_{t \in [0,T]}, \mathbb{P})$ be a filtered probability space supporting an $(\mathcal{F}_t)$-Brownian motion $W$ on $[0,T]$ and $\beta_t$ a uniformly bounded stochastic process on this space. Let $\sigma > 0$, and define $\tilde{W}_t \defeq W_t + \int_0^t \beta_s \diff s$ and $N \defeq \inf\{n \in \mathbb{N}: \tau_i(t,y,\sigma\tilde{W}) = T\}$. Then $\mathbb{E}[N^k]<\infty$ for all $k \in \mathbb{N}^*$.
\end{lemma}

\begin{proof}
    \begin{description}
        \item[Case 1: $\beta \equiv 0$.] %Let $(\Omega', \mathcal{F}', (\mathcal{F}'_t)_{t \in [0,2T]}, \mathbb{P}')$ be a probability space supporting a Brownian motion $W'$ on $[0,2T]$. We denote its restriction to $[0,T]$ again by $W'$. 
        Let $n \in \mathbb{N}$. For concision, we write $\tau_n \defeq \tau_n(t,y,\sigma W)$.
        %Then $\mathbb{P}(\tau_n(t,y,\sigma W)<T) = \mathbb{P}'(\tau_n(t,y,\sigma W')<T)$. For concision, we define $\tau_n' \defeq \tau_n(t,y,\sigma W')$. 
        By the fourth property of $\tau_i$ in \eqref{eq:appA1},
        \begin{align*}
            \mathbb{P}(\tau_{n+1} < T)
            & \leqslant \mathbb{E}\left[\mathds{1}_{\{\tau_n < T\}}\mathbb{P}\left(\left.\sup\limits_{s \in[\tau_n,T]} \sigma|W_s - W_{\tau_n}| \geqslant \Delta\right| \mathcal{F}_{\tau_n} \right)\right] \\
           % & \leqslant \mathbb{E}'\left[\mathds{1}_{\{\tau_n' < T\}}\mathbb{P}'\left(\left.\sup\limits_{s \in[0,T]} \sigma|W'_{s + \tau_n'} - W'_{\tau_n'}| \geqslant \Delta\right| \mathcal{F}'_{\tau_n'} \right)\right]\\
            & \leqslant \mathbb{E}\left[\mathds{1}_{\{\tau_n < T\}}\mathbb{P}\left(\sup\limits_{s \in[0,T]} \sigma|W_{s}| \geqslant \Delta \right)\right]\\
            & = a \mathbb{P}(\tau_{n}< T)
        \end{align*}
        where $a \defeq \mathbb{P}\left(\sup\limits_{s \in[0,T]} \sigma|W_{s}| \geqslant \Delta \right) < 1$. The second inequality is a consequence of  the strong independence and invariance property of the Brownian motion. 
        
        By induction, we have for all $n \in \mathbb{N}$,
        $\mathbb{P}(\tau_{n}< T) \leqslant a^{n}$. Let $k \in \mathbb{N}^*$. We conclude
        \begin{equation*}
            \mathbb{E}[N^k] = \sum_{n=0}^{\infty} \mathbb{P}(N^k > n) = \sum_{n=0}^{\infty} \mathbb{P}\left(\tau_{\lfloor n^{\frac{1}{k}} \rfloor}< T\right) \leqslant \sum_{n=0}^{\infty} a^{\lfloor n^{\frac{1}{k}}\rfloor} < \infty.
        \end{equation*}
        \item[Case 2: general $\beta$.] Let $m > 0$ be a constant such that $|\beta| \leqslant m$. Define
        \begin{equation*}
            Z \defeq \exp\left(-\int_0^T \beta_s \diff W_s - \frac{1}{2} \int_0^T \beta_s^2 \diff s\right).
        \end{equation*}
        Since $\beta$ is bounded, by Novikov's criterion $\frac{\diff \mathbb{Q}}{\diff \mathbb{P}} = Z$ defines a change of probability. By Girsanov's theorem, $\tilde{W}$ is a brownian motion under $\mathbb{Q}$. Let $n \in \mathbb{N}$. Using the Cauchy-Schwarz inequality and the results from case 1,
        \begin{align*}
            \mathbb{P}\left(\tau_n(t,y,\sigma \tilde{W}) < T\right)
            &= \E[\mathbb{Q}][Z^{-1} \mathds{1}_{\{\tau_n(t,y,\sigma \tilde{W}) < T\}}]   \\
            &\leqslant  \sqrt{\E[\mathbb{Q}][Z^{-2}]}  \sqrt{\mathbb{Q}(\tau_n(t,y,\sigma \tilde{W}) < T)}\\
           % \mathbb{P}\left(\tau_n(t,y,\sigma \tilde{W}) < T\right)
           % &\leqslant  (\sqrt{a})^n\exp\left(\int_0^T \beta_s^2 \diff s\right)\sqrt{\E[\mathbb{Q}][\exp\left(\int_0^T \beta_s \diff W_s - \frac12\int_0^T \beta_s^2 \diff s\right)]}\\
            &\leqslant  (\sqrt{a})^n\exp\left(\frac12 m^2T\right).
        \end{align*}
        We conclude in the same way as in case 1.
    \end{description}
\end{proof}

\subsection{Properties of the mid-price function}

\begin{lemma}
    \label{lem:y_independence}
    For all $(t,s,y,w) \in [0,T] \times [0,T] \times (a,b) \times \mathcal{C}_0^0([0,T])$, $t_1 \leqslant t$ and $t_2 \geqslant s$,
    \begin{equation*}
        Y(t,s,y,w) = Y \left(t,s,y,\left(w_{(u\wedge t_2) \vee t_1} - w_{t_1}\right)_{u \in [0,T]}\right).
    \end{equation*}
\end{lemma}

\begin{proof}
    This is a direct consequence of Lemmas \ref{lem:tau_independence_past} and \ref{lem:tau_independence_future} and the definition of $Y$ \eqref{eq:y}.
\end{proof}

This lemma implies the following proposition regarding the adaptedness of the process $Y$, which is used throughout our study.
\begin{proposition}\label{prop:Yadapted}
    Let $(t, y) \in [0,T] \times (a,b)$. Let $(\Omega, \mathcal{F})$ a measurable space and $X = (X_r)_{r \in [0,T]}$ a process on this space. Let $\mathcal{F}^{X,t}$ the natural filtration of $(X_{r \vee t} - X_t)_{r \in [0,T]}$. Then, $Y(t,\cdot,y,X)$ is a $\mathcal{F}^{X, t}$-adapted process.
\end{proposition}

\section{Linear parabolic PDEs with nonlocal boundary conditions}
\label{sec:linear_nonlocal_pde}

We work on a nonempty interval $(a,b)$. Let $T > 0$ and $a_0, b_0 \in (a,b)$. We define, for this section $\D \defeq C([0,T] \times [a,b]) \cap C^{1,2}([0,T) \times (a,b))$.

We shall use some results from \parencite{friedman_partial_1983}. To better suit our framework, we work with final conditions instead of initial ones. All results remain the same, by using the transformation $t \mapsto T-t$, with the time derivatives $\partial_t$ being replaced by $-\partial_t$.

Let $\rho_1, \rho_2 \in (0,\infty)$. Let $\alpha$, $\beta$, $\gamma$, $f$ be real-valued functions defined on $[0,T) \times (a,b)$ and $u_T$ be a real-valued function defined on $[a,b]$. We are interested in equations, for $u \in \D$, of the form
\begin{equation}
    \label{eq:linear_equation_general}
    \left\{
        \begin{array}{rl}
            \partial_t u(t,y) + \alpha(t,y)\partial^2_{yy}u(t,y)+ \beta(t,y)\partial_y u (t,y) + \gamma(t,y) u(t,y) = f(t,y) & \text{on } [0,T) \times (a,b) \\
            u(T,y) = u_T(y,) & y \in [a,b]\\
            u(t,a) = \rho_1 u(t,a_0), & t \in [0,T]\\
            u(t,b) = \rho_2 u(t,b_0), &  
            t \in [0,T].
        \end{array}
    \right.
\end{equation}

The following lemma follows from straightforward computations.
\begin{lemma}
    \label{lem:equivalence_linear_pde}
    Let $u \in \D$. Let $\psi:[a,b]\to (0, \infty)$ be a function in $C^2$. Let $m \in \mathbb{R}$, and define $\phi(t,y):=e^{m(t-T)}\psi (y)$ and $v \defeq \phi u \in \D$. Then, $u$ solves \eqref{eq:linear_equation_general} if and only if $v$ solves 
    \begin{equation}
        \label{eq:linear_equation_modified}
        \left\{
        \begin{array}{rl}
            \partial_t v(t,y) + \alpha(t,y)\partial^2_{yy}v(t,y)+ \tilde{\beta}(t,y)\partial_y v (t,y) + \tilde{\gamma}(t,y) v(t,y) = \tilde{f}(t,y) & \text{on } [0,T) \times (a,b) \\
            v(T,y) = \tilde{u}_T(y), &
            y \in [a,b]\\
            v(t,a) = \tilde{\rho}_1 v(t,a_0), &  t \in [0,T]\\
            v(t,b) = \tilde{\rho}_2 v(t,b_0), & t \in [0,T]
        \end{array}
        \right.
    \end{equation}
    with
    \begin{align*}
        &\tilde{\beta} = \beta - 2\alpha\frac{\psi'}{\psi},\ \tilde{\gamma} = \gamma - m + 2 \alpha \left(\frac{\psi'}{\psi}\right)^2 - \alpha \frac{\psi''}{\psi} - \beta \frac{\psi'}{\phi},\ \tilde{f} = \phi f,\\
        &\tilde{u}_T = \psi u_T,\ \tilde{\rho}_1 = \rho_1\frac{\psi(a)}{\psi(a_0)},\ \tilde{\rho}_2 = \rho_2\frac{\psi(b)}{\psi(b_0)}.
    \end{align*}
\end{lemma}

\subsection{Uniqueness results}
\label{subsec:uniqueness_results}

In this section we fix three bounded continuous functions $\alpha, \beta, \gamma$ on $[0,T]\times [a,b]$ such that $\alpha \geqslant 0$. The following results are based on the maximum principle.

\begin{lemma}
\label{lem:linear_uniqueness}
Let $\rho_1, \rho_2 \in (0,1)$ and assume $f= u_T \equiv 0$. Suppose that $\gamma < 0$ and $u \in \D$ solves \eqref{eq:linear_equation_general}.
Then $u \equiv 0$.
\end{lemma}

\begin{proof}
Suppose that there exists $(t,y) \in [0,T] \times [a,b]$ such that $u(t,y) > 0$.
Then by continuity of $u$, there exists $(t^*, y^*) \in [0,T] \times [a,b]$ such that $u(t^*,y^*) = \max u > 0$.

Since $u(T,  \cdot) = 0$, necessarily $t^* < T$. Following \parencite{baldacci_bid_2020}, we show that $y^*\in(a,b)$. Indeed, if $y^*=a$, $u(t^*,a_0) = \frac{1}{\rho_1}u(t^*, y^*)> u(t^*,y^*)$, which contradicts the definition of $(t^*,y^*)$. Hence, $y^*\neq a$. 
A similar argument shows that $y^* \neq b$.

Then, we have $\partial_t u(t^*,y^*) \leqslant 0$, $\partial_y u(t^*, y^*)=0$, $\partial^2_{yy}u(t^*, y^*)\leqslant 0$.
Thus, $\gamma(t^*, y^*) u(t^*, y^*) \geqslant 0$. Since $\gamma < 0$, $u(t^*,y^*) \leqslant 0$, which is a contradiction.

Similar arguments can be carried out to find a contradiction if $u(t,y) < 0$ for some $(t,y)$.
\end{proof}

We now remove the hypotheses $\gamma < 0$ and $\rho_1, \rho_2 \in (0,1)$ from the previous lemma.

\begin{proposition}
\label{prop:linear_uniqueness}
Let $\rho_1, \rho_2 \in (0,\infty)$ and assume $f= u_T \equiv 0$. Suppose that $u \in \D$ solves \eqref{eq:linear_equation_general}.
Then $u \equiv 0$.
\end{proposition}

\begin{proof}
Let $N = \max\left(\frac{\rho_1}{(a_0-a)(b-a_0)},\frac{\rho_2}{(b_0-a)(b-b_0)}\right)$. Define $\psi:y \in[a,b] \mapsto 1+N(y-a)(b-y)$.
Consider $m > \max\left(\gamma + 2 \alpha \left(\frac{\psi'}{\psi}\right)^2 - \alpha \frac{\psi''}{\psi} - \beta \frac{\psi'}{\phi}\right) $ and define $\phi(t,y) := e^{m(t-T)}\psi(t,y)$. Let $v \defeq \phi u$ and take $\tilde{\beta}$, $\tilde{\gamma}$, $\tilde{f}$, $\tilde{u}_T$, $\tilde{\rho}_1$ and $\tilde{\rho}_2$ as in Lemma \ref{lem:equivalence_linear_pde}. Then, $\tilde{\rho}_1$, $\tilde{\rho}_2 \in (0,1)$, $\tilde{\gamma} < 0$, and $v$ solves \eqref{eq:linear_equation_modified}. By Lemma \ref{lem:linear_uniqueness}, $v \equiv 0$, and consequently $u \equiv 0$.
\end{proof}

\subsection{Existence results}

For all the PDEs mentioned in this section, uniqueness holds by the results of Section \ref{subsec:uniqueness_results}.

\begin{proposition}
    \label{prop:linear_existence}
    Let $\rho_1, \rho_2 \in (0,1)$ and $\delta \in (0,1)$. Consider $\alpha, \beta, \gamma, f$ four continuous functions on $(0,T)\times (a,b)$ such that $\min \alpha > 0$ and $|\alpha|_{\delta},|\beta|_{\delta},|\gamma|_{\delta}, |f|_{\delta}< \infty$. Let $u_T$ be a bounded continuous function on $[a,b]$ such that $u_T(a) = \rho_1 u_T(a_0)$ and $u_T(b) = \rho_2 u_T(b_0)$. Assume $\gamma \leqslant 0$. Then, there exists $u \in \D$ solving \eqref{eq:linear_equation_general}.
    
    Furthermore, if $|\alpha|_{\delta},|\beta|_{\delta},|\gamma|_{\delta} \leqslant K_1$, for some $K_1>0$. Then, there exists a constant $C > 0$, depending on the domain, $\delta$, $K_1$, $\rho_1$, $\rho_2$ and $\min \alpha$, but not on $f$ and $u_T$, such that
    \begin{equation*}
        |u|_{2+\delta} \leqslant C(|u_T|_{\infty} + |f|_{\delta}).
    \end{equation*}
\end{proposition}

\begin{proof}
    The proof is almost identical to the one of \parencite[Theorem 2.1]{friedman_monotonic_1986}.
    Let $\epsilon > 0$ and extend $\alpha, \beta,\gamma,f$ to $(-\epsilon, T) \times (a,b)$ such that$|\alpha|_{\delta},|\beta|_{\delta},|\gamma|_{\delta}, |f|_{\delta}< \infty$ still holds in $(-\epsilon, T) \times (a,b)$.
    We define $u_1(t,y) \defeq u_T(y)$. Suppose we have built $u_1,\dots,u_n$. By \parencite[Chapter 3, Theorem 9]{friedman_partial_1983}, there exists a unique $C([-\epsilon,T]\times [a,b]) \cap C^{1,2}((-\epsilon,T)\times (a,b))$ solution of
    \begin{equation*}
        \left\{
            \begin{array}{rl}
                \partial_t u(t,y) + \alpha(t,y)\partial^2_{yy}u(t,y)+ \beta(t,y)\partial_y u (t,y) + \gamma(t,y) u(t,y) = f(t,y) & \text{on } (-\epsilon,T) \times (a,b) \\
                u(T,y) = u_T(y) & y \in [a,b]\\
                u(t,a) = \rho_1 u_n(t,a_0) &  t \in (-\epsilon,T]\\
                u(t,b) = \rho_2 u_n(t,b_0) &  t \in (-\epsilon,T]
            \end{array},
        \right.
    \end{equation*}
    that we denote by $u_{n+1}$. Then, for every $n \geqslant 2$, $u_{n+1} - u_n$ solves 
    \begin{equation*}
        \left\{
            \begin{array}{rl}
                \partial_t u(t,y) + \alpha(t,y)\partial^2_{yy}u(t,y)+ \beta(t,y)\partial_y u (t,y) + \gamma(t,y) u(t,y) = 0 & \text{on } (-\epsilon,T) \times (a,b) \\
                u(T,y) = 0 &  y \in [a,b]\\
                u(t,a) = \rho_1 (u_n-u_{n-1})(t,a_0) & t \in (-\epsilon,T]\\
                u(t,b) = \rho_2 (u_n-u_{n-1})(t,b_0) & t \in (-\epsilon,T]
            \end{array},
        \right.
    \end{equation*}
    By the maximum principle \parencite[Chapter 2, equation (3.8)]{friedman_partial_1983}, $|u_{n+1} - u_n|_{\infty} \leqslant \rho |u_n-u_{n-1}|_{\infty}$, where $\rho \defeq \min(\rho_1,\rho_2) \in (0,1)$. By induction, for all $n \geqslant 2$,
    \begin{equation}\label{eq:appB1}
        |u_n-u_{n-1}|_{\infty} \leqslant \rho^{n-2}|u_2-u_1|_{\infty}.
    \end{equation}
    Thus, $(u_n)_n$ is a Cauchy sequence in $C([-\epsilon,T]\times [a,b])$, which converges uniformly to a function $u \in C([-\epsilon,T]\times [a,b])$. It is clear that $u$ verifies the boundary conditions. Additionally, thanks to \eqref{eq:appB1},
    \begin{equation}\label{eq:appB2}
        |u_n|_{\infty} \leqslant \sum_{i=2}^n|u_i-u_{i-1}|_{\infty} + |u_1|_{\infty} \leqslant \frac{1}{1-\rho}|u_2 - u_1|_{\infty} + |u_1|_{\infty},\quad n\geqslant 2.
    \end{equation}
    By the maximum principle \parencite[Chapter 2, equation (3.10)]{friedman_partial_1983}, there exists a constant $C$, depending only on $\alpha$ and $\beta$, such that $|u_2|_{\infty} \leqslant |u_1|_{\infty} + C|f|_{\infty}$. Given that $|u_1|_{\infty} = |u_T|_{\infty}$, we have, by \eqref{eq:appB2}, the existence of a constant $C'$, depending only on $\alpha,\beta,\rho$, such that 
    \begin{equation}\label{eq:appB3}
    |u_n|_{\infty} \leqslant C'(|u_T|_{\infty} + |f|_{\infty}).
    \end{equation}
    
    Let $K$ be a compact included in $(-\epsilon,T) \times (a,b)$. By \parencite[Chapter 3, Theorem 5]{friedman_partial_1983}, there exists a constant $D$, independent of $n$, such that for all $n \geqslant 2$, the Hölder norms on $K$ of $u_n$, $\partial_t u_n$, $\partial_y u_n$ and $\partial^2_{yy} u_n$ are bounded by $D(|u_n|_{\infty} + |d^2 f|_{\delta})$ and thus, thanks to \eqref{eq:appB3}, by $D(C'|u_T|_{\infty} + C' |f|_{\infty} + |d^2f|_{\delta})$.

    Thus, by the Arzelà-Ascoli theorem, there exists a subsequence $(u_{n_k})_k$ such that $(u_{n_k})_k$, $(\partial_t u_{n_k})_k$, $(\partial_y u_{n_k})_k$ and $(\partial^2_{yy} u_{n_k})_k$ converge uniformly on $K$. Finally, we have that $u \in C^{1,2}(K)$ and, uniformly on $K$,
    \begin{equation*}
        u_{n_k}\to u, \quad \partial_t u_{n_k}\to \partial_t u,
        \quad \partial_y u_{n_k}\to \partial_y u,
        \quad \partial^2_{yy} u_{n_k}\to \partial^2_{yy} u.
    \end{equation*}

    We conclude that the restriction of $u$ to $[0,T) \times [a,b]$, verifies the equation \eqref{eq:linear_equation_general}.

    Now, suppose that $|\alpha|_{\delta},|\beta|_{\delta},|\gamma|_{\delta} \leqslant K_1$. By \parencite[Chapter 3, Theorem 5]{friedman_partial_1983} and \eqref{eq:appB3}, there exists a constant $C''$ depending only on the domain, $K_1$, $\delta$, $\rho_1$, $\rho_2$, $\min \alpha$ such that for all $n \geqslant 2$,
    \begin{equation*}
        |u_n|_{2+\delta} \leqslant C''\left(|u_n|_{\infty} + |f|_{\delta}\right)
        \leqslant C''\left(C'|u_T|_{\infty} + C'|f|_{\infty} + |f|_{\delta}\right)
        \leqslant C''(1+C')(|u_T|_{\infty} + |f|_{\delta}).
    \end{equation*}

    Hence, for all $P,Q \in [0,T) \times (a,b)$ and $n \geqslant 2$ such that $P \neq Q$,
    \begin{equation*}
        d^{2+\delta}_{PQ}\frac{|\partial^2_{yy}u_n(P) - \partial^2_{yy}u_n(Q)|}{d(P,Q)^{\delta}} \leqslant |u_n|_{2+\delta} \leqslant C''(1+C')(|u_T|_{\infty} + |f|_{\delta}).
    \end{equation*}
    Since $\partial_{yy}^2 u_{n_k} \to \partial_{yy}^2 u$ pointwise, taking the limit $n \to \infty$ and the supremum over $P,Q$, yields
    \begin{equation*}
        \sup_{P \neq Q} d^{2+\delta}_{PQ}\frac{|\partial^2_{yy}u(P) - \partial^2_{yy}u(Q)|}{d(P,Q)^{\delta}} \leqslant C'''(|u_T|_{\infty} + |f|_{\delta}).
    \end{equation*}
    Similarly, $|d^2\partial_{yy}^2 u|_{\infty} \leqslant  C''(1+C')(|u_T|_{\infty} + |f|_{\delta})$. Analogous arguments are valid for the other derivatives. Hence, $|u|_{2+\delta} <  8C''(1+C')(|u_T|_{\infty} + |f|_{\delta})$.
\end{proof}

Now we remove the assumptions $\gamma \leqslant 0$ and $\rho_1,\rho_2 < 1$ from the previous result.
\begin{corollary}
    \label{corol:linear_existence}
    Let $\rho_1, \rho_2 \in (0,\infty)$ and $\delta \in (0,1)$. Consider $\alpha, \beta, \gamma, f$ four continuous functions on $(0,T)\times (a,b)$ such that $\min \alpha > 0$ and $|\alpha|_{\delta},|\beta|_{\delta},|\gamma|_{\delta}, |f|_{\delta}< \infty$. Let $u_T$ be a bounded continuous function on $[a,b]$ such that $u_T(a) = \rho_1 u_T(a_0)$ and $u_T(b) = \rho_2 u_T(b_0)$. Then, there exists $u \in \D$ solving \eqref{eq:linear_equation_general}.
    
    Furthermore, if $|\alpha|_{\delta},|\beta|_{\delta},|\gamma|_{\delta} \leqslant K_1$, for some $K_1>0$. Then, there exists a constant $C > 0$, depending on the domain, $\delta$, $K_1$, $\rho_1$, $\rho_2$ and $\min \alpha$, but not on $f$ and $u_T$, such that
    \begin{equation*}
        |u|_{2+\delta} \leqslant C(|u_T|_{\infty} + |f|_{\delta}).
    \end{equation*}
\end{corollary}

\begin{proof}
    Consider the functions $\psi$, $m$ and $\phi$ as in the proof of Proposition \ref{prop:linear_uniqueness}.
   Define $\tilde{\beta}$, $\tilde{\gamma}$, $\tilde{f}$, $\tilde{u}_T$, $\tilde{\rho}_1$ and $\tilde{\rho}_2$ as in Lemma \ref{lem:equivalence_linear_pde}. 
    
    Then $\tilde{\rho}_1, \tilde{\rho}_2 \in (0,1)$, $|\tilde{\beta}|_{\delta}, |\tilde{\gamma}|_{\delta}, |\tilde{f}|_{\delta} < \infty$ and $\tilde{\gamma} \leqslant 0$.
    By Proposition \ref{prop:linear_existence}, there exists $v \in \D$ solving \eqref{eq:linear_equation_modified}, and by Lemma \ref{lem:equivalence_linear_pde},  $u\defeq \frac{v}{\phi}$ solves \eqref{eq:linear_equation_general}.

    Now suppose that $|\alpha|_{\delta},|\beta|_{\delta},|\gamma|_{\delta} \leqslant K_1$ and $|f|_{\delta} < \infty$. Let $K_2 \defeq \max \left\{\left|\frac{\partial_y\phi}{\phi}\right|_{\delta},\left|\left(\frac{\partial_y\phi}{\phi}\right)^2\right|_{\delta}, \left|\frac{\partial^2_{yy}\phi}{\phi}\right|_{\delta}\right\}$, which depends only on $a_0$, $b_0$, $\rho_1$, $\rho_2$, and the domain.
    
    We chose $m = K_1\left(1+4K_2\right)+1$, which is still greater than $\max\left(\gamma + 2 \alpha \left(\frac{\psi'}{\psi}\right)^2 - \alpha \frac{\psi''}{\psi} - \beta \frac{\psi'}{\psi}\right)$. We have
    \begin{align*}
        \left|\tilde{\beta}\right|_{\delta} &\leqslant 
        \left|\beta\right|_{\delta} + 2 \left|\alpha\right|_{\delta}\left|\frac{\partial_y \phi}{\phi}\right|_{\infty} + 2 \left|\alpha\right|_{\infty}\left|\frac{\partial_y \phi}{\phi}\right|_{\delta}
        \leqslant K_1(1+4K_2),\\
        \left|\tilde{\gamma}\right|_{\delta} &\leqslant 
        \left|\gamma\right|_{\delta} + m + 2 \left|\alpha\left(\frac{\partial_y\phi}{\phi}\right)^2\right|_{\delta}
        + \left|\alpha\frac{\partial^2_{yy} \phi}{\phi}\right|_{\delta}
        + \left|\beta\frac{\partial_y \phi}{\phi}\right|_{\delta}
        \leqslant K_1(2+12K_2)+1.
    \end{align*}

    By Proposition \ref{prop:linear_existence}, there exists $C > 0$ depending only on the domain, $K_1$, $\rho_1$, $\rho_2$, $\delta$, $\min \alpha$ such that $|v|_{2+\delta} \leqslant C(|\tilde{u}_T|_{\infty} + |\tilde{f}|_{\delta})$ ($C$ is chosen by replacing $K_1$ by $K_1(2+12K_2)+1$ in Proposition \ref{prop:linear_existence}, which only depends on the mentioned constants). Hence,
    \begin{equation*}
        |v|_{2+\delta} \leqslant C(|\psi|_{\infty}|u_T|_{\infty} + |\phi|_{\infty}|f|_{\delta} + |f|_{\infty}|\phi|_{\delta}) \leqslant 3C|\phi|_{\delta}(|u_T|_{\infty} + |f|_{\delta}).
    \end{equation*}
    Furthermore, since $u = \frac{v}{\phi}$, there exists a constant $C' > 0$ depending only on the domain, $|\phi|_{2+\delta}$ and $\left|\frac{1}{\phi}\right|_{\infty}$ (hence only on the domain, $a_0$, $b_0$, $\rho_1$ and $\rho_2$) such that $|u|_{2+\delta} \leqslant C'|v|_{2+\delta}$. We conclude that
    \begin{equation*}
        |u|_{2+\delta} \leqslant 3CC'|\phi|_{\delta}(|u_T|_{\infty} + |f|_{\delta}).
    \end{equation*}
\end{proof}

\subsection{Probabilistic representation}

In this section we derive in Proposition \ref{prop:feynmankac} a Feynman-Kac-type formula which is useful to establish the results in Section \ref{subsec:continuity} and the a priori estimates in Corollary \ref{corol:bound_on_u}. Let $a < b$, $a_0, b_0 \in (a,b)$. Let $(\Omega, \mathcal{F}, (\mathcal{F}_t)_{t\in [0,T]}, \mathbb{P})$ be a filtered probability space supporting a Brownian motion $W$ and let $Y$ be the mapping built in Appendix \ref{sec:midprice_construction}. Let $\sigma > 0$ and $\beta \in \mathbb{R}$.

For $(t,y) \in [0,T] \times (a,b)$, define $\Y \defeq Y\left(t, \cdot, y, (\beta s + \sigma W_s)_{s \in [0,T]}\right)$, which is a right-continuous adapted process, due to Proposition \ref{prop:Yadapted}.

\begin{proposition}
    \label{prop:feynmankac}
    Let $\gamma, f:[0,T) \times (a,b)\to \mathbb{R}$ be two bounded measurable functions. Let $\beta \in \mathbb{R}$ and $u \in \D$. Suppose that $u$ solves
    \begin{equation*}
        \left\{
            \begin{array}{rl}
                \partial_t u(t,y) + \frac{\sigma^2}{2}\partial^2_{yy}u(t,y)+ \beta\partial_y u (t,y) + \gamma(t,y) u(t,y) + f(t,y) = 0 & \text{on } [0,T) \times (a,b) \\
                u(t,a) = u(t,a_0) & t \in [0,T]\\
                u(t,b) = u(t,b_0) & t \in [0,T]
            \end{array}.
        \right.
    \end{equation*}
    Then, for every $(t,y) \in [0,T] \times (a,b)$ and $(\mathcal{F}_t)$-stopping time $\tau$ taking values in $[t,T]$,
    \begin{equation*}
        u(t,y) = \mathbb{E}\left[u(\tau,\Y_{\tau})e^{\int_t^{\tau}\gamma(r, \Y_r)\diff r} + \int_t^{\tau}f(s, \Y_s) e^{\int_t^{s}\gamma(r, \Y_r)\diff r}  \diff s\right].
    \end{equation*}
\end{proposition}

\begin{proof}
    Let  $(t,y) \in [0,T] \times (a,b)$ and $\tau$ be an $(\mathcal{F}_t)$-stopping time taking values in $[t,T]$. For $i \in \mathbb{N}$, we denote by $\tau_i \defeq \tau_{i}\left(t,y,(\beta s + \sigma W_s)_s\right) \wedge \tau$ the barrier reaching time of $\Y$, which is defined in Appendix \ref{sec:midprice_construction} and is a $(\mathcal{F}_t)$-stopping time. Let $k_0 \in \mathbb{N}^*$ such that $\min\{y - a, b - y, a_0-a, b_0-a, b-a_0,b-b_0\} > \frac{1}{k_0}$. For $k \geqslant k_0$, $i \in \mathbb{N}$, define the stopping time
    \begin{equation*}
        \tau_i^{(k)} \defeq \inf\left\{s \in (\tau_i, \tau_{i+1}) : \Y_s \in\left\{a + \frac{1}{k}, b - \frac{1}{k}\right\}\right\}\wedge\left(T - \frac{1}{k}\right) \wedge \tau_{i+1}.
    \end{equation*}
    By the intermediate value theorem, we have that for all $k \geqslant k_0$ and  $i \in \mathbb{N}$, $\tau_i \leqslant \tau_i^{(k)} \leqslant \tau_{i+1}$ and for all $s \in [\tau_i, \tau_{i}^{(k)}]$, $(s,\Y_s) \in \left[0, T - \frac{1}{k}\right] \times \left[a + \frac{1}{k}, b - \frac{1}{k}\right]$. Furthermore, for fixed $i$, $(\tau_i^{(k)})_{k \geqslant k_0}$ is nondecreasing.

    Let $N \defeq \min\{i \in \mathbb{N} : \tau_i  = \tau\}$ which is finite thanks to Lemma \ref{lem:finite_activity}. 
    Fix $k \geqslant k_0$. Then,
    \begin{equation*}
        \begin{split}
            e^{\int_t^{\tau} \gamma(s, \Y_s) \diff s}u(\tau, \Y_{\tau}) &- u(t,y) =
        \sum_{i=0}^{N-1} \left(e^{\int_t^{\tau^{(k)}_{i}} \gamma(s, \Y_s) \diff s}u\left(\tau^{(k)}_{i}, \Y_{\tau^{(k)}_{i}}\right) -  e^{\int_t^{\tau_i} \gamma(s, \Y_s) \diff s}u(\tau_i, \Y_{\tau_i})\right) \\
        &+
        \sum_{i=0}^{N-1} \left( e^{\int_t^{\tau_{i+1}} \gamma(s, \Y_s) \diff s}u\left(\tau_{i+1}, \Y_{\tau_{i+1}}\right) -  e^{\int_t^{\tau^{(k)}_{i}} \gamma(s, \Y_s) \diff s}u\left(\tau^{(k)}_{i}, \Y_{\tau^{(k)}_{i}}\right)\right).
        \end{split}
    \end{equation*}
    Now, using Itô's formula in each interval $[\tau_i, \tau_i^{(k)}]$, where $\Y_s = y + \beta (s-t) + \sigma(W_s - W_t) + \sum_{j=1}^i \epsilon_j'(t,y,(\beta r + W_r)_r)$,
    \begin{equation*}
        \begin{split}
            e^{\int_t^{\tau} \gamma(s, \Y_s) \diff s}&u(\tau, \Y_{\tau}) - u(t,y) \\
            &=
             \sum_{i=0}^{N-1} \left( e^{\int_t^{\tau_{i+1}} \gamma(s, \Y_s) \diff s}u\left(\tau_{i+1}, \Y_{\tau_{i+1}}\right) -  e^{\int_t^{\tau^{(k)}_{i}} \gamma(s, \Y_s) \diff s}u\left(\tau^{(k)}_{i}, \Y_{\tau^{(k)}_{i}}\right)\right) \\
            &+ \sum_{i=0}^{N-1} \left( \int_{\tau_i}^{\tau_i^{(k)}} e^{\int_t^{s} \gamma(r, \Y_r) \diff r}\left( \partial_t u + \beta \partial_y u + \frac{\sigma^2}{2} \partial_{yy}^2 u + \gamma(s,\Y_s)u \right) (s,\Y_s) \diff s \right) \\
            &+\sum_{i=0}^{N-1} \int_{\tau_i}^{\tau_i^{(k)}} e^{\int_t^{s} \gamma(r, \Y_r) \diff r} \partial_y u(s, \Y_s) \diff W_s.
        \end{split}
    \end{equation*}
    Using the fact that $u$ solves the PDE, we get
    \begin{equation}\label{eq:appB4}
        \begin{split}
            e^{\int_t^{\tau} \gamma(s, \Y_s) \diff s}&u(\tau, \Y_{\tau}) - u(t,y) \\
            &=
            -\int_{t}^{\tau}e^{\int_t^{s} \gamma(r, \Y_r) \diff r} f(s, \Y_s) 
            \left(1-\sum_{i=0}^{N-1}\mathds{1}_{[\tau_i^{(k)},\tau_{i+1}]}(s)\right)\diff s\\
             &+\sum_{i=0}^{N-1} \left( e^{\int_t^{\tau_{i+1}} \gamma(s, \Y_s) \diff s}u\left(\tau_{i+1}, \Y_{\tau_{i+1}}\right) -  e^{\int_t^{\tau^{(k)}_{i}} \gamma(s, \Y_s) \diff s}u\left(\tau^{(k)}_{i}, \Y_{\tau^{(k)}_{i}}\right)\right) \\
            &+\sum_{i=0}^{N-1} \int_{\tau_i}^{\tau_i^{(k)}} e^{\int_t^{s} \gamma(r, \Y_r) \diff r} \partial_y u(s, \Y_s) \diff W_s.
        \end{split}
    \end{equation}
    Since $\partial_t u$ is bounded on $\left[0, T - \frac{1}{k}\right] \times \left[a + \frac{1}{k}, b - \frac{1}{k}\right]$, the last term has 0 expectation. Hence, taking the expectation in \eqref{eq:appB4}
    \begin{equation*}
        \begin{split}
            \mathbb{E}\big[ e^{\int_t^{\tau} \gamma(s, \Y_s) \diff s}&u(\tau, \Y_{\tau})\big] - u(t,y) \\
            &=
            -\mathbb{E}\left[\int_{t}^{\tau}e^{\int_t^{s} \gamma(r, \Y_r) \diff r} f(s, \Y_s) 
            \left(1-\sum_{i=0}^{N-1}\mathds{1}_{[\tau_i^{(k)},\tau_{i+1}]}(s)\right)\diff s \right]\\
             &+\mathbb{E}\left[\sum_{i=0}^{N-1} \left( e^{\int_t^{\tau_{i+1}} \gamma(s, \Y_s) \diff s}u\left(\tau_{i+1}, \Y_{\tau_{i+1}}\right) -  e^{\int_t^{\tau^{(k)}_{i}} \gamma(s, \Y_s) \diff s}u\left(\tau^{(k)}_{i}, \Y_{\tau^{(k)}_{i}}\right)\right)\right].
        \end{split}
    \end{equation*}

    The integrand of the first term is bounded by $e^{|\gamma|_{\infty}T}|f|_{\infty}$ and converges to 0 almost everywhere as $k\to \infty$. Hence, by dominated convergence, the first term converges to 0 as $k\to \infty$. To complete the proof, it only remains to show that the second term tends to 0 as $k \to \infty$. 

    The term inside the expectation is bounded by $2N |u|_{\infty} e^{T|\gamma|_{\infty}}$, which is integrable thanks to Lemma \ref{lem:finite_activity_expectation}).

    One can show  that for every $i$ 
    \begin{equation*}
        \tau_i^{(k)}\xrightarrow[k \to \infty]{}\tau_{i+1}\text{ and } u\left(\tau^{(k)}_{i}, \Y_{\tau^{(k)}_{i}}\right)\xrightarrow[k \to \infty]{} u\left(\tau_{i+1}, \Y_{\tau_{i+1}}\right),
    \end{equation*}
    the last limit being a consequence of the boundary conditions. Hence, by dominated convergence, the second term also tends to 0 as $k \to \infty$.
\end{proof}

\begin{corollary}
    \label{corol:bound_on_u}
    Let $\gamma, f:[0,T) \times (a,b)\mapsto \mathbb{R}$ be two bounded measurable functions and $\beta \in \mathbb{R}$. Suppose that $u\in\D$ solves
    \begin{equation*}
        \left\{
            \begin{array}{rl}
                \partial_t u(t,x) + \frac{\sigma^2}{2}\partial^2_{xx}u(t,x)+ \beta\partial_x u (t,x) + \gamma(t,x) u(t,x) + f(t,x) = 0 & \text{on } (0,T] \times (a,b) \\
                u(T,y) = 0 &  y \in [a,b]\\
                u(t,a) = u(t,a_0) & t \in [0,T]\\
                u(t,b) = u(t,b_0) & t \in [0,T]
            \end{array}.
        \right.
    \end{equation*}
    Let $g$ be a measurable function on $[0,T]$ such that for almost every $s \in [0,T]$, $g(s) \geqslant \sup_{y \in (a,b)}|f(s,y)|$. Then, for all $t \in [0,T]$ and $c > \sup \gamma$,
    \begin{equation*}
        e^{c t}\sup_{y \in [a,b]}|u(t,y)| \leqslant
        \int_t^T e^{c s}g(s)\diff s.
    \end{equation*}
\end{corollary}

\subsection{The Krylov-Safonov estimates}

In this section, we recall the Krylov-Safonov estimates, which control the local Hölder norms of $u$ with respect to the $L^2$-norm of $f$ and $|u|_{\infty}$. For $D \subset [0,T] \times [a,b]$, we denote by $W^{1,2}_2(D)$ the space of measurable functions $u$ on $D$ admitting weak derivatives $\partial_t u$, $|\partial_y u|$ and $\partial^2_{yy}u$ such that $\int_D|u|^2$, $\int_D |\partial_t u|^2$, $\int_D |\partial_yu|^2$ and $\int_D |\partial^2_{yy}u|^2$ are all finite.

We recall the Krylov-Safonov estimate \parencite[Theorem 4.2]{krylov_certain_1981} (replacing the variable $t$ by $T-t$ to suit our framework).
\begin{theorem}
    Let $D$ an open subset of $[0,T] \times [a,b]$. Let $K > 0$ and $\alpha$, $\beta$, $\gamma$ be three bounded measurable functions such that $\frac{1}{K} \leqslant \alpha \leqslant K$, $|\beta| \leqslant K$ and $-K \leqslant \gamma \leqslant 0$. Then, there exist two constants $C > 0$, $\delta \in (0,1)$ depending on $K$ (but not on the specific subset $D$) such that for every $P, P' \in D$ with $d(P,P')\leqslant \frac{1}{4} \left(d^D_{PP'} \wedge 1 \right)$, and every $u \in W^{1,2}_2(D)$
    \begin{equation*}
        \left(d^D_{PP'} \wedge 1\right)^{\delta}\left|u(P) - u(P')\right| \leqslant C d(P,P')^{\delta} \left(|u|_{\infty} + \sqrt{\int_D |\partial_t u + \alpha \partial^2_{yy} u + \beta \partial_y u + \gamma u|^2}\right).
    \end{equation*}
\end{theorem}

Proceeding like in the proof of \textcite[Theorem 4.3]{krylov_certain_1981}, we deduce the following corollary, which shows that we can remove the hypotheses $\gamma\leqslant 0$ and $d(P,P')\leqslant \frac{1}{4} \left(d^D_{PP'} \wedge 1 \right)$, and obtain uniform estimates with respect to $\delta'\in(0,\delta]$.
\begin{corollary}
    \label{corol:get_rid_of_one_fourth_condition}
    Let $D$ an open subset of $[0,T] \times [a,b]$. Let $K > 0$ and $\alpha$, $\beta$, $\gamma$ be three bounded measurable functions such that $\frac{1}{K} \leqslant \alpha \leqslant K$, $|\beta| \leqslant K$ and $|\gamma| \leqslant K$. Then, there exist two constants $C > 0$, $\delta \in (0,1)$ depending on $K$, $a$, $b$, $T$ (but not on the specific subset $D$) such that for every $P, P' \in D$, $\delta' \in (0,\delta]$ and $u \in W^{1,2}_2(D)$
    \begin{equation*}
        \left(d^D_{PP'} \right)^{\delta'}\left|u(P) - u(P')\right| \leqslant C d(P,P')^{\delta'} \left(|u|_{\infty} + \sqrt{\int_D |\partial_t u + \alpha \partial^2_{yy} u + \beta \partial_y u + \gamma u|^2}\right).
    \end{equation*}
\end{corollary}

\begin{proof}
    First observe that, if $P,P' \in D$ satisfy $d(P,P')> \frac{1}{4}(d^D_{PP'} \wedge 1)$, then, $\left(d^D_{PP'} \wedge 1\right)^{\delta}\frac{\left|u(P) - u(P')\right|}{d(P,P')^{\delta}} \leqslant 2\cdot 4^{\delta}|u|_{\infty}$. Defining $C' \defeq \max\{C, 2\cdot 4^{\delta}\}\max\left\{1, T \wedge \frac{b-a}{2}\right\}^{\delta}> 2$ ($(C, \delta)$ given by the preceding theorem replacing $K$ by $2K$), thanks to the preceding theorem, and since $d^D_{PP'} \leqslant T \wedge \frac{b-a}{2}$, we have, for any $P,P' \in D$,
    \begin{equation*}
        \left(d^D_{PP'}\right)^{\delta}\left|u(P) - u(P')\right| \leqslant C' d(P,P')^{\delta} \left(|u|_{\infty} + \sqrt{\int_D |\partial_t u + \alpha \partial^2_{yy} u + \beta \partial_y u + (\gamma - K) u|^2}\right).
    \end{equation*}

    Using the Minkowski inequality on the $\sqrt{\int}$ term, we get
    \begin{equation*}
        \left(d^D_{PP'}\right)^{\delta}\left|u(P) - u(P')\right| \leqslant C' d(P,P')^{\delta} \left(|u|_{\infty} + K\sqrt{\int_D |u|^2} +\sqrt{\int_D |\partial_t u + \alpha \partial^2_{yy} u + \beta \partial_y u + \gamma u|^2}\right).
    \end{equation*}
    Then, setting $C'' \defeq C'(1+K\sqrt{T(b-a)}) > 2$ (which only depends on $K$, $b-a$ and $T$)
    \begin{equation*}
        \left(d^D_{PP'} \right)^{\delta}\left|u(P) - u(P')\right| \leqslant C'' d(P,P')^{\delta} \left(|u|_{\infty} + \sqrt{\int_D |\partial_t u + \alpha \partial^2_{yy} u + \beta \partial_y u + \gamma u|^2}\right).
    \end{equation*}
    
    Let $\delta' \in (0,\delta)$ and $P, P' \in D$ be such that $P \neq P'$. We have
    \begin{equation*}
        \frac{\left(d^D_{PP'} \right)^{\delta'}\left|u(P) - u(P')\right|}{d(P,P')^{\delta'}} \leqslant
        C'' \left(|u|_{\infty} + \sqrt{\int_D |\partial_t u + \alpha \partial^2_{yy} u + \beta \partial_y u + \gamma u|^2}\right) \frac{\left(d^D_{PP'} \right)^{\delta- \delta'}}{d(P,P')^{\delta-\delta'}}.
    \end{equation*}
If $d(P,P') > d^D_{PP'}$, then $\frac{\left(d^D_{PP'} \right)^{\delta- \delta'}}{d(P,P')^{\delta-\delta'}} \leqslant 1$. If $d(P,P') \leqslant d^D_{PP'}$, then $\frac{\left(d^D_{PP'} \right)^{\delta'}\left|u(P) - u(P')\right|}{d(P,P')^{\delta'}} \leqslant 2 |u|_{\infty}$. Thus, we have for any $P,P' \in D$,
    \begin{equation*}
        \frac{\left(d^D_{PP'} \right)^{\delta'}\left|u(P) - u(P')\right|}{d(P,P')^{\delta'}} \leqslant
        C'' \left(|u|_{\infty} + \sqrt{\int_D |\partial_t u + \alpha \partial^2_{yy} u + \beta \partial_y u + \gamma u|^2}\right)
    \end{equation*}
as desired.
\end{proof}

Now we state the result when a function $u \in \D$ solves a parabolic PDE. In this case $u$ is not necessarily in $W^{1,2}_{2}$ because its derivatives may be unbounded close to the boundary.

\begin{corollary}
    \label{corol:krylovsafonov}
    Let $K > 0$ and $\alpha$, $\beta$, $\gamma$ be three bounded measurable functions such that $\frac{1}{K} \leqslant \alpha \leqslant K$, $|\beta| \leqslant K$ and $|\gamma| \leqslant K$. Then, there exist two constants $C' > 0$, $\delta \in (0,1)$, depending on $K$, $a$, $b$ and $T$, such that for every $P, P' \in [0,T)\times (a,b)$, $\delta' \in (0,\delta]$ and $u \in \D$,
    \begin{equation*}
        \left(d_{PP'} \right)^{\delta'}\left|u(P) - u(P')\right| \leqslant C' d(P,P')^{\delta'} \left(|u|_{\infty} + \sqrt{\int_{[0,T)\times (a,b)} |\partial_t u + \alpha \partial^2_{yy} u + \beta \partial_y u + \gamma u|^2}\right).
    \end{equation*}
    In particular, there exists a constant $C$ depending only on $K$, $a$, $b$ and $T$ such that for every measurable function $f$ on $[0,T) \times (a,b)$, if $u \in \D$ solves
    \begin{equation*}
        \partial_t u + \alpha \partial^2_{yy} u + \beta \partial_y u + \gamma u = f
    \end{equation*}
    on $[0,T) \times (a,b)$, then for every $\delta' \in (0,\delta]$, $|u|_{\delta'} \leqslant C\left(|u|_{\infty} + |f|_{\infty}\right)$.
\end{corollary}

\begin{proof}
    Let $C$, $\delta$ be the constants given by corollary \ref{corol:get_rid_of_one_fourth_condition}. Let $P,P' \in [0,T) \times(a,b)$. Let $\varepsilon > 0$ such that $P,P' \in D^{\varepsilon} \defeq [0,T-\varepsilon) \times (a + \varepsilon , b - \varepsilon)$. Then,
    \begin{equation*}
        \left(d^{D^{\varepsilon}}_{PP'} \right)^{\delta'}\left|u(P) - u(P')\right| \leqslant C d(P,P')^{\delta'} \left(|u|_{\infty} + \sqrt{\int_{[0,T)\times (a,b)} |\partial_t u + \alpha \partial^2_{yy} u + \beta \partial_y u + \gamma u|^2}\right)
    \end{equation*}
    for all $\delta' \in (0,\delta]$. Since $\lim\limits_{\varepsilon \to 0}d^{D^{\varepsilon}}_{PP'} = d_{PP'}$, we get the desired result.
\end{proof}

\begin{rem}
    Note that in the previous corollary $\sqrt{\int_{(0,T)\times (a,b)} |\partial_t u + \alpha \partial^2_{yy} u + \beta \partial_y u + \gamma u|^2}$ might be infinite.
\end{rem}

\section{Convergence from the discrete inventory to the continuous inventory case}
\label{sec:convergence_discrete_continuous}
\begin{dummyenv}
\renewcommand{\Q}{\mathcal{Q}_n}
\renewcommand{\Y}{\mathcal{Y}}
In this section we prove Theorem \ref{thm:existence}(ii).
We keep all the notations from Section \ref{sec:hjb}. For $n \in \mathbb{N}^*$, define $N^a_n \defeq \frac{n}{\Qmax} \qa_n \in \mathbb{N}$, $N^b_n \defeq \frac{n}{\Qmax} \qb_n \in \mathbb{N}$, and the probability measures
\begin{align}
    \label{eq:discrete_measures_2}
    \begin{split}
        \mu^a_n & = \sum_{j = 0}^{N^a_n - 1}\delta_{\Qmax\frac{j}{n}}\mu^a\left(\left[\Qmax\frac{j}{n}, \Qmax\frac{j+1}{n}\right)\right) + \delta_{\qa_n}\mu^a\left(\left[\qa_n, \qa\right]\right) \\
    \mu^b_n & = \sum_{j = 0}^{N^b_n - 1}\delta_{\Qmax\frac{j}{n}}\mu^b\left(\left[\Qmax\frac{j}{n}, \Qmax\frac{j+1}{n}\right)\right) + \delta_{\qb_n}\mu^b\left(\left[\qb_n, \qb\right]\right)
    \end{split}
\end{align}
on $\Qa$ and $\Qb$, respectively. Notice that $\mu^a_n$ and $\mu^b_n$ can be also seen as measures on $\mathcal{Q}^{+,a}_{\infty}$ and $\mathcal{Q}^{+,b}_{\infty}$. Define $\mu^a_{\infty} \defeq \mu^a$ and $\mu^b_{\infty} \defeq \mu^b$.

For $n \in \mathbb{N}^*\cup \{\infty\}$, let $(u_n^Q)_{Q \in \Q}$ be the (unique by Theorem \ref{thm:verification}) continuous solution of \eqref{eq:HJB_interior}-\eqref{eq:HJB_border} with associated execution measures $\mu^a_n$ and $\mu^b_n$ such that $\sup_{Q \in \Q}|u^Q_n|_{2+\beta} < \infty$.

By Section \ref{subsubsec:uniform_majoration}, and since the considered solutions are unique, there exists a constant $C > 0$, depending on $T$, $\Qmax$, $\gamma$, $\sigma$, $\delta$, $\eta$, $\Lambda^*$, $\beta$, $\Lambda^*_{\beta}$ but not on $n$, nor the execution measures such that
\begin{equation*}
    \sup_{n \in \mathbb{N^*} \cup \{\infty\}} \sup_{Q \in \Q} |u^Q_n|_{2+\beta} \leqslant C.
\end{equation*}

By the results of Section \ref{subsec:continuity}, there exists a constant $L > 0$ depending on $T$, $\Qmax$, $\gamma$, $\sigma$, $\delta$, $\eta$, $\Lambda^*$, $C$ but not on $n$, nor the execution measures, such that for all $n \in \mathbb{N}^*\cup \{\infty\}$, $(t,y) \in [0,T] \times \bar{\Y}$, $L \varpi$ is a modulus of continuity of $Q \mapsto u^Q_n(t,y)$.

For $n \in \mathbb{N}^*\cup \{\infty\}$, and $(t,Q,y)\in [0,T) \times \Q \times Y$, define
\begin{align*}
    f_n^Q(t,y) &\defeq \Lambda^a(t,y)H^a_n\left((u^Q_n(t,y))_{Q\in\Q}, Q, y, \mu^a_n\right) + \Lambda^b(t,y)H^b_n\left((u^Q_n(t,y))_{Q\in\Q}, Q, y, \mu^b_n\right).
\end{align*}

Let $n \in \mathbb{N}^*$ and $Q \in \Q$, then $v:= u^Q_n - u^Q_{\infty}$ solves 
\begin{equation*}
    \begin{aligned}
        0 =& \left(\partial_t v + \frac{\sigma^2}{2}\partial^2_{yy} v
        -\sigma^2 \gamma Q \partial_y v + \left(\frac{\sigma^2 \gamma^2 Q^2}{2}- \left(\Lambda^a + \Lambda^b\right)(y)\right)v+
            f^Q_{n} - f^Q_{\infty}\right)(t,y)
    \end{aligned}
\end{equation*}
with boundary conditions
\begin{equation*}
    \left\{
        \begin{array}{ll}
            v(T, y) &=0 \\
            v\left(t, \bar{y}\right) &= v\left(t, y_+\right) \\
            v\left(t, -\bar{y}\right) &= v\left(t, y_-\right).
        \end{array}
    \right.
\end{equation*}
%By Proposition \ref{prop:feynmankac}, for all $(t,y) \in [0,T] \times \Y$, $e^{ct}|u^Q_n(t,y) - u^Q_{\infty}(t,y)| \leqslant \int_t^T e^{cs}\left|f^Q_n(s,y) - f^Q_{\infty}(s,y) \right|\diff s$, where $c \defeq \frac{\sigma^2 \gamma^2 \Qmax^2}{2} + 2 \Lambda^*$.
By Lemma \ref{lem:ineq_H_different_n}, there exists a constant $L'> 0$ depending on $T$, $\Qmax$, $\gamma$, $\sigma$, $\delta$, $\eta$, $\Lambda^*$, $C$ but not on $n$, nor the execution measures, such that for all $n \in \mathbb{N}^*$, $(t,Q,y) \in [0,T] \times \Q \times \bar{\Y}$,
\begin{equation*}
    \left|f_n^Q(t,y) - f^Q_{\infty}(t,y)\right| \leqslant
    L'\left(\sup_{R \in \Q} |u_n^R(t,y) - u_{\infty}^R(t,y)| + \varpi\left(\frac{\Qmax}{n}\right)\right).
\end{equation*}

For $t \in [0,T]$, define $g_n(t) \defeq e^{ct}\sup\limits_{R \in \Q, y \in \Y}\left|u^R_n(t,y) - u^R_{\infty}(t,y)\right|$, where $c \defeq \frac{\sigma^2 \gamma^2 \Qmax^2}{2} + 2 \Lambda^*$. The function $g_n$ is measurable -- the supremum can be taken over a countable subset, since all the functions involved are continuous. It follows by Corollary \ref{corol:bound_on_u} that
\begin{equation*}
     g_n(t) \leqslant L'T \varpi\left(\frac{\Qmax}{n}\right) + L'\int_t^T g_n(s) \diff s,\quad n \in \mathbb{N}^*,\, t \in [0,T].
\end{equation*}
Using Grönwall's inequality, we conclude that
\begin{equation*}
    \sup\limits_{R \in \mathcal{Q}_{n}}\left|u^R_{n}(t,y) - u^R_{\infty}(t,y)\right|_{\infty} \leqslant L'T e^{L'T}\varpi\left(\frac{\Qmax}{n}\right) \xrightarrow[n \to \infty]{} 0.
\end{equation*}
\end{dummyenv}

\section{Properties of the Hamiltonians}
\begin{dummyenv}
\renewcommand{\Q}{\mathcal{Q}_n}
\renewcommand{\Y}{\mathcal{Y}}
In this section we state some properties about the Hamiltonians (see Definition \ref{def:Hamiltonians}) that are used in the proof of existence of solutions of the Hamilton-Jacobi-Bellman equation \eqref{eq:HJB_interior}-\eqref{eq:HJB_border}. 

Let $n \in \mathbb{N}^* \cup \{\infty\}$. We denote by $\mathcal{O}$ the set of $(w,t,Q,y,\mu^a, \mu^b)$ such that
$(t,Q,y) \in [0,T) \times \Q \times Y$, $\mu^{a}$, and $\mu^b$ are probability measures on $\Qa$ and $\Qb$, respectively, and $w$ is a bounded function defined on $\Q$ for which the Hamiltonians are defined (i.e. $w$ is measurable or arbitrary $w$ if $\mu^a$ and $\mu^b$ have finite support). For $(w,t,Q,y,\mu^a, \mu^b) \in \mathcal{O}$, we define
\begin{equation*}
    H_n(w, t, Q, y, \mu^a, \mu^b)
    = \Lambda^a(t,y)H_n^a(w, Q, y, \mu^a) + \Lambda^b(t,y)H_n^b(w, Q, y, \mu^b).
\end{equation*}
We shall omit $n$ whenever is clear for the context.

Recall that for $\beta \in (0,\alpha]$, $\Lambda^*_{\beta} = \max\{|\Lambda^a|_{\beta}, |\Lambda^b|_{\beta}\}$. We fix such $\beta$ for this section.
For $\varepsilon > 0$, we say that $q \in \Qa \cap [0,Q + \Qmax]$ is $\varepsilon$-optimal for $H^a(w,t,Q,y,\mu^a)$ if
\begin{equation}
    \label{eq:epsilon_optimal}
    H^a(w,t,Q,y,\mu^a) \leqslant \int_{\Qa} e^{-\gamma(q \wedge z)\left(\frac{\delta}{2}-y\right)}w(Q - q \wedge z) \mu^a(\diff z) + \varepsilon.
\end{equation}
We say $q$ is optimal if equality holds with $\varepsilon=0$ in \eqref{eq:epsilon_optimal}.

\begin{rem}
    Recall that $\qa, \qb \leqslant 2\Qmax$, and therefore the upper bounds involving $\qa$ and $\qb$ depend directly on $\Qmax$.
\end{rem}

\subsection{Some inequalities involving the Hamiltonians}

\begin{lemma}
    \label{lem:H_ineq_ty}
    There exists a constant $C > 0$ depending only on $T$, $\gamma$, $\delta$, $\eta$, $\Qmax$, $\Lambda^*$, $\Lambda^*_{\beta}$ such that for all
    $(w, t, Q, y, \mu^a, \mu^b) \in \mathcal{O}$ and $(t',y') \in [0,T) \times \Y$,
    \begin{equation*}
        \left|H(w, t, Q, y, \mu^a, \mu^b) - H(w, t', Q, y', \mu^a, \mu^b)\right|
        \leqslant C|w|_{\infty} \frac{d((t,y), (t',y'))^{\beta}}{d_{(t,y)(t',y')}^{\beta}}.
    \end{equation*}
\end{lemma}

\begin{proof}
    We only show the inequality for the \enquote{a} part, since the \enquote{b} part is similar. Define $A \defeq \Lambda^a(t,y)H^a(w, Q, y, \mu^a) - \Lambda^a(t',y')H^a(w, Q, y', \mu^a)$. Let $\varepsilon > 0$ and $q \in \Qa \cap [0,Q + \Qmax]$ be $\varepsilon$-optimal for $H^a(w,Q,y, \mu^a)$. Then,
    \begin{equation*}
        A \leqslant \Lambda^a(t,y) \int_{\Qa} e^{-\gamma(q \wedge z)\left(\frac{\delta}{2}-y\right)}w(Q - q \wedge z) \mu^a(\diff z)
        -\Lambda^a(t',y') \int_{\Qa} e^{-\gamma(q \wedge z)\left(\frac{\delta}{2}-y'\right)}w(Q - q \wedge z) \mu^a(\diff z) + \Lambda^*\varepsilon.
    \end{equation*}
    Hence,
    \begin{equation*}
        \begin{split}
            A \leqslant (\Lambda^a(t,y) - &\Lambda^a(t',y'))\int_{\Qa} e^{-\gamma(q \wedge z)\left(\frac{\delta}{2}-y\right)}w(Q - q \wedge z) \mu^a(\diff z) \\
            &+ \Lambda^a(t',y') \int_{\Qa} \left(e^{-\gamma(q \wedge z)\left(\frac{\delta}{2}-y\right)} - e^{-\gamma(q \wedge z)\left(\frac{\delta}{2}-y'\right)}\right)w(Q - q \wedge z) \mu^a(\diff z) + \Lambda^*\varepsilon.
        \end{split}
    \end{equation*}
    Taking the absolute values on the right-hand side, we deduce
    \begin{equation*}
        A \leqslant e^{\gamma \qa \delta(\eta + 1)}|w|_{\infty}\Lambda^*_{\beta}\frac{d((t,y), (t',y'))^{\beta}}{d_{(t,y)(t',y')}^{\beta}}
        + \Lambda^* e^{\gamma \qa \delta(\eta + 1)}|w|_{\infty} \gamma \qa \delta (\eta +1)|y'-y| + \Lambda^*\varepsilon.
    \end{equation*}
    Since $|y'-y| \leqslant \delta^{1-\beta}(1+2\eta)^{1-\beta}|y'-y|^{\beta} \leqslant \delta(1+2\eta) T^{\beta} \frac{d((t,y), (t',y'))^{\beta}}{d_{(t,y)(t',y')}^{\beta}}$, we have the existence of a constant $C'$ depending only on $T$, $\gamma$, $\delta$, $\eta$, $\qa$, $\qb$, $\Lambda^*$, $\Lambda^*_{\beta}$ such that
    \begin{equation*}
        A
        \leqslant C'|w|_{\infty} \frac{d((t,y), (t',y'))^{\beta}}{d_{(t,y)(t',y')}^{\beta}} + \Lambda^*\varepsilon.
    \end{equation*}
    This holds for all $\varepsilon > 0$, hence by continuity, also for $\varepsilon = 0$. The same reasoning on $-A$ gives the result.
\end{proof}

\begin{lemma}
    \label{lem:H_ineq_w}
    There exists a constant $C > 0$ depending only on $\gamma$, $\delta$, $\eta$, $\Qmax$, $\Lambda^*$ such that for all
    $(w, t, Q, y, \mu^a, \mu^b) \in \mathcal{O}$ and $w'$ another bounded function such that the Hamiltonians are defined with $w'$,
    \begin{equation*}
        \left|H(w, t, Q, y, \mu^a, \mu^b) - H(w', t, Q, y, \mu^a, \mu^b)\right|
        \leqslant C|w - w'|_{\infty}.
    \end{equation*}
\end{lemma}

\begin{proof}
    Same reasoning as in lemma \ref{lem:H_ineq_ty}: take an $\varepsilon$-optimal $q$ and use $|w(Q-q\wedge z) - w'(Q-q\wedge z)| \leqslant |w-w'|_{\infty}$.
\end{proof}

\begin{lemma}
    \label{lem:H_ineq_Q}
    There exists a constant $C > 0$ depending only on $\gamma$, $\delta$, $\eta$, $\Qmax$, $\Lambda^*$ such that for all
    $(w, t, Q, y, \mu^a, \mu^b) \in \mathcal{O}$ and $Q' \in \Q$,
    \begin{equation*}
        \left|H(w, t, Q, y, \mu^a, \mu^b) - H(w, t, Q', y, \mu^a, \mu^b)\right|
        \leqslant C(|w|_{\infty}|Q-Q'| + m(w,|Q-Q'|))
    \end{equation*}
    where, for $D \in [0, \infty)$, $m(w,D) \defeq \sup \{w(Q_1) - w(Q_2) : (Q_1, Q_2) \in \Q \times \Q, |Q_1 - Q_2| \leqslant D\}$.
\end{lemma}

\begin{proof}
    We only show it for the \enquote{a} part, since the \enquote{b} part is similar.

    Define $A \defeq \Lambda^a(t,y)H^a(w, Q, y, \mu^a) - \Lambda^a(t,y)H^a(w, Q', y, \mu^a)$.
    Let $\varepsilon > 0$ and $q \in \Qa \cap [0,Q + \Qmax]$ be $\varepsilon$-optimal for $H^a(w,Q,y, \mu^a)$.
    
    If $q \leqslant Q' + \Qmax$, then
    \begin{align*}
        A &\leqslant \Lambda^a(t,y) \int_{\Qa}e^{-\gamma (q \wedge z)\left(\frac{\delta}{2} - y\right)}\left(w(Q - q \wedge z) - w(Q' - q \wedge z)\right)\mu^a(\diff z) + \Lambda^*\varepsilon\\
         &\leqslant \Lambda^* e^{\gamma \qa \delta \left(1+ \eta\right)}m(w, |Q-Q'|) + \Lambda^*\varepsilon.
    \end{align*}

    Suppose now that $q > Q' + \Qmax$. Then,
    \begin{equation*}
        \begin{split}
            A \leqslant &\Lambda^a(t,y)\int_{\Qa} \left(e^{-\gamma(q \wedge z)\left(\frac{\delta}{2}-y\right)} - e^{-\gamma\left(\left(Q' + \Qmax\right) \wedge z\right)\left(\frac{\delta}{2}-y\right)}\right)w(Q - q \wedge z) \mu^a(\diff z)\\
            &+ \Lambda^a(t,y)\int_{\Qa} e^{-\gamma\left(\left(Q' + \Qmax\right) \wedge z\right)\left(\frac{\delta}{2}-y\right)}\left(w(Q - q \wedge z) -w(Q' - \left(Q' +\Qmax\right) \wedge z)\right) \mu^a(\diff z) + \Lambda^*\varepsilon.
        \end{split}
    \end{equation*}
    Since for all $z$, $0 \leqslant \left(q \wedge z\right) - \left(\left(Q' + \Qmax\right) \wedge z\right) \leqslant Q - Q'$, we deduce
    \begin{equation*}
        A \leqslant \Lambda^* |w|_{\infty} e^{\gamma \qa \delta(\eta+1)}\gamma \qa \delta (\eta+1)|Q-Q'|
        +\Lambda^* e^{\gamma \qa \delta \left(1+ \eta\right)}m(w, |Q-Q'|) +\Lambda^*\varepsilon.
    \end{equation*}
    The same inequality holds for $-A$, for every $\varepsilon > 0$, therefore we get the desired result.
\end{proof}

\renewcommand{\Q}{\mathcal{Q}_{\infty}}
\begin{lemma}
    \label{lem:ineq_H_discrete_meas}
    Let $w$ and $w'$ be two bounded functions defined on $\Q$ admitting a modulus of continuity $\omega$ such that for all $\Delta \geqslant 0$, $\omega(\Delta) \geqslant \Delta$. Let $k \in \mathbb{N}^*$,  and $\mu_k^a$ and $\mu_k^b$ be the two probability measures defined by \eqref{eq:discrete_measures}. Fix $(t,Q,y)\in [0,T) \times \Q \times \Y$. Then, there exists a constant $C$ depending only on $\gamma$, $\delta$, $\eta$, $\Qmax$, $\Lambda^*$ but not on $w$, $w'$, $\omega$, $k$, $t$, $Q$, $y$ such that
    \begin{equation*}
        \left|H(w, t, Q, y, \mu_k^a, \mu_k^b) - H(w', t, Q, y, \mu_{k+1}^a, \mu_{k+1}^b)\right|
        \leqslant C\left(|w-w'|_{\infty} + (|w|_{\infty}+1)\omega\left(\frac{\max\{\qa, \qb\}}{2^{a_k}}\right)\right).
    \end{equation*}
\end{lemma}

\begin{proof}
    By Lemma \ref{lem:H_ineq_w}, there exists a constant $C'$ depending only on $\gamma$, $\delta$, $\eta$, $\Qmax$, $\Lambda^*$ such that
    \begin{equation*}
        \begin{split}
            \big|H(w, t, Q, y, \mu_k^a, \mu_k^b) &- H(w', t, Q, y, \mu_{k+1}^a, \mu_{k+1}^b)\big| \\
        &\leqslant C'|w-w'|_{\infty} + \left|H(w, t, Q, y, \mu_k^a, \mu_k^b) - H(w, t, Q, y, \mu_{k+1}^a. \mu_{k+1}^b)\right|
        \end{split}
    \end{equation*}
    Consequently, we only need to find the upper bound for the second term. As before, we only show it for the \enquote{a} part, since the \enquote{b} part is similar. Let $\varepsilon > 0$ and $q \in \Qa \cap [0,Q + \Qmax]$ be $\varepsilon$-optimal for $H^a(w,Q,y, \mu_{k+1}^a)$. Define $A \defeq \Lambda^a(t,y)H^a(w, Q, y, \mu_{k+1}^a) - \Lambda^a(t,y)H^a(w, Q, y, \mu_{k}^a)$. We have
    \begin{equation*}
        A \leqslant \Lambda^a(t,y)\left( \int_{\mathcal{Q}_{\infty}^{+,a}} e^{-\gamma(q \wedge z)\left(\frac{\delta}{2}-y\right)}w(Q - q \wedge z) \left(\mu_{k+1}^a(\diff z)-\mu_{k}^a(\diff z)\right)
        \right)+ \Lambda^*\varepsilon.
    \end{equation*}
    Using the expressions \eqref{eq:discrete_measures} of $\mu^a_k$ and $\mu^a_{k+1}$, it is immediate that
    \begin{equation*}
        A \leqslant \Lambda^a(t,y)\sum_{i=0}^{2^{a_k}-1} \sum_{j = 1}^{2^{a_{k+1}-a_k}-1}
        B_{i,j}\mu^a\left(\left[\qa\frac{i}{2^{a_k}} + \qa\frac{j}{2^{{a_{k+1}}}}, \qa\frac{i}{2^{a_k}} + \qa\frac{j+1}{2^{{a_{k+1}}}}\right)\right) + \Lambda^*\varepsilon
    \end{equation*}
    where, for all $i$ and $j$, defining $q_{i,j} \defeq \qa\frac{i}{2^{a_k}} + \qa\frac{j}{2^{{a_{k+1}}}}$,
    \begin{align*}
        B_{i,j} &\defeq e^{-\gamma\left(q \wedge q_{i,j}\right)\left(\frac{\delta}{2}-y\right)}w\left(Q - q \wedge q_{i,j}\right)
        - e^{-\gamma\left(q \wedge q_{i,0}\right)\left(\frac{\delta}{2}-y\right)}w\left(Q - q \wedge q_{i,0}\right) \\
        &\leqslant e^{\gamma \qa \delta(\eta +1)}\omega\left(q_{i,j}-q_{i,0}\right) + |w|_{\infty}\gamma\delta(\eta+1)e^{\gamma \qa \delta(\eta+1)}\left(q_{i,j}-q_{i,0}\right)\\
        &\leqslant e^{\gamma \qa\delta(\eta+1)}(1+|w|_{\infty}\gamma\delta(\eta+1))\omega\left(\frac{\qa}{2^{a_k}}\right)
        .
    \end{align*}
    Since $\sum_{i=0}^{2^{a_k}-1} \sum_{j = 1}^{2^{a_{k+1}-a_k}-1}
        \mu^a\left(\left[q_{i,j},q_{i,j+1}\right)\right) \leqslant 1$, $A \leqslant \Lambda^* e^{\gamma \qa\delta(\eta+1)} (1+|w|_{\infty}\gamma\delta(\eta+1))\omega\left(\frac{\qa}{2^{a_k}}\right) + \Lambda^*\varepsilon$. Similarly, $-A \leqslant \Lambda^* e^{\gamma \qa\delta(\eta+1)} (1+|w|_{\infty}\gamma\delta(\eta+1))\omega\left(\frac{\qa}{2^{a_k}}\right) + \Lambda^*\varepsilon$. This being valid for all $\varepsilon > 0$, we conclude the desired result.
\end{proof}

\renewcommand{\Q}{\mathcal{Q}_{n}}
\begin{lemma}
    \label{lem:ineq_H_different_n}
    For $n \in \mathbb{N}^*$ define $\mu^a_n$ and $\mu^b_n$ by \eqref{eq:discrete_measures_2}. Let $w':\Q \to \mathbb{R}$ and $w:\mathcal{Q}_{\infty} \to \mathbb{R}$ be two bounded functions. Suppose $w$ admits a modulus of continuity $\omega$ such that for all $\Delta \geqslant 0$, $\omega(\Delta) \geqslant \Delta$. Then, there exists a constant $C$ depending only on $\gamma$, $\delta$, $\eta$, $\Qmax$, $\Lambda^*$ but not on $n$, $w$, $w'$, $\omega$, $k$, $t$, $Q$, $y$, nor the execution measures, such that
    \begin{equation*}
        \left|H_n(w', t, Q, y, \mu_n^a, \mu_n^b) - H_{\infty}(w, t, Q, y, \mu^a, \mu^b)\right|
        \leqslant C \left(\sup_{R \in \Q}|w'(R) - w(R)| + (1+|w|_{\infty})\omega\left(\frac{\Qmax}{n}\right)\right).
    \end{equation*}
\end{lemma}

\begin{proof}
    $\left|H_n(w', t, Q, y, \mu_n^a, \mu_n^b) - H_{\infty}(w, t, Q, y, \mu^a, \mu^b)\right| \leqslant A + B$ where
    \begin{align*}
        A &\defeq \left|H_n(w', t, Q, y, \mu_n^a, \mu_n^b) - H_{\infty}(w, t, Q, y, \mu_n^a, \mu_n^b)\right|\\
        B &\defeq \left|H_{\infty}(w, t, Q, y, \mu_n^a, \mu_n^b) - H_{\infty}(w, t, Q, y, \mu^a, \mu^b)\right|.
    \end{align*}
    As usual, we only show the inequality for the \enquote{a} part. We start with $A$. 
    
    Define $A' \defeq H_n^a(w', Q, y, \mu_n^a) - H_{\infty}^a(w, Q, y, \mu_{n}^a)$. Let $q \in \Q \cap [0, Q + \Qmax] \subset \mathcal{Q}_{\infty} \cap [0, Q + \Qmax]$ be optimal for $H_n^a(w', Q, y, \mu_n^a)$. Then,
    \begin{align}
        A' &\leqslant \int_{\Qa} e^{-\gamma(q \wedge z)\left(\frac{\delta}{2}-y\right)}w'(Q - q \wedge z) \mu_{n}^a(\diff z) - \int_{\mathcal{Q}^{+,a}_{\infty}} e^{-\gamma(q \wedge z)\left(\frac{\delta}{2}-y\right)}w(Q - q \wedge z) \mu_{n}^a(\diff z)\notag\\
         &\leqslant e^{\gamma \qa \delta (\eta+1)}\sup_{R \in \Q}|w'(R) - w(R)|.\label{eq:appD-1}
    \end{align}
    Let $q \in \mathcal{Q}_{\infty} \cap [0, Q + \Qmax]$ be optimal for $H_{\infty}^a(w, Q, y, \mu_{n}^a)$ (it exists since, by dominated convergence, $q \mapsto \int_{\Qa} e^{-\gamma(q \wedge z)\left(\frac{\delta}{2}-y\right)}w(Q - q \wedge z) \mu_n^a(\diff z)$ is continuous, and defined on a compact set). Let $q' \defeq \sup\{z \in \Qa: z \leqslant q\} \in \Qa$, then $|q'-q| \leqslant \frac{\Qmax}{n}$. Hence,
    \begin{align}
        -A' &\leqslant \int_{\mathcal{Q}^{+,a}_{\infty}} e^{-\gamma(q \wedge z)\left(\frac{\delta}{2}-y\right)}w(Q - q \wedge z) \mu_{n}^a(\diff z) - \int_{\Qa} e^{-\gamma(q' \wedge z)\left(\frac{\delta}{2}-y\right)}w'(Q - q' \wedge z) \mu_{n}^a(\diff z)\notag\\
        &\leqslant e^{\gamma \qa \delta (\eta+1)}\sup_{R \in \Q}|w'(R) - w(R)| + \left|e^{-\gamma q \left(\frac{\delta}{2}-y\right)}w(Q-q) - e^{-\gamma q' \left(\frac{\delta}{2}-y\right)}w'(Q-q')\right|\notag\\
        & \leqslant e^{\gamma \qa \delta (\eta+1)}\left(2\sup_{R \in \Q}|w'(R) - w(R)| + \omega(q-q') + \gamma \qa \delta(\eta+1)|q-q'||w|_{\infty}\right).\label{eq:appD0}
    \end{align}
    Thus, combining \eqref{eq:appD-1} and \eqref{eq:appD0}, there exists a constant $C'$ depending only on $\gamma$, $\delta$, $\eta$, $\Qmax$ such that
    \begin{equation*}
        A=|A'| \leqslant C' \left(\sup_{R \in \Q}|w'(R) - w(R)| + (1+|w|_{\infty})\omega\left(\frac{\Qmax}{n}\right)\right).
    \end{equation*}

    Define now $B' \defeq H_{\infty}^a(w, Q, y, \mu_n^a) - H_{\infty}^a(w, Q, y, \mu^a)$.  Let $q \in \mathcal{Q}_{\infty} \cap [0, Q + \Qmax]$ be optimal for $H_{\infty}^a(w, Q, y, \mu_{n}^a)$. Then,
    \begin{align*}
        B' & \leqslant  \int_{\mathcal{Q}^{+,a}_{\infty}} e^{-\gamma(q \wedge z)\left(\frac{\delta}{2}-y\right)}w(Q - q \wedge z) \mu_{n}^a(\diff z) - \int_{\mathcal{Q}^{+,a}_{\infty}} e^{-\gamma(q \wedge z)\left(\frac{\delta}{2}-y\right)}w(Q - q \wedge z) \mu^a(\diff z) \\
         & \leqslant  \sum_{j = 0}^{N^a_n - 1} \int_{\left[\Qmax\frac{j}{n}, \Qmax\frac{j+1}{n}\right)} B_j(z) \mu^a(\diff z) + \int_{\left[\qa_n, \qa\right)}B_{N^a_n}(z) \mu^a(\diff z),
    \end{align*}
    where
    \begin{align*}
        &B_j(z) \defeq e^{-\gamma\left(q \wedge \frac{\Qmax j}{n}\right)\left(\frac{\delta}{2}-y\right)}w\left(Q - q \wedge \frac{\Qmax j}{n}\right) -  e^{-\gamma(q \wedge z)\left(\frac{\delta}{2}-y\right)}w(Q - q \wedge z),\, j \leqslant N^a_n-1,\, z \in \left[\Qmax\frac{j}{n}, \Qmax\frac{j+1}{n}\right), \\
        &B_{N^a_n}(z) \defeq e^{-\gamma\left(q \wedge \qa_n\right)\left(\frac{\delta}{2}-y\right)}w\left(Q - q \wedge \qa_n\right) -  e^{-\gamma(q \wedge z)\left(\frac{\delta}{2}-y\right)}w(Q - q \wedge z),\, z \in \left[\qa_n, \qa\right).
    \end{align*}
    We have the upper bound
    \begin{equation*}
        |B_j|_{\infty} \leqslant e^{\gamma \qa \delta (\eta +1)}\left(\gamma \qa \delta (\eta + 1)|w|_{\infty}\frac{\Qmax}{n} + \omega\left(\frac{\Qmax}{n}\right)\right),\quad j \leqslant N_n^a.
    \end{equation*}
    Thus, $B' \leqslant e^{\gamma \qa \delta (\eta +1)}\left(\gamma \qa \delta (\eta + 1)\frac{\Qmax}{n} + \omega\left(\frac{\Qmax}{n}\right)\right)$. By a similar argument, we have the same bound for $-B'$. Consequently, there exists a constant $C''$, depending only on $\gamma$, $\delta$, $\eta$, $\Qmax$, such that
    \begin{equation*}
        B=|B'| \leqslant C'' (1+|w|_{\infty})\omega\left(\frac{\Qmax}{n}\right).
    \end{equation*}
\end{proof}

\subsection{Convergence}

\begin{lemma}
    \label{lem:H_convergence_measures}
    Let $(w,t,Q,y,\mu^a, \mu^b) \in \mathcal{O}$. Suppose that $w$ is continuous in the case $n = \infty$. Let $(\mu^a_k)_{k \in \mathbb{N}^*}$ and $(\mu^b_k)_{k \in \mathbb{N}^*}$ be two sequences of measures on $\Qa$ and $\Qb$ converging in distribution to $\mu^a$ and $\mu^b$, respectively. Then,
    \begin{equation*}
        \lim_{k \to \infty} H(w,t,Q,y,\mu^a_k,\mu^b_k) = H(w,t,Q,y,\mu^a,\mu^b).
    \end{equation*}
\end{lemma}

\begin{proof}
    We show that $\lim\limits_{k \to \infty}H^a(w,Q,y,\mu^a_k) = H^a(w,Q,y,\mu^a)$, the \enquote{b} part being analogous.
    
    Let $q \in \Qa\cap[0,\Qmax + Q]$ be optimal for $H^a(w,Q,y,\mu^a)$ (it exists since, by dominated convergence, $q \mapsto \int_{\Qa} e^{-\gamma(q \wedge z)\left(\frac{\delta}{2}-y\right)}w(Q - q \wedge z) \mu^a(\diff z)$ is continuous, and defined on a compact set). For all $k \in \mathbb{N}^*$, by definition,
    \begin{equation*}
        \int_{\Qa} e^{-\gamma(q \wedge z)\left(\frac{\delta}{2}-y\right)}w(Q - q \wedge z) \mu_k^a(\diff z) \leqslant H^a(w,Q,y,\mu_k^a).
    \end{equation*}
    Since $\mu_k^a \to \mu^a$ in distribution and the integrand is continuous (with respect to $z$), the left-hand side converges to $H^a(w,Q,y,\mu^a)$. Thus,
    \begin{equation}\label{eq:appD1}
        H^a(w,Q,y,\mu^a) \leqslant \liminf_{k \to \infty} H^a(w,Q,y,\mu_k^a).
    \end{equation}

    For $k \in \mathbb{N}^*$, let $q_k \in \Qa \cap [0,\Qmax + Q]$ be optimal for $H^a(w,Q,y,\mu_k^a)$. Let $(k'_i)_{i \in \mathbb{N}^*}$ be a strictly increasing sequence of positive integers such that
    \begin{equation}\label{eq:appD2}
        \lim_{i \to \infty} H^a(w,Q,y,\mu^a_{k'_i}) = \limsup_{k \to \infty} H^a(w,Q,y,\mu_k^a).
    \end{equation}
    Let $(k_i)_{i \in \mathbb{N}^*}$ be a a subsequence of $(k'_i)_{i \in \mathbb{N}^*}$ and $q^*\in \Qa \cap [0,\Qmax + Q]$ such that $\lim\limits_{i \to \infty} q_{k_i} = q^*$.

    Let $\varepsilon > 0$. Let $\delta > 0$ such that for $q,q' \in \Qa \cap [0,Q + \Qmax]$,
    \begin{equation*}
        |q'-q| \leqslant \delta \implies \left| e^{-\gamma q \left(\frac{\delta}{2} - y\right)}w(Q-q)  - e^{-\gamma q' \left(\frac{\delta}{2} - y\right)}w(Q-q')\right| \leqslant \varepsilon.
    \end{equation*}
    Let $i_0 \in \mathbb{N}^*$ such that $i \geqslant i_0$ implies $|q_{k_i} - q^*| \leqslant \delta$. Then, for $i \geqslant i_0$,
    \begin{equation}\label{eq:appD3}
        H^a(w,Q,y,\mu_{k_i}^a) \leqslant \varepsilon + \int_{\Qa} e^{-\gamma(q^* \wedge z)\left(\frac{\delta}{2}-y\right)}w(Q - q^* \wedge z) \mu_{k_i}^a(\diff z).
    \end{equation}
    The right-hand side converges to $\varepsilon + \int_{\Qa} e^{-\gamma(q^* \wedge z)\left(\frac{\delta}{2}-y\right)}w(Q - q^* \wedge z) \mu^a(\diff z)$ which is smaller (or equal) than $\varepsilon + H^a(w,Q,y,\mu^a)$. Hence, thanks to \eqref{eq:appD2} and \eqref{eq:appD3},
    \begin{equation*}
         \limsup_{k \to \infty} H^a(w,Q,y,\mu_k^a) = \lim_{i \to \infty} H^a(w,Q,y,\mu^a_{k'_i}) =
        \lim_{i \to \infty} H^a(w,Q,y,\mu_{k_i}^a) \leqslant \varepsilon + H^a(w,Q,y,\mu^a).
    \end{equation*}
    This being valid for all $\varepsilon > 0$, it also holds for $\varepsilon=0$. Combined with \eqref{eq:appD1}, this yields the desired result.
\end{proof}
\end{dummyenv}

\section{Some measurability proofs for Theorem \ref{thm:existence}}
\begin{dummyenv}
\renewcommand{\Q}{\mathcal{Q}_n}
\renewcommand{\Y}{\mathcal{Y}}
\begin{lemma}
    \label{lem:meas_gk}
    Let $k \in \mathbb{N}$. Then $g_k$ defined by \eqref{eq:g_k} is measurable.
\end{lemma}

\begin{proof}
    Let $t \in [0,T]$. Then -- since for fixed $Q$, $u^Q_k$ and $u^Q_{k+1}$ are continuous --
    \begin{equation*}
        g_k(t) = e^{ct}\sup_{Q \in \Q}\sup_{y \in \Y \cap \mathbb{Q}}|u_{k+1}^Q(t,y) - u^Q_k(t,y)|
        = e^{ct}\sup_{y \in \Y \cap \mathbb{Q}} g_{k, y}(t)
    \end{equation*}
where for $y \in \Y$, $g_{k,y}(t) \defeq \sup_{Q \in \Q}|u_{k+1}^Q(t,y) - u^Q_k(t,y)|$. It is sufficient to show that the $g_{k,y}$'s are measurable. We actually show that they are continuous on $[0,T)$. Let $y \in \Y$ and $K$ be a compact included in $[0,T)$. Since $\sup_{Q \in \Q}|u_i^Q|_{\beta} < \infty$, for $i=k,k+1$, there exists a constant $C \in (0, \infty)$ such that
\begin{equation*}
    \sup_{Q \in \Q} |u_k^Q(t,y) - u_k^Q(t',y)| + \sup_{Q \in \Q} |u_{k+1}^Q(t,y) - u_{k+1}^Q(t',y)| \leqslant C |t-t'|^{\frac{\beta}{2}},\quad (t,t') \in K^2.
\end{equation*}
Let $\varepsilon > 0$, $(t,t') \in K^2$, and $Q \in \Q$ be such that $g_{k,y}(t) \leqslant |u_{k+1}^Q(t,y) - u^Q_k(t,y)| + \varepsilon$. Then
\begin{align*}
    g_{k,y}(t) - g_{k,y}(t') &\leqslant |u_{k+1}^Q(t,y) - u^Q_k(t,y)| - |u_{k+1}^Q(t',y) - u^Q_k(t',y)| + \varepsilon\\
     &\leqslant |u_{k+1}^Q(t,y) - u^Q_{k+1}(t',y)| + |u_{k}^Q(t,y) - u^Q_k(t',y)| + \varepsilon \\
     &\leqslant C|t-t'|^{\frac{\beta}{2}} + \varepsilon.
\end{align*}
This being valid for all $\varepsilon > 0$, and for $g_{k,y}(t') - g_{k,y}(t)$, we deduce that $g_{k,y}$ is (Hölder) continuous on $K$. Since $K$ was arbitrary, this yields the result.
\end{proof}

\begin{lemma}
    \label{lem:meas_md}
    Let $\Delta \in [0, \infty)$. Then $m^{\Delta}$ defined by \eqref{eq:modulus_cty_Q} is measurable.
\end{lemma}

\begin{proof}
    The proof is very similar to the one of Lemma \ref{lem:meas_gk}. Since for fixed $(Q, Q')$, $u^Q$ and $u^{Q'}$ are continuous, then
    \begin{equation*}
        m^{\Delta}(t) = \sup_{Q, Q' \in \mathcal{Q}_{\infty}, |Q' - Q| \leqslant \Delta}\sup_{y \in \Y \cap \mathbb{Q}}(u^Q(t,y) - u^{Q'}(t,y))
        = \sup_{y \in \Y \cap \mathbb{Q}} m^{\Delta, y}(t)
    \end{equation*}
    where for $y \in \Y$, $m^{\Delta, y}(t) \defeq \sup_{Q, Q' \in \mathcal{Q}_{\infty}, |Q' - Q| \leqslant \Delta}|u^Q(t,y) - u^Q_k(t,y)|$. It is sufficient to show that the $m^{\Delta, y}$'s are measurable. Let $K$ be a compact included in $[0,T)$. Since $\sup_{Q \in \mathcal{Q}_{\infty}}|u^Q|_{\beta} < \infty$, there exists a constant $C \in (0, \infty)$ such that
    \begin{equation*}
        \sup_{Q \in \mathcal{Q}_{\infty}} |u^Q(t,y) - u^Q(t',y)|\leqslant C |t-t'|^{\frac{\beta}{2}},\quad (t,t') \in K^2.
    \end{equation*}
Let $\varepsilon > 0$, $(t,t') \in K^2$, and $(Q, Q') \in \mathcal{Q}_{\infty}$ be such that $|Q'-Q| \leqslant \Delta$ and $m^{\Delta,y}(t) \leqslant u^Q(t,y) - u^{Q'}(t,y) + \varepsilon$. Then,
\begin{equation*}
    m^{\Delta,y}(t) - m^{\Delta,y}(t') \leqslant u^Q(t,y) - u^{Q'}(t,y) + u^Q(t',y) - u^{Q'}(t',y) + \varepsilon \leqslant 2C|t-t'|^{\frac{\beta}{2}} + \varepsilon.
\end{equation*}
This being valid for all $\varepsilon > 0$, and for $m^{\Delta,y}(t') - m^{\Delta,y}(t)$, we deduce that $m^{\Delta,y}$ is (Hölder) continuous on $K$. Since $K$ was arbitrary, this yields the result.
\end{proof}
\end{dummyenv}

\section{About concavity}
\begin{dummyenv}
\renewcommand{\Q}{\mathcal{Q}_{\infty}}
\renewcommand{\Qa}{\mathcal{Q}^{+,a}_{\infty}}
\renewcommand{\Qb}{\mathcal{Q}^{+,b}_{\infty}}
\renewcommand{\Qi}{\mathcal{Q}^{+,i}_{\infty}}
In this section, we derive some properties of continuous concave functions, their integrals and the continuous concave enveloppes. We use the notations from Section \ref{sec:uniqueness}.

\subsection{Some lemmas about concave functions}

\begin{lemma}
    \label{lem:basic_max_concave}
    Let $I$ be a compact interval and $\phi$ be a real-valued continuous concave function defined on $I$.\\
    (i) Let $q^* \in I$ be such that $\phi(q^*) = \max \phi$. Then, for all  $z \in \mathbb{R}_+$, $\phi\left(q^* \wedge z\right) = \max_{q \in I} \phi(q \wedge z)$.\\
    (ii) Let $q \in I$ be such that $\phi(q) < \max \phi$. Then, for all $ z \geqslant \max I$, $\phi\left(q \wedge z\right) < \max \phi(\cdot \wedge z)$.
\end{lemma}

\begin{proof}
    Point (ii) is immediate since for $q\in I$ and $z \geqslant \max I$, $q\wedge z = q$. We now show (i). Let $z \in \mathbb{R}_+$. If $z \geqslant q^*$, $\phi(q^* \wedge z) = \phi(q^*) = \max \phi \geqslant \max \phi(\cdot \wedge z)$. Suppose that $z < q^*$. Since $\phi$ is increasing on $[0, q^*]$, which contains $[0,z]$, we have $\max \phi(\cdot \wedge z) = \phi(z) = \phi(q^* \wedge z)$.
\end{proof}

We have the following corollary regarding the maximization of an integral by maximizing the integrand. 
\begin{corollary}
    \label{corol:max_integral}
    Let $Q \in \Q$ and $i \in \{a,b\}$. Define $\varepsilon \defeq -1$ if $i = b$ and $\varepsilon \defeq 1$ if $i = a$.
    Let $\phi: \Qi \cap [0, \varepsilon Q + \Qmax] \to \mathbb{R}$ be a continuous concave function. Let $F:\mathbb{R} \to \mathbb{R}$ be a strictly increasing function.\\
    (i) Let $q^* \in \argmax \phi$. Then $q^* \in \argmax\limits_{q \in \Qi \cap [0, \varepsilon Q +\Qmax]} \int_{\Qi} F(\phi(q \wedge z)) \mu^i(\diff z)$.\\
    (ii) Let $\hat{q} \in [0, \varepsilon Q +\Qmax] \setminus \argmax \phi$. Then $\hat{q} \notin \argmax\limits_{q \in \Qi \cap [0, \varepsilon Q +\Qmax]} \int_{\Qi} F(\phi(q \wedge z)) \mu^i(\diff z)$.
\end{corollary}

\begin{proof}
    Point (i) is an immediate consequence of Lemma \ref{lem:basic_max_concave} (i), since the maximum holds pointwise.
    
    We now prove (ii). By Lemma \ref{lem:basic_max_concave}, since $\phi$ is continuous, there exists $\eta > 0$ such that for all $z \in (\qi - \eta, \qi]$, $\phi\left(\hat{q} \wedge z\right) < \max \phi(\cdot \wedge z) = \phi(q^* \wedge z)$. Thus, since $\qi$ is in the support of $\mu^i$,
    \begin{equation}\label{eq:appF1}
        \int_{\Qi \cap (\qi - \eta, \qi]} F(\phi(\hat{q} \wedge z)) \mu^i(\diff z) < \int_{\Qi \cap (\qi - \eta, \qi]} F(\phi(q^* \wedge z)) \mu^i(\diff z).
    \end{equation}
    In addition, by Lemma \ref{lem:basic_max_concave} (i), and since the inequality holds pointwise,  we have
    \begin{equation}\label{eq:appF2}
        \int_{\Qi \cap [0,\qi - \eta]} F(\phi(\hat{q} \wedge z)) \mu^i(\diff z) \leqslant \int_{\Qi \cap [0,\qi - \eta]} F(\phi(q^* \wedge z)) \mu^i(\diff z).
    \end{equation}
    The inequalities \eqref{eq:appF1} and \eqref{eq:appF2} yield the conclusion.
\end{proof}

Recall the definition of the log-Hamiltonians $h^a,h^b$ in \eqref{eq:deflogHamiltonian}.
\begin{lemma}
    \label{lem:same_maximizer}
    Let $k \in \mathbb{R}$, $Q \in \Q$ and $g:\Q \to \mathbb{R}$ be a continuous concave function.\\
    (i) Define $Q^* \defeq \inf \left\{R \in \Q \setminus \{\Qmax\}: g'_+(R) \leqslant k\right\}\wedge \Qmax$.\\
    If $Q^* > Q$, then $h^a(g,k,Q) = -1$.\\
    If $Q^* \leqslant Q$, then $h^a(g,k,Q) = \int_{\Qa} - e^{-k\left((Q-Q^*) \wedge z\right) - g\left(Q - (Q-Q^*)\wedge z\right) + g(Q)}\mu^a(\diff z)$.\\
    (ii) Define $Q^* \defeq \inf \left\{R \in \Q \setminus \{-\Qmax\}: g'_-(R) \leqslant -k\right\}\vee (-\Qmax)$.\\
    If $Q^* < Q$, then $h^b(g,k,Q) = -1$.\\
    If $Q^* \geqslant Q$, then $h^b(g,k,Q) = \int_{\Qa} - e^{-k\left((Q^*-Q) \wedge z\right) - g\left(Q - (Q^*-Q)\wedge z\right) + g(Q)}\mu^b(\diff z)$.
\end{lemma}

\begin{proof}
    We only prove (i) since (ii) is similar.
    Let $\phi: q \in [Q - \Qmax, Q + \Qmax] \mapsto kq + g(Q- q)$. The function $\phi$ is continuous and concave, and for all $q$, $\phi'_-(q) = k - g'_+(Q - q)$. We have
    \begin{equation*}
        \sup\left\{q : \phi'_{-}(q) \geqslant 0\right\} \wedge (Q + \Qmax) = Q - Q^*
    \end{equation*}
    therefore $Q - Q^*$ maximizes $\phi$.

    If $Q^* > Q$, then $\phi$ is nonincreasing on $[0, Q + \Qmax]$ and $\phi(0) = g(Q) = \max\limits_{\Qa \cap [0,Q + \Qmax]}\phi$. Corollary \ref{corol:max_integral} allows us to conclude.

    If $Q^* \leqslant Q$ and $Q - Q^* \leqslant \qa$, then $\phi(Q- Q^*) = \max\limits_{\Qa \cap [0,Q + \Qmax]}\phi$. The conclusion follows from Corollary \ref{corol:max_integral}

    If $Q^* \leqslant Q$ and $Q - Q^* > \qa$, then $\phi(\qa) = \max\limits_{\Qa \cap [0,Q + \Qmax]}\phi$ and for all $ z \in \Qa$, $(Q-Q^*)\wedge z = z = \qa \wedge z$. By Corollary \ref{corol:max_integral},
    \begin{equation*}
        h^a(g,k,Q) =\int_{\Qa} - e^{-k\left(\qa \wedge z\right) - g\left(Q - \qa\wedge z\right) + g(Q)}\mu^a(\diff z) = \int_{\Qa} - e^{-k\left((Q-Q^*) \wedge z\right) - g\left(Q - (Q-Q^*)\wedge z\right) + g(Q)}\mu^a(\diff z).
    \end{equation*}
\end{proof}

The following lemma shows that over the intervals where $g$ is affine, the log-Hamiltonians $h^a$ and $h^b$ are concave. In addition, it provides sufficient conditions to have a strict concavity inequality.
\begin{lemma}
    \label{lem:bid_ask_max}
    Let $k \in \mathbb{R}$ and $g: \Q \to \mathbb{R}$ be a continuous concave function.\\
    Suppose that for $(Q,Q') \in \Q \times \Q$ and $\lambda \in [0,1]$, $(1-\lambda)g(Q) + \lambda g(Q') = g((1-\lambda)Q + \lambda Q')$. Then,\\
    (ia) $(1-\lambda) h^a(g,k,Q) + \lambda h^a(g,k,Q') \leqslant h^a(g,k,(1-\lambda)Q + \lambda Q')$.\\
    (ib) $(1-\lambda) h^b(g,k,Q) + \lambda h^b(g,k,Q') \leqslant h^b(g,k,(1-\lambda)Q + \lambda Q')$.\\
    (iia) If $\lambda \notin \{0,1\}$, $Q < Q'$, $p \defeq \frac{g(Q')-g(Q)}{Q'-Q} < k$, and
    $g'_+(R) > p$ for all $ R \in \Q \cap (-\infty, Q)$, then\\
    $(1-\lambda) h^a(g,k,Q) + \lambda h^a(g,k,Q') < h^a(g,k,(1-\lambda)Q + \lambda Q')$.\\
    (iib) If $\lambda \notin \{0,1\}$, $Q < Q'$, $p\defeq\frac{g(Q')-g(Q)}{Q'-Q} > -k$, and
    $g'_-(R) < p$ for all $R \in \Q \cap (Q, \infty)$, then\\
    $(1-\lambda) h^b(g,k,Q) + \lambda h^b(g,k,Q') < h^b(g,k,(1-\lambda)Q + \lambda Q')$.
\end{lemma}

\begin{rem}
    The hypothesis in $(iia)$ states that, while $g$ is affine on $[Q,Q']$, it is not affine on $[R,Q']$ for any $R < Q$.
\end{rem}

\begin{proof}
    We only prove parts (ia) and (iia), since parts (ib) and (iib) are similar. Without loss of generality, suppose $Q < Q'$. Let $Q^* \defeq \inf \left\{R \in \Q \setminus \{\Qmax\}: g'_+(R) \leqslant k\right\}\wedge \Qmax$. Since $g'_+$ is constant on $[Q,Q')$ and equal to $p$, then $Q^* \in \Q \setminus (Q,Q')$. For the rest of the proof we write $\tilde{Q} \defeq (1-\lambda)Q + \lambda Q'.$

    We first prove point (ia). If $Q^* \geqslant Q'$, by Lemma \ref{lem:same_maximizer}, $h^a(g,k,Q) = h^a(g,k,Q') = h^a(g,k,\tilde{Q})=-1$ and the result follows immediately.\\
    Suppose that $Q^* \leqslant Q$. Let $z \in \Qa$. Since the function $x \mapsto -e^{-x}$ is strictly increasing and strictly concave, and since $g$ is concave and affine between $Q$ and $Q'$, 
    \begin{equation}
        \label{eq:pointwise_concavity}
        \begin{split}
            -(1-\lambda)&e^{-k((Q-Q^*)\wedge z) - g(Q - (Q-Q^*)\wedge z) + g(Q)}
        - \lambda e^{-k((Q'-Q^*)\wedge z) - g(Q' - (Q'-Q^*)\wedge z) + g(Q')}\\
        &\leqslant -e^{-k\left[(1-\lambda)((Q - Q^*) \wedge z) + \lambda((Q' - Q^*) \wedge z)\right] - (1- \lambda)g(Q - (Q-Q^*)\wedge z) - \lambda g(Q' - (Q'-Q^*)\wedge z) + (1-\lambda)g(Q) + \lambda g(Q')}\\
        &\leqslant
        -e^{-k\left[(1-\lambda)((Q - Q^*) \wedge z) + \lambda((Q' - Q^*) \wedge z)\right] - g\left(\tilde{Q} - (1-\lambda)((Q-Q^*)\wedge z) - \lambda((Q'-Q^*)\wedge z)\right) + g(\tilde{Q})}.
        \end{split}
    \end{equation}
    We define the function $\phi:q \in [0,\tilde{Q} + \Qmax] \mapsto kq + g(\tilde{Q} - q)$. Reasoning as in the proof of Lemma \ref{lem:same_maximizer}, we have that $\phi$ is concave and maximized at $\tilde{Q} - Q^*$ and is therefore nondecreasing on $[0, \tilde{Q} - Q^*]$. Since, by the concavity of $(\cdot \wedge z)$, we have
    \begin{equation*}
        (1-\lambda)((Q-Q^*)\wedge z) + \lambda ((Q' - Q^*) \wedge z) \leqslant \left(\tilde{Q} - Q^*\right) \wedge z \leqslant \tilde{Q}-Q^*,
    \end{equation*}
    the inequality \eqref{eq:pointwise_concavity} becomes
    \begin{equation}
        \label{eq:pointwise_concavity_end}
        \begin{split}
            -(1-\lambda)e^{-k((Q-Q^*)\wedge z) - g(Q - (Q-Q^*)\wedge z) + g(Q)}
        &- \lambda e^{-k((Q'-Q^*)\wedge z) - g(Q' - (Q'-Q^*)\wedge z) + g(Q')}\\
        &\leqslant
        -e^{-k((\tilde{Q} - Q^*) \wedge z) - g\left(\tilde{Q} - (\tilde{Q}-Q^*)\wedge z\right) + g(\tilde{Q})}.
        \end{split}
    \end{equation}
    Thanks to Lemma \ref{lem:same_maximizer}, the result follows by integrating \eqref{eq:pointwise_concavity_end} with respect $\mu^a(\diff z)$.

    We now prove point (iia). Since $g_+'$ is nonincreasing and $g_+'(Q) = p < k$, we have that $Q^* \leqslant Q$. Suppose that
    \begin{equation}
        \label{eq:different_gains}
        \begin{split}
            k\left((Q-Q^*)\wedge \qa\right) +& g\left(Q - (Q-Q^*) \wedge \qa\right) - g(Q)\\ &<
        k\left((Q'-Q^*)\wedge \qa\right) + g\left(Q' - (Q'-Q^*) \wedge \qa\right) - g(Q').
        \end{split}
    \end{equation}
    Then, since $x \mapsto -e^{-x}$ is strictly concave, the first inequality in \eqref{eq:pointwise_concavity} is strict with $z = \qa$, and consequently so is the one in \eqref{eq:pointwise_concavity_end} (with $z = \qa$). By the continuity of all the functions involved, strict inequality in \eqref{eq:pointwise_concavity_end} holds for $z \in (\qa - \eta ,\qa]$ for some $\eta > 0$. We have the desired result because $\mu^a((\qa - \eta, \qa]) > 0$.
    
    It only remains to show that \eqref{eq:different_gains} is true. Considering the function $q \in [0,Q + \Qmax] \mapsto g(Q-q)$ which is concave and has left derivative $-g_+'(Q - \cdot)$, \parencite[][Proposition 1.6.1]{niculescu_convex_2005} yields
    \begin{equation}\label{eq:Niculescu1}
        k\left((Q-Q^*)\wedge \qa\right) + g\left(Q - (Q-Q^*) \wedge \qa\right) - g(Q) = \int_0^{(Q-Q^*)\wedge \qa}\left(-g_+'(Q - q) + k\right) \diff q
    \end{equation}
    and, similarly,
    \begin{equation}\label{eq:Niculescu2}
        k\left((Q'-Q^*)\wedge \qa\right) + g\left(Q' - (Q'-Q^*) \wedge \qa\right) - g(Q') = \int_0^{(Q'-Q^*)\wedge \qa}\left(-g_+'(Q' - q) + k\right) \diff q.
    \end{equation}
    By the definition of $Q^*$, the terms in the integrals above  are strictly positive. If $Q = Q^*$, \eqref{eq:different_gains} follows from the hypothesis $p<k$ and the fact that $g$ is affine between $Q$ and $Q'$. Assume now that $Q^* < Q$. For $q \in [0, (Q-Q^*)\wedge \qa]$, $-g'_+(Q-q) \leqslant -g'_+(Q'-q)$ since $g'_+$ is nonincreasing. Let $\epsilon > 0$ be such that $\epsilon < (Q-Q^*)\wedge \qa$ and $\epsilon < (Q'-Q)\wedge \qa$. We have
    \begin{equation*}
        \begin{split}
            \int_0^{(Q'-Q^*)\wedge \qa}\left(-g_+'(Q' - q) + k\right) \diff q
        = &\int_0^{\epsilon}\left(-g_+'(Q' - q) + k\right) \diff q
        + \int_{\epsilon}^{(Q-Q^*)\wedge \qa}\left(-g_+'(Q' - q) + k\right) \diff q\\
        &+ \int_{(Q-Q^*)\wedge \qa}^{(Q'-Q^*)\wedge \qa}\left(-g_+'(Q' - q) + k\right) \diff q.
        \end{split}
    \end{equation*}
    For $q \in (0,\epsilon]$, by the condition on $Q$, $g'_+(Q-q) > p = g'_+(Q'-q)$. This implies that the first integral is strictly greater that $\int_0^{\epsilon}\left(-g_+'(Q - q) + k\right) \diff q$. The second integral is greater than or equal to $\int_{\epsilon}^{(Q-Q^*)\wedge \qa}\left(-g_+'(Q - q) + k\right) \diff q$ since $g_+'$ is nonincreasing. The third one is nonnegative by the definition of $Q^*$. The result follows from \eqref{eq:Niculescu1} and \eqref{eq:Niculescu2}.
\end{proof}

\subsection{Some lemmas about the continuous concave envelope}

For a continuous function $g:\Q \to \mathbb{R}$, $\tilde{Q} \in \Q$ and $(Q, Q') \in [-\Qmax,\tilde{Q}] \times [\tilde{Q}, \Qmax]$, we define the quantity
\begin{equation*}
    A_{g, \tilde{Q}}(Q,Q') \defeq g\left(\tilde{Q}\right) - \frac{Q'-\tilde{Q}}{Q'-Q} g\left(Q\right) - \frac{\tilde{Q} - Q}{Q' - Q} g\left(Q'\right)
\end{equation*}
if $Q' > Q$ and $A_{g, \tilde{Q}}(Q,Q') = 0$ otherwise. It is clear that $A_{g,\tilde{Q}}$ is continuous and $\min A_{g, \tilde{Q}} \leqslant 0$. Recall also the definition of $C_g$ in \eqref{eq:defCg}.

\begin{lemma}
    \label{lem:envelope_shape}
    Let $g:\Q \to \mathbb{R}$ be a continuous function and $\tilde{Q} \in \Q$.\\
    (i) If $\min A_{g, \tilde{Q}} = 0$, then $\hat{g}\left(\tilde{Q}\right) = g \left(\tilde{Q}\right)$.\\
    (ii) If $\min A_{g, \tilde{Q}} < 0$ and $(Q, Q') \in [-\Qmax,\tilde{Q}] \times [\tilde{Q}, \Qmax]$ minimizes $A_{g, \tilde{Q}}$, then $\hat{g}(Q) = g(Q)$, $\hat{g}(Q') = g(Q')$, $C_{\hat{g}}\left(Q,Q', \frac{\tilde{Q}-Q}{Q'-Q}\right) = 0 = \min C_{\hat{g}}$ and
    \begin{equation*}
        \hat{g}\left(\tilde{Q}\right) - g\left(\tilde{Q}\right) = -g\left(\tilde{Q}\right) + \frac{Q'-\tilde{Q}}{Q'-Q}g(Q) + \frac{\tilde{Q} - Q}{Q'-Q}g(Q') = -C_{g}\left(Q,Q',\frac{\tilde{Q}-Q}{Q'-Q}\right).
    \end{equation*}
\end{lemma}
\begin{proof}
    (i) Suppose $\min A_{g, \tilde{Q}} = 0$. Let $\epsilon > 0$ and define the functions
    \begin{equation*}
        \begin{array}[b]{rccl}
            r:&\mathbb{R} & \to & \mathbb{R}\\
            & a & \mapsto & \min\limits_{R \in [\tilde{Q}, \Qmax]}
            \left(\epsilon + g(\tilde{Q}) + a (R-\tilde{Q}) - g(R) \right)
        \end{array}
    \end{equation*}
    and
    \begin{equation*}
        \begin{array}[b]{rccl}
            l:&\mathbb{R} & \to & \mathbb{R}\\
            & a & \mapsto & \min\limits_{R \in [-\Qmax, \tilde{Q}]}
            \left(\epsilon + g(\tilde{Q}) + a (R-\tilde{Q}) - g(R) \right)
        \end{array}.
    \end{equation*}
    It is sufficient to prove that there exists $a \in \mathbb{R}$ such that $r(a) \geqslant 0$ and $l(a) \geqslant 0$. Indeed, this would imply that the affine function $\phi:R \mapsto \epsilon + g(\tilde{Q}) + a (R-\tilde{Q})$ is greater or equal than $g$ and thus that $g(\tilde{Q}) \leqslant \hat{g}(\tilde{Q}) \leqslant \phi(\tilde{Q}) = g(\tilde{Q}) + \epsilon$.
    
We have that for $a \in \mathbb{R}$, $l(a) \leqslant 0$ implies $r(a) > 0$ and $r(a) \leqslant 0$ implies $l(a)> 0$, otherwise we would have the existence of $Q < \tilde{Q} < Q'$ such that $A_{g,\tilde{Q}}(Q,Q') \leqslant -\epsilon$, contradicting the main assumption. If $r(0) \geqslant 0$ and $l(0) \geqslant 0$, there is nothing to prove. 

Without loss of generality, suppose that $r(0)<0$ (and thus $l(0)>0$). By the continuity of $g$, there exists $\eta > 0$ such that $g < g(\tilde{Q}) + \epsilon$ on $[\tilde{Q},\tilde{Q}+\eta]$. Let $a > 0$, then
\begin{equation*}
    \min\limits_{R \in [\tilde{Q}, \tilde{Q} + \eta]}
            \left(\epsilon + g(\tilde{Q}) + a (R-\tilde{Q}) - g(R) \right) \geqslant 0
\end{equation*}
and
\begin{equation*}
    \min\limits_{R \in [\tilde{Q} + \eta, \Qmax]}
            \left(\epsilon + g(\tilde{Q}) + a (R-\tilde{Q}) - g(R) \right) \geqslant 
            \left(\epsilon + g(\tilde{Q}) + a \eta - \max g \right) \xrightarrow[a \to \infty]{}\infty.
\end{equation*}
Therefore, there exists $a$ such that $r(a)>0$. Since $r$ is continuous, there exists $a > 0$ such that $r(a) = 0$ (and therefore $l(a) > 0$) which is what we wanted to prove.

(ii) Suppose that $\min A_{g, \tilde{Q}} < 0$ and that $(Q, Q') \in [-\Qmax,\tilde{Q}] \times [\tilde{Q}, \Qmax]$ minimizes $A_{g, \tilde{Q}}$. Since $A_{g, \tilde{Q}}(Q,Q') < 0$, $Q < \tilde{Q} < Q'$. Define $\lambda \defeq \frac{\tilde{Q} - Q}{Q' - Q}$ and $\phi:R \in \Q \mapsto \frac{g(Q')-g(Q)}{Q'-Q}(R-Q) + g(Q)$.

Suppose that $\phi \geqslant g$. Then, $\hat{g} \geqslant \hat{g} \wedge \phi \geqslant g$ and, because $\hat{g} \wedge \phi$ is concave, $\hat{g} = \hat{g} \wedge \phi$. Furthermore, since $\phi(Q) = g(Q)$ and $\phi(Q') = g(Q')$, then $g(Q) = \hat{g}(Q)$ and $g(Q') = \hat{g}(Q')$. In addition,
\begin{equation*}
    0 \leqslant C_{\hat{g}}(Q,Q',\lambda) = \hat{g}(\tilde{Q}) - (1-\lambda)\phi(Q) - \lambda \phi(Q') \leqslant \phi(\tilde{Q})- (1-\lambda)\phi(Q) - \lambda \phi(Q') = 0,
\end{equation*}
the last equality coming from the fact that $\phi$ is affine. Hence, all the above inequalities are equalities and, in particular, $\hat{g}(\tilde{Q}) = \phi(\tilde{Q})$. Thus,
\begin{equation*}
    (\hat{g} - g)(\tilde{Q}) = \phi(\tilde{Q}) - g(\tilde{Q}) = \frac{g(Q')-g(Q)}{Q'-Q}(\tilde{Q}-Q) + g(Q) - g(\tilde{Q}) = -g(\tilde{Q}) + (1-\lambda)g(Q) + \lambda g(Q')
\end{equation*}
which is the desired result.

It remains to show that $\phi \geqslant g$. Let $R \in \Q \setminus \{Q,Q'\}$. Suppose that $R \leqslant \tilde{Q}$ (the other case is treated similarly). By the optimality of $(Q,Q')$, we have $A_{g, \tilde{Q}}(Q,Q')\leqslant A_{g, \tilde{Q}}(R,Q')$, which, by straightforward computations, is equivalent to $g(R)\leqslant\phi(R)$. 
\begin{comment}
\begin{align*}
    A_{g, \tilde{Q}}(Q,Q') &\leqslant A_{g, \tilde{Q}}(R,Q') \\
    g(\tilde{Q}) - \frac{Q'-\tilde{Q}}{Q'-Q}g(Q) - \frac{\tilde{Q} - Q}{Q'-Q}g(Q')
    &\leqslant
    g(\tilde{Q}) - \frac{Q'-\tilde{Q}}{Q'-R}g(R) - \frac{\tilde{Q} - R}{Q'-R}g(Q')\\
    \frac{Q' - \tilde{Q}}{Q'-R} g(R)
    &\leqslant
    \frac{Q'-\tilde{Q}}{Q'-Q}g(Q) + \frac{(\tilde{Q} - Q')(Q-R)}{(Q'-Q)(Q'-R)}g(Q')\\
    g(R) &\leqslant
    \frac{Q'-R}{Q'-Q}g(Q) + \frac{R-Q}{Q'-Q}g(Q')\\
    g(R) &\leqslant
    \frac{g(Q')-g(Q)}{Q'-Q}(R-Q) + \frac{Q(g(Q')-g(Q)) + Q'g(Q) -Qg(Q')}{Q'-Q}\\
    g(R) &\leqslant
    \frac{g(Q')-g(Q)}{Q'-Q}(R-Q) + g(Q) = \phi(R).
\end{align*}
\end{comment}
\end{proof}

\begin{corollary}
    \label{corol:envelope_shape}
    Let $g:\Q \to \mathbb{R}$ be a continuous function. Suppose that $(Q,Q', \lambda) \in \Q \times \Q \times [0,1]$ minimizes $C_g$. Then, $g(Q) = \hat{g}(Q)$, $g(Q') = \hat{g}(Q')$, $C_{\hat{g}}(Q,Q',\lambda)=0$ and
    $(\hat{g} - g)((1-\lambda)Q + \lambda Q') = \max (\hat{g} - g)$.
\end{corollary}

\begin{proof}
    Assume that $\min C_g < 0$.
    Define $\tilde{Q} \defeq (1-\lambda)Q + \lambda Q'$ and suppose, without loss of generality, that $Q < Q'$. We have that $(Q,Q')$ minimizes $A_{g,\tilde{Q}}$ and the results follow from Lemma \ref{lem:envelope_shape}.

    If $\min C_g = 0$, then $g$ is concave and, since $\hat{g} = g$, the result is immediate.
\end{proof}

\begin{lemma}
    \label{lem:envelope_inequality}
    Let $g:\Q \to \mathbb{R}$ be a continuous function and $k \in \mathbb{R}$. Suppose that $(Q,Q', \lambda) \in \Q \times \Q \times [0,1]$ minimizes $C_g$. Then,
    \begin{equation*}
        \begin{split}
            (i)\ h^a(g,k,(1-\lambda)Q + \lambda Q') &- (1-\lambda)h^a(g,k,Q) - \lambda h^a(g,k,Q') \\&\geqslant h^a(\hat{g},k,(1-\lambda)Q + \lambda Q') - (1-\lambda)h^a(\hat{g},k,Q) - \lambda h^a(\hat{g},k,Q')
        \end{split}
    \end{equation*}
    and
    \begin{equation*}
        \begin{split}
            (ii)\ h^b(g,k,(1-\lambda)Q + \lambda Q') &- (1-\lambda)h^b(g,k,Q) - \lambda h^b(g,k,Q') \\&\geqslant h^b(\hat{g},k,(1-\lambda)Q + \lambda Q') - (1-\lambda)h^b(\hat{g},k,Q) - \lambda h^b(\hat{g},k,Q').
        \end{split}
    \end{equation*}
\end{lemma}

\begin{proof}
    We only prove (i) since (ii) is similar. Let $q \in \Qa \cap [0, Q + \Qmax]$ and $z \in \Qa$. By Corollary \ref{corol:envelope_shape} $g(Q) = \hat{g}(Q)$. In addition, $g(Q - q \wedge z) \leqslant \hat{g}(Q - q \wedge z)$, hence
    \begin{equation*}
        \int_{\Qa}-e^{-k(q\wedge z) - g(Q - q \wedge z) + g(Q)}\mu^a(\diff z)
        \leqslant
        \int_{\Qa}-e^{-k(q\wedge z) - \hat{g}(Q - q \wedge z) + \hat{q}(Q)}\mu^a(\diff z) \leqslant h^a(\hat{g},k,Q).
    \end{equation*}
    Taking the supremum over $q$, we obtain $h^a(g,k,Q) \leqslant h^a(\hat{g},k,Q)$. Similarly, $h^a(g,k,Q') \leqslant h^a(\hat{g},k,Q')$.
    
    Let $q \in \Qa \cap [0, (1-\lambda)Q + \lambda Q' + \Qmax]$ and $z \in \Qa$. By Corollary \ref{corol:envelope_shape},
    \begin{equation*}
        (\hat{g} - g)((1-\lambda)Q + \lambda Q') \geqslant (\hat{g} - g)((1-\lambda)Q + \lambda Q' - q \wedge z)
    \end{equation*}
    therefore
    \begin{equation*}
        \hat{g}((1-\lambda)Q + \lambda Q') - \hat{g}((1-\lambda)Q + \lambda Q' - q \wedge z) \geqslant
        g((1-\lambda)Q + \lambda Q') - g((1-\lambda)Q + \lambda Q' - q \wedge z). 
    \end{equation*}
    Reasoning as before, we deduce that $h^a(\hat{g}, k, (1-\lambda)Q + \lambda Q') \leqslant h^a(g, k, (1-\lambda)Q + \lambda Q')$, which yields the conclusion.
\end{proof}

\end{dummyenv}

\printbibliography

\end{document}